\documentclass[letterpaper,11pt]{article}
\usepackage{fullpage}
\usepackage{amsmath,amsthm,amssymb,amsfonts,mathtools}
\usepackage[dvipdfmx]{graphicx}

\usepackage{color,url,xspace,ascmac}
\usepackage[normalem]{ulem}
\usepackage{comment,tikz}
\usepackage{thmtools,thm-restate}
\usepackage{enumitem}

\usepackage{float}
\usepackage{bbm}
\usepackage{relsize}

\usetikzlibrary{positioning}
\usetikzlibrary{calc}

\tikzset{every picture/.style={font issue=\footnotesize},font issue/.style={execute at begin picture={#1\selectfont}}}

\usepackage[unicode,colorlinks=true,citecolor=blue,linkcolor=blue,pdfusetitle]{hyperref}

\usepackage[style=alphabetic,natbib=true,natbib=true,maxnames=99,maxalphanames=99,isbn=false,url=false,backref=true,backend=biber,giveninits=true]{biblatex}

\AtEveryBibitem{%
\ifentrytype{article}{%
  \clearlist{publisher}
  \clearfield{editor}
}{}
\ifentrytype{book}{
}{}
\ifentrytype{inproceedings}{%
  \clearfield{volume}
  \clearfield{number}
  \clearfield{editor}
  \clearlist{publisher}
  \clearlist{location}
}{}
  \clearfield{day}
  \clearfield{month}
  \clearfield{isbn}
  \clearfield{url}
}

\usepackage{doi}
\usepackage[capitalise]{cleveref}
\newtheorem{theorem}{Theorem}[section]
\newtheorem{lemma}[theorem]{Lemma}
\newtheorem{corollary}[theorem]{Corollary}
\newtheorem{proposition}[theorem]{Proposition}

\newtheorem{remark}{Remark}
\newtheorem{definition}{Definition}[section]

\def\defeq{\mathrel{\mathop:}=}

\DeclarePairedDelimiter{\abs}{\lvert}{\rvert} % | | absolute value
\DeclarePairedDelimiter{\norm}{\lVert}{\rVert} % || || norm
\DeclarePairedDelimiter{\rbra}{\lparen}{\rparen} % () round brackets
\DeclarePairedDelimiter{\cbra}{\lbrace}{\rbrace} % {} curly brackets
\DeclarePairedDelimiter{\sbra}{\lbrack}{\rbrack} % [] square brackets
\newcommand{\indicator}{\mathbf{1}}

\DeclareMathOperator*{\E}{\mathbb{E}}
\DeclareMathOperator*{\Var}{\mathbf{Var}}
\newcommand{\condition}{\;\middle\vert\;}

\newcommand{\e}{\mathrm{e}}
\newcommand{\Nat}{\mathbb{N}}
\newcommand{\Real}{\mathbb{R}}
\newcommand{\Int}{\mathbb{Z}}

\DeclareMathOperator*{\argmax}{argmax}

\newcommand{\ThreeMajority}{\textsc{3-Majority}\xspace}
\newcommand{\Voter}{\textsc{Voter}\xspace}
\newcommand{\TwoChoices}{\textsc{2-Choices}\xspace}
\newcommand{\Majority}{\textsc{-Majority}\xspace}
\newcommand{\alphamax}[1]{\|\alpha_{#1}\|_{\infty}}
\newcommand{\tauweak}{\tau^{\mathrm{weak}}}

\newcommand{\taucons}{\tau_{\mathrm{cons}}}
\newcommand{\istar}[1]{I_{#1}}
\newcommand{\ThreeMajorityProcess}{\sigma}
\newcommand{\alphanorm}{\gamma}

\newcommand{\alpharatio}{R}

\crefname{figure}{Figure}{Figures}
\crefname{remark}{Remark}{Remarks}
\crefname{equation}{}{}

\renewcommand{\c}{\mathbf{c}}
\newcommand{\ct}{\widetilde{\mathbf{c}}}
\newcommand{\tc}{\widetilde{c}}
\newcommand{\Pull}{\mathsf{Voter}}
\newcommand{\ThreeMaj}{\mathsf{3Maj}}

\def\red{\color{red} \bf}

\makeatother

\title{Asynchronous 3-Majority Dynamics with Many Opinions}
\author{
Colin Cooper\\
\small{King’s College London}\\
\small{\texttt{\href{mailto:colin.cooper@kcl.ac.uk}{colin.cooper@kcl.ac.uk}}}
\and
Frederik Mallmann-Trenn\\
\small{King’s College London}\\
\small{\texttt{\href{mailto:frederik.mallmann-trenn@kcl.ac.uk}{frederik.mallmann-trenn@kcl.ac.uk}}}
\and
Tomasz Radzik\\
\small{King’s College London}\\
\small{\texttt{\href{tomasz.radzik@kcl.ac.uk}{tomasz.radzik@kcl.ac.uk}}}
\and
Nobutaka Shimizu\\
\small{Institute of Science Tokyo}\\
\small{\texttt{\href{shimizu.n.ah@m.titech.ac.jp}{shimizu.n.ah@m.titech.ac.jp}}}
\and
Takeharu Shiraga\\
\small{Chuo University}\\
\small{\texttt{\href{shiraga.076@g.chuo-u.ac.jp}{shiraga.076@g.chuo-u.ac.jp}}}
}

\date{}
\begin{document}
\maketitle
\begin{abstract}
  We consider 3-Majority, a probabilistic consensus dynamics on a complete graph with $n$ vertices, each vertex starting with one of $k$ initial opinions.
	At each discrete time step, a vertex $u$ is chosen uniformly at random.
	The selected vertex $u$ chooses three neighbors $v_1,v_2,v_3$ uniformly at random with replacement
  and takes the majority opinion held by the three, where ties are broken in favor of the opinion of $v_3$.
  The main quantity of interest is the consensus time, the number of steps required for all vertices to hold the same opinion.
  This asynchronous version turns out to be considerably harder to analyze than the synchronous version and so far results have only been obtained for $k=2$. Even in the synchronous version the results for large $k$ are far from  tight.
  In this paper we prove that  the consensus time is $\tilde{\Theta}( \min(nk,n^{1.5}) )$ for all $k$. 
  These are the first bounds for all $k$ that are tight up to a polylogarithmic factor. 
\end{abstract}

\tableofcontents

%\fnote{I think it would be useful to link the table of consents with the theorems (I'm struggling to find the main theorem for example)}
\section{Introduction}
\emph{Consensus dynamics} is a time-evolving  distributed process on a network consisting of \emph{anonymous} nodes (i.e., they do not possess their own ID) which hold their own value called \emph{opinion}.
    
    The nodes update their own opinions according to a common protocol and the aim of the protocol is to reach  \emph{consensus}, a configuration where all vertices hold the same opinion.
    We usually consider protocols specified by a simple algorithm that consists of elementary local calculations (cf.\ light-weight algorithms \cite{HP01}).
    This algorithmic perspective describes real-world dynamics as the outcome of the interaction of simple local algorithms and appears in a broad range of areas including interacting particle systems \cite{Lig85}, distributed computing \cite{HP01}, chemical reaction networks \cite{CRN_SODA14,Condon2020}, biological systems \cite{biology}, to name a few.
    See \cite{consensus_dynamics_SIGACT} for an excellent survey of this general topic.

    To state it more formally, consider a graph $G=(V,E)$, where each vertex $v\in V$ initially holds a value $\sigma_0(v) \in \Sigma$.
    This vector $\sigma_0 = (\sigma_0(v))_{v \in V}$ is called a \emph{configuration}.
    During the dynamics, a subset of vertices is chosen to be activated, and the activated vertices update their opinions according to a common update rule.
    This paper focuses on discrete-time dynamics.
    A consensus dynamics is \emph{synchronous} (or parallel) if all vertices become active at every discrete time step and \emph{asynchronous}\footnote{ The usage of the term “asynchronous" here is standard in the literature of consensus dynamics (e.g., \cite{consensus_dynamics_SIGACT}) but different from some literature of distributed computing, where the word refers to unbounded/unpredictable delays.} (or sequential) if exactly one uniformly random vertex is chosen to be activated.
    This yields a sequence of configurations $(\sigma_t)_{t\ge 0}$, and we are interested in the convergence time of $\sigma_t$.

     \Voter is one of the simplest protocols studied in distributed computing \cite{HP01} and physics \cite{Lig85}.
    In this protocol, an activated vertex selects a uniformly random neighbor $v$ and updates its opinion to match the opinion of $v$.
%    The expected convergence time of a synchronous \Voter on an expander graph is $\Theta(|V|)$ \cite{CEO+13}.
    The \Voter exhibits \emph{linearity} in the sense that the local update algorithm can be written as a (random) linear function of the current configuration vector.
    That is, if $\sigma_t \in \Sigma^V$ is the current configuration, then the next configuration $\sigma_{t+1}\in\Sigma^V$ can be written as $\sigma_{t+1} = P_t \sigma_t$, where $P_t$ is a random matrix chosen from an appropriate distribution \cite{linear_voting}.
    %{\red (reference or explain what it means?)}. %TODO
    Linear dynamics contains several important dynamics, including the push voting model \cite{DW83}, Moran process \cite{Moran_process14}, and averaging dynamics \cite{randomized_gossip_algorithms_Boyd}.
    This class of dynamics is widely recognized as a tractable class, as we can invoke martingale theory \cite{linear_voting}.
    In particular, \Voter exhibits an elegant duality relationship with coalescing random walks, which enables us to bound the convergence time (e.g., \cite{CEO+13}).

    While linear dynamics has been extensively investigated, a more recent line of research addresses dynamics based on nonlinear operations.
    Such nonlinear dynamics has a wide range of applications in the consensus problem \cite{Becchetti2015}, community detection \cite{CNS19,Shimizu_Shiraga_SBM}, chemical reaction network \cite{undecided_chemical}, among others.
    One prominent example is the \ThreeMajority dynamics~\cite{simple_dynamics,BCNPT17,ignore_or_comply,nearly_tight_analysis,hierarchy_Berenbrink}. 
    In this dynamics, an activated vertex $u$ chooses three neighbors $v_1,v_2,v_3$ uniformly at random (with replacement).
    If $v_1$ and $v_2$ have the same opinion,
    then $u$ replaces its opinion with the opinion of $v_1$ and $v_2$.
    Otherwise, $u$ updates its opinion to the opinion held by $v_3$.%
    \footnote{Observe that this is equivalent to $u$ taking on the majority opinion 
    held by the three vertices $v_1$, $v_2$ and $v_3$, or, if they all have different opinions, the opinion of $v_i$ for random $ i\sim\{1,2,3\} $.}
    Such majority-like update rules appear naturally  as  models of opinion forming.
    For example, suppose that you visit an unfamiliar city for a conference and plan to have lunch in a restaurant.
    To decide which restaurant to go to, you can ask some of your friends living in the city and then take the majority among the recommendations, where ties are broken at random.

% nonlinear dynamics
\subsection{Main Result} \label{sec:main result}
    We study the convergence time of the asynchronous version of  \ThreeMajority on an $n$-vertex complete graph (with self-loops).
    See \cref{def:3Majority} for the formal definition.
    The quantity of interest is the \emph{consensus time}, which is the number of time steps required to reach consensus.

    In the case of binary opinions ($k=2$), \citet{hierarchy_Berenbrink} proved that the consensus time is $O(n\log n)$ with high probability.%
    \footnote{Throughout this paper, ‘‘with high probability'' means with probability $1 - O(n^{-\beta})$, for 
    some constant $\beta>0$.}
    Although the case of many opinions ($k>2$) is much less understood in the asynchronous model,
    \ThreeMajority dynamics in the synchronous model have been well studied.
    \citet{simple_dynamics} proved that the dynamics reaches consensus within $O\left( k\log n \right)$ rounds with high probability if there is initially a sufficiently large gap in the number of supporters between the most and second most popular opinions.
    They also proved that the dynamics requires $\Omega(k\log n)$ rounds to reach consensus with high probability for some initial configuration if $k\le (n/\log n)^{1/4}$.
    The condition of a large gap in the initial configuration was removed by \citet{stabilizing_consensus}, where they proved that the consensus time is $\tilde O\left(k^3 \right)$ with high probability for any initial opinion configuration if $k\le n^{1/3-\epsilon}$ and $\epsilon>0$ is an arbitrary constant.\footnote{$\tilde O\left( \cdot\right)$ hides a $(\log n)^{O(1)}$-factor.}
    They asked for  an improved bound on the consensus time, which was subsequently given by
    \citet{nearly_tight_analysis}.
    They proved that the consensus time is $O\left( k\log n\right)$ with high probability if $k = O(n^{1/3}/\sqrt{\log n})$.
    The question is: Does this \emph{linear-in-$k$} dependence \cite{consensus_dynamics_SIGACT} hold for all $k$?

    Interestingly, \citet{ignore_or_comply} proved the following remarkable result: Starting from any initial configuration and any $k$ (say, $k=n$), 
        the number of remaining opinions decreases to at most $\kappa$ after $O\left( n\log n/\kappa\right)$ rounds with high probability.
Combining this result  with the bound from \cite{nearly_tight_analysis} for $\kappa=O(n^{1/3}/\sqrt{\log n})$,
        we have that the consensus time of synchronous \ThreeMajority dynamics is at most $\tilde O\left( n^{2/3}\right)$ with high probability.
        
    In summary, the linear-in-$k$ dependence of the consensus time 
    holds for $k\ll n^{1/3}$ \cite{nearly_tight_analysis} but does not hold for $k\gg n^{2/3}$ \cite{ignore_or_comply}.
    This raises the following natural question:
        \emph{What is the range of $k$ in which the linear-in-$k$ dependence holds?}
    This paper determines the range of such $k$ in the \emph{asynchronous} setting.
    More generally, we obtain the first consensus time bound that is tight for \emph{all} $2\le k \le n$ up to $(\log n)^{O(1)}$ factors.
    %
%    This is the first result for the consensus time of \ThreeMajority in the asynchronous model for $k>2$.
%

    \begin{theorem}[Main Theorem] \label{thm:all k}
        For any $k\ge 2$, the asynchronous \ThreeMajority dynamics with $k$ opinions on an $n$-vertex complete graph reaches consensus within $\tilde{O}\rbra*{\min\rbra*{kn, n^{1.5} }}$ steps with high probability.
        Moreover, there exists an initial configuration on which the dynamics requires $\tilde{\Omega}\rbra*{\min\rbra*{kn, n^{1.5}}}$ steps to reach consensus with high probability.
    \end{theorem}
    More precisely, the hidden $(\log n)^{O(1)}$ factors in \cref{thm:all k} are as follows: The upper bound is $O\rbra*{ \min\cbra*{ kn(\log n)^2, n^{1.5}(\log n)^{1.5} }}$ and the lower bound is $\Omega\rbra*{ \min\rbra*{ kn, \frac{n^{1.5}}{\sqrt{\log n}}} }$.

    \Cref{thm:all k} means that the linear-in-$k$ dependence holds for $2\le k \ll \sqrt{n}$ and not for $k \gg \sqrt{n}$.
    The main ingredient of the proof of \cref{thm:all k} is to bound the consensus time for $k \ll n^{1/2}$.
    This extends the previous range of $k\ll n^{1/3}$ in the synchronous setting \cite{nearly_tight_analysis}.
    
    We also obtain a lower bound on the consensus time for the asynchronous \TwoChoices dynamics, which is a fundamental nonlinear consensus dynamics.
    In the \TwoChoices dynamics, 
the activated vertex $u$ randomly picks  two vertices $v_1$ and $v_2$ (with replacement) and updates its opinion to match that of $v_1$ if both $v_1$ and $v_2$ have the same opinion (otherwise, the opinion of $u$ does not change).
    Specifically, we show that  asynchronous \TwoChoices requires $\Omega(kn)$ steps to reach consensus with high probability provided that $k \ll n/\log n$ (see \cref{lem:consensus time lower bound_Bo2} for the detail).
    This establishes a polynomial gap between  \ThreeMajority and \TwoChoices 
    in the asynchronous setting:
    The asynchronous \ThreeMajority reaches consensus within $\tilde{O}(n^{1.5})$ steps, whereas \TwoChoices requires $\tilde{\Omega}(n^2)$ steps for $k\approx n/\log n$.
    Such gap was previously known in the synchronous setting \cite{ignore_or_comply,nearly_tight_analysis}.

    This work gives an improved analysis of \ThreeMajority dynamics in the asynchronous update setting.
    We leave open whether the analogous consensus time bound holds for the synchronous \ThreeMajority dynamics.
    A general reduction between synchronous and asynchronous models is not known.
    Indeed some dynamics are known to exhibit a huge synchronous-asynchronous gap \cite{Becchetti_minority}.
    However, we believe that there is no significant gap between asynchronous and synchronous \ThreeMajority dynamics and our proof has the potential to be adapted to prove optimal consensus time bound for the synchronous \ThreeMajority dynamics.

    The overview of the proof of \cref{thm:all k} and our technical contributions are given in \cref{sec:proof overview}.

\begin{comment}
    \begin{figure} %TODO fix
    \centering
\begin{tikzpicture}
 \draw[->] (0,0)--(8,0) node[below]{$k$};
 \draw[->] (0,0)--(0,8) node[left]{$\taucons$};
 \draw (0.25,6.2) rectangle (3,7.25);
 \draw[line width=1.5pt,blue,opacity=0.2] (0.4,7)--(0.7,7);
 \draw[line width=1pt,blue] (0.4,6.97)--(0.7,6.97);
 \draw (0.8,7) node[right]{\TwoChoices};
 \draw[line width=2pt,red] (0.4,6.5)--(0.7,6.5);
 \draw (0.8,6.5) node[right]{\ThreeMajority};
 \draw[line width=2pt,red] (0,0.03)--(3,3.03) -- (7.5,3.03);
 \draw[line width=1pt,blue] (0,0) -- (7.5,7.5);
 \draw[line width=1.5pt,blue,opacity=0.2] (0,0.05) -- (7.45,7.5);
 \draw (4,2.7) node[right]{$\tilde{\Theta}(\min\{kn,n^{1.5}\})$};
 \draw (4,2.2) node[right]{(\cref{thm:all k})};
 \draw (2,5) node[right]{$\Omega(kn)$};
 \draw (2,4.5) node[right]{(\cref{lem:consensus time lower bound_Bo2})};
 \draw[dotted] (3,3)--(0,3) node[left]{$n^{1.5}$};
 \draw[dotted] (3,3)--(3,0) node[below]{$\sqrt{n}$};
 \draw[dotted] (7.5,7.5)--(7.5,0) node[below]{$n$};
 \draw[dotted] (7.5,7.5)--(0,7.5) node[left]{$n^2$};
 \draw[below] (0,0) node{$2$};
 \draw[left] (0,0) node{$n$};
\end{tikzpicture}
    \caption{Summary of our results.
    Red and blue lines indicate the consensus time bounds for \ThreeMajority and \TwoChoices, respectively (we ignore $
(\log n)^{O(1)}$ factors). \red{todo: fix this. F: I'd suggest to just remove all float text inside the figure and just put a $\tilde\Theta(kn)$ for 2 choices and $\tilde\Theta(min(...)$ for 3Maj }
    \label{fig:results} }
\end{figure}
\end{comment}

\subsection{Related Work}
The convergence time in consensus dynamics, 
where each vertex repeatedly updates its opinion based on the opinions of others with the goal of reaching consensus, 
has been studied for various opinion updating rules, including \ThreeMajority. 
This research topic spans multiple contexts, including distributed computing, physics, social science, and biology. 
See, for example, \cite{consensus_dynamics_SIGACT} for an overview of recent studies in this field.

\Voter (also known as the pull voting) is one of the simplest and most well-known consensus dynamics~\cite{Pel02,voter_dynamic_graph,linear_voting,NIY99}.
It has been shown that, for any $k$, the expected consensus time is $O(n)$ for synchronous \Voter~\cite{voter_dynamic_graph} and $O(n^2)$ for the asynchronous \Voter \cite{linear_voting}.

\TwoChoices (also known as two-sample voting or best-of-two)
has also been extensively studied in the synchronous setting \cite{Doerr11,nearly_tight_analysis,ignore_or_comply,twochoice_ICALP14,twochoice_expander_DISC15,twochoice_expander_DISC17,quasi-majority,EFKMT17}. 
\citet{nearly_tight_analysis} showed that the consensus time is $O(k \log n)$ with high probability if $k=O(\sqrt{n/\log n})$. 
On the other hand, \citet{ignore_or_comply} showed that the consensus time is $\Omega(n/\log n)$ with high probability if each vertex initially holds a different opinion ($k=n$). 
% synchronous \ThreeMajority is $O(n^{2/3}\log^{1.5}n)$ for any initial configuration~\cite{nearly_tight_analysis}.

%
General classes of opinion updating rules have been studied especially for the binary opinions case ($k=2$)~\cite{linear_voting,consensus_ER_Schoenebeck18,quasi-majority,hierarchy_Berenbrink}. 
\citet{hierarchy_Berenbrink} studied $j$\Majority dynamics with binary opinions, 
where each vertex $v$ randomly picks $j$ vertices $u_1,\dots,u_j$ (with replacement) and $v$ updates its opinion to match the majority opinion among $u_1,\dots,u_j$ (with ties broken randomly). 
They showed that the consensus time of  $j$\Majority stochastically dominates the consensus time of $(j+1)$\Majority for any $j\geq 2$.
\citet{consensus_ER_Schoenebeck18} studied asynchronous “majority-like” voting, including $(2j+1)$\Majority for $j\geq 1$.
They showed that the expected consensus time is $O(n \log n)$ for a dense Erd\H{o}s-R{\'e}nyi random graph.

For the case of synchronous updates, several studies focus on scenarios where an adversary capable of changing the opinions of vertices exists~\cite{stabilizing_consensus,simple_dynamics,nearly_tight_analysis,ignore_or_comply}. 
For example, \citet{nearly_tight_analysis} showed that, even with an adversary that can change the opinions of at most $O(\sqrt{n}/k^{1.5})$ vertices, 
the convergence time remains $O(k \log n)$ for both \ThreeMajority with $k=O(n^{1/3}/\sqrt{\log n})$ and \TwoChoices with $k=O(\sqrt{n/\log n})$.

Undecided State Dynamics is a well-known model within the field of consensus dynamics, which uses a special state called ``Undecided''. 
In this model, each ``Undecided'' vertex updates its state to match that of an interacted vertex, and each ``not-Undecided'' vertex becomes ``Undecided'' upon interacting with a vertex of a differing state.
The consensus time for Undecided State Dynamics has been well studied for $k=2$~\cite{Angluin2007,undecided_MFCS} and, 
more recently, for the general case of $k>2$ \cite{Becchetti2015,fast_convergence_undecided}
in both synchronous and asynchronous model.

It is worth mentioning that the asynchronous model 
is essentially equivalent to the continuous-time model
where each vertex has an independent Poisson clock with rate $1$.
For details see~\cite{randomized_gossip_algorithms_Boyd}.
\section{Proof Overview and Technical Contribution} \label{sec:proof overview}
In this section, we present the proof overview of \cref{thm:all k} and our key ingredients.
\subsection{Key Results}
The proof of \cref{thm:all k} consists of two parts.
Firstly, we show that, starting from any initial configuration, the number of remaining opinions in the asynchronous \ThreeMajority quickly decreases.
\begin{theorem} \label{thm:many opinions}
    Consider the asynchronous \ThreeMajority starting from any initial configuration.
    There exists a universal constant $C>0$ such that,
    for any $1\le \kappa \le n$, the number of remaining opinions becomes at most $\kappa$ within $Cn^2\log n/\kappa$ steps with high probability.
\end{theorem}
In the synchronous setting, results analogous to \cref{thm:many opinions} were obtained in \cite[Section 3.2]{ignore_or_comply}.
However, their proof crucially relies on the synchronicity.
Indeed, extending their approach to the asynchronous setting is “the most interesting open question" of \cite{hierarchy_Berenbrink} and we resolve it for \ThreeMajority.

Secondly, we show that, the linear-in-$k$ dependence of the consensus time holds for $k\ll \sqrt{n}$.

\begin{theorem} \label{thm:small k}
    Consider the asynchronous 3-Majority dynamics with $k \le c\sqrt{n/\log n}$ opinions on an $n$-vertex complete graph, where $c>0$ is a sufficiently small constant. %\marginpar{\textcolor{red}{\framebox{no $c$?}}}
    Then, for any initial configuration, the consensus time is $O(kn\log n)$ with high probability.
    Moreover, if the dynamics starts from the balanced initial configuration (that is, the configuration in which all opinions are supported by the same number of vertices), then the consensus time is $\Omega(kn)$ with high probability.
\end{theorem}

The upper bound of \cref{thm:small k} addresses a wider range of $k$ compared to \cite{nearly_tight_analysis}, who proved that the analogous upper bound in the synchronous setting holds for $k\ll n^{1/3}$.
Indeed, the common argument in previous works for the synchronous model essentially requires $k\ll n^{1/3}$ (see \cref{sec:small opinion case} for the detail).
We overcome the “$n^{1/3}$-barrier" by invoking sharp martingale concentration inequalities.
The lower bound of \cref{thm:small k} improves the analysis of \cite{simple_dynamics} who proved the analogous lower bound in the synchronous setting for $k\le (n/\log n)^{1/4}$.

The proof of \cref{thm:all k} is straightforward from \cref{thm:many opinions,thm:small k}.
\begin{proof}[Proof of \cref{thm:all k} using \cref{thm:many opinions,thm:small k}.]
    We first prove the upper bound.
    Let $k$ be the number of opinions in the initial configuration and
    $\kappa = c\sqrt{n/\log n}$ (for a sufficiently small constant $c>0$).
    If $k \le \kappa$, then from \cref{thm:small k}, the consensus time is $O(k n \log n)$.
    If $k > \kappa$, from \cref{thm:many opinions}, the number of remaining opinions becomes at most $\kappa$ within $O((n\log n)^{1.5})$ steps.
    Then from \cref{thm:small k}, the dynamics reaches consensus within additional $O(\kappa n \log n) = O(n^{1.5}\sqrt{\log n})$ steps with high probability.
    Therefore, for any $k\ge 2$, the consensus time is $O\rbra*{\min\rbra*{ kn(\log n)^2, (n\log n)^{1.5} }}$ with high probability.

    We prove the lower bound.
    If $k \le \kappa= c\sqrt{n/\log n}$ and the dynamics starts with the balanced configuration, then the consensus time is $\Omega(kn)$ from \cref{thm:small k}.
    Even if $k > \kappa$, for the balanced initial configuration with $\kappa$ opinions (i.e., $\kappa$ opinions supported by the same number of $n/\kappa$ vertices and the remaining opinions not supported by any vertex),
    the consensus time is $\Omega(\kappa n) = \Omega(n^{1.5}/\sqrt{\log n})$ with high probability from \cref{thm:small k}.
    This proves the claim.
\end{proof}

The rest of this paper is devoted to proving \cref{thm:many opinions,thm:small k}.

\subsection{Many Opinion Case} \label{sec:many opinion case}
To prove \cref{thm:many opinions}, we borrow the framework of comparing \Voter and \ThreeMajority from \cite{ignore_or_comply}.
Specifically, we couple \Voter and \ThreeMajority so that the number of remaining opinions (i.e., an opinion that is supported by at least one vertex) in \ThreeMajority is no more than that of \Voter.
\begin{lemma} \label{lem:coupling}
    Consider the asynchronous \Voter and \ThreeMajority dynamics starting from the same initial configuration.
    There exists an explicit coupling of these dynamics such that, the number of remaining opinions in the \ThreeMajority dynamics is no more than that of \Voter.
\end{lemma}

An analogous result of \cref{lem:coupling} in the synchronous setting was shown by \cite{ignore_or_comply}.
However, their proof crucially relies on the synchronicity of the update schedule.
In the asynchronous setting, \cref{lem:coupling} was shown by \cite{hierarchy_Berenbrink} for the special case of $k=2$.
Indeed, they left the case of general $k\ge 3$ as an open question and \cref{lem:coupling} resolves it.
Moreover, a remarkable point is that the proof of \cref{lem:coupling} is \emph{constructive} in the sense that the coupling is explicit.
Actually, the coupling of both \cite{ignore_or_comply,hierarchy_Berenbrink} rely on a tool from Majorization theory (Strassen's theorem) and are nonconstructive.
Unfortunately, our analysis essentially relies on the asynchronicity.
We leave the constructive proof of \cref{lem:coupling} in the synchronous setting as an open question.

We then prove an analogous result of \cref{thm:many opinions} for the asynchronous \Voter.
This can be done by exploiting the well-known duality between \Voter and coalescing random walks \cite[Section 14.3]{AF02}.

\paragraph*{Explicit Coupling of \texorpdfstring{\Voter}{Voter} and \texorpdfstring{\ThreeMajority}{3-Majority} (\cref{sec:many opinions}).}
The proof strategy for \cref{lem:coupling} is to construct a coupling that preserves the \emph{majorization order}, following the previous approach of \cite{ignore_or_comply} for the synchronous model.
Let $[k]=\{1,\dots,k\}$ denote the set of possible opinions.
For each opinion $i \in [k]$, let $c_i$ be the number of vertices that hold opinion $i$ in the current configuration.
Since the underlying graph is a complete graph,
    we treat $\c=(c_1,\dots,c_k)$ as a configuration;
    thus, \ThreeMajority (and \Voter) is a sequence of random configurations $(\c_t)_{t\ge 0}$.
We always rearrange the elements of a configuration $\c=(c_1,\dots,c_k)$ in descending order, i.e., $c_1 \ge \dots \ge c_k$.
For two configurations $\c,\ct$, we say $\c$ \emph{majorizes} $\ct$, denoted by $\c \succeq \ct$, if for all $i\in[k]$, $\sum_{j=1}^i c_j \ge \sum_{j=1}^i \tc_j$.
For example, $(4,\,3,\,2,\,1) \succeq (3,\,3,\,2,\,2)$, $(5,\,5,\,0,\,0) \succeq (4,\,4,\,1,\,1)$ but $(4,\,3,\,2,\,1) \not \succeq (5,\,3,\,2,\,0)$.
Note that, if $\c \succeq \ct$ then the number of remaining opinions in $\c$ is no more than that of $\ct$.
We construct a coupling of “one-step simulations" of $\Voter$ and $\ThreeMajority$ that preserves the majorization relation.

To state it more formally, let $\Pull(\c)$ ($\ThreeMaj(\c)$) be the random configuration after the one-step update of \Voter (resp.\ \ThreeMajority) on configuration $\c$ (here we also keep the descending order by rearranging).
We show that, for any two configurations $\c,\ct$ such that $\c \succeq \ct$, there exists a coupling of $\ThreeMaj(\c)$ and $\Pull(\ct)$ such that $\ThreeMaj(\c) \succeq \Pull(\ct)$ holds.
From this claim, we can immediately prove \cref{lem:coupling} by repeatedly simulating the coupling.
To prove the claim,
we construct (i) a coupling of $\Pull(\c)$ and $\ThreeMaj(\c)$ for any $\c$ so that $\ThreeMaj(\c) \succeq \Pull(\c)$,
and (ii) a coupling of $\Pull(\c)$ and $\Pull(\ct)$ for any $\c \succeq \ct$ so that $\Pull(\c) \succeq \Pull(\ct)$.
The claim immediately follows from these two couplings as $\ThreeMaj(\c) \succeq \Pull(\c) \succeq \Pull(\ct)$ for any $\c$.

The most technical part is the coupling of (ii).
The key idea is to decompose the process of generating $\Pull(\c)$ into \emph{deletion process} and \emph{addition process}.
Specifically, in $\Pull(\c)$,
a uniformly random vertex $u$ is deleted, and then a new vertex, which shares the same opinion as another uniformly random vertex $v$, is added.
We show that these deletion and addition processes preserve the majorization relation.
To this end, we decompose the configuration vector $\c=(c_1,\dots,c_k)$ and $\ct=(\tc_1,\dots,\tc_k)$ into smaller components called the \emph{block structure}
and then construct a coupling for each block.
Since both deletion and addition processes change exactly one vertex, it suffices to apply the coupling for one of the blocks.
See \cref{sec:many opinions} for details.

\subsection{Small Opinion Case} \label{sec:small opinion case}
In the proof of \cref{thm:small k},
we consider two phases.

\begin{itemize}
\item \emph{Phase I: gap emergence.} We show that, starting from any initial configuration with $k \ll \sqrt{n/\log n}$ opinions, a significant gap between the most and second most popular opinions emerges within $O(kn\log n)$ steps (\cref{lem:unique strong opinion lemma}).
\item \emph{Phase II: gap amplification.} Once a significant gap arises between the most and second most popular opinions, we show that the dynamics reaches consensus within $O(kn\log n)$ steps (\cref{lem:AfterOneStrong}).
\end{itemize}

The analysis of the gap-emerging phase is the most technical part of this paper
and, in fact, the most challenging part of previous work \cite{stabilizing_consensus,nearly_tight_analysis}.
In what follows, we explain why previous techniques for synchronous setting require the condition $k \ll n^{1/3}$ and how to overcome this barrier.

\paragraph*{The \texorpdfstring{$n^{1/3}$}{n^{1/3}}-barrier in the synchronous model.}
    Before describing the ingredients of our analysis,
        we explain why previous works require the condition that $k\ll n^{1/3}$.
    Consider the synchronous \ThreeMajority dynamics. 
    Let $A_t(i)$ denote the number of vertices that hold opinion $i$ at the beginning of the $t$-th round, and let $\alpha_t(i)=\frac{A_t(i)}{n}$.
    We treat $\alpha_t=( \alpha_t(i) )_{i\in[k]}$ as a $k$-dimensional vector in $[0,1]^k$.
    For any $t\in\Nat$, it is not hard to see that 
    \begin{align}
    \E_{t-1}[\alpha_t(i)]=\alpha_{t-1}(i) \left( 1 + \alpha_{t-1}(i) - \norm{\alpha_{t-1}}^2\right) \label{eq:alpha expectation synchronous}
    \end{align}
    where
    $\E_{t-1}[\cdot]$ denotes the expectation conditioned on the opinion configuration at round $t-1$
    and
    $\norm{\cdot}$ denotes the $\ell^2$-norm  (cf.\ \cite[Lemma 2.1]{simple_dynamics}).

    Suppose we have a configuration which is still fairly balanced but
    the differences in the support of various opinions start emerging.    
    More precisely, we start with a configuration
      so that $\alpha_0(j)=\Theta(1/k)$ for all $j\in[k]$
      and some opinion $i\in[k]$ satisfies $\alpha_0(i)=\norm{\alpha_0}^2 + \Theta(1/k)$.
    Fix such an opinion $i$.
    Then, \cref{eq:alpha expectation synchronous}, we have
        $\E_0[\alpha_1(i)]=\alpha_0(i)\left( 1+\Theta \left( 1/k \right)\right) = \alpha_0(i) + \Theta \left( 1/k^2 \right)$.
    In other words, $\alpha_1(i)$ has an additive drift of $\Theta \left( 1/k^2 \right)$ in expectation.
    In each round, the $n$ vertices simultaneously update their opinions using their coin flips.
    Therefore, the random variable $\alpha_1(i)$ can be written as the sum of $n$ independent random variables.
    By the central limit theorem, we have $\alpha_1(i)\approx \E[\alpha_1(i)] \pm O\left( \sqrt{\frac{\alpha_0(i)}{n}} \right) \approx \alpha_0(i) + \Theta \left( \frac{1}{k^2} \right) \pm O \left( \sqrt{\frac{1}{kn}} \right)$.
    To mitigate the influence of the variance term in the one-step transition due to the central limit theorem, it is necessary to satisfy $\frac{1}{k^2} \gg \sqrt{\frac{1}{kn}}$, or equivalently, $k\ll n^{1/3}$.

    Tracking one-step change of $\alpha_t(i)$ is essential in the proof of previous works \cite{stabilizing_consensus,simple_dynamics,nearly_tight_analysis}.
    In \cite{nearly_tight_analysis}, the authors considered three classes of opinions:
        An opinion is strong at round $t$ if $\alpha_t(i)\ge \frac{ \left\| \alpha_t \right\|_\infty }{ 5 }$,
        weak if $\alpha_t(i) < \frac{ \left\| \alpha_t \right\|_\infty }{ 5 }$,
        and super-weak if $\alpha_t(i) \le \frac{1}{10k}$.
    Here, $\norm{\cdot}_\infty$ denotes the $\ell^{\infty}$-norm.
    Their proof crucially relies on the
    “weak remain weak" property: Once an opinion becomes weak (super-weak), then the opinion remains weak (resp.\ super-weak) with high probability in the next round.
    Since this property refers to the concentration of the one-step change of $\alpha_t(i)$, for $k \gg n^{1/3}$, their analysis cannot move beyond “fairly balanced" configurations.

\subsubsection{Our Tool: Drift Analysis via the Freedman Inequality} \label{sec:out tool}

\paragraph*{Drift Analysis.}
Unlike synchronous models, we cannot ensure the one-step concentration of $\alpha_t(i)$ as in the asynchronous model since only one vertex updates its opinion.
However, we can prove a multi-step concentration of $\alpha_t(i)$ by exploiting martingale concentration.
Such an argument is known as \emph{drift analysis} and is a key component in previous works on consensus dynamics in the asynchronous or population protocol model of relevant dynamics \cite{hierarchy_Berenbrink,fast_convergence_undecided}.

A prototype argument of drift analysis \cite{drift_analysis_concentration_Kotzing} goes as follows:
Consider a sequence of nonnegative random variables $(X_t)_{t\ge 0}$ that satisfy $\E_{t-1}[X_t]\ge X_{t-1}+\varepsilon$ (i.e., $X_t$ has an additive drift).
Then $Y_t:=X_t-\varepsilon t$ forms a submartingale and we can apply the Azuma--Hoeffding inequality to ensure the linear growth of $Y_t$.
Similarly, if $(X'_t)_{t\ge 0}$ satisfies $\E_{t-1}[X'_t]\ge (1+\varepsilon')X'_{t-1}$ (i.e., $X'_t$ has a multiplicative drift), then $Y'_t:=\frac{X'_t}{(1+\varepsilon')^t}$ is a submartingale and we can repeat the same argument to ensure the exponential growth of $Y'_t$.

\paragraph*{The naive approach does not work.}
If we try to apply this drift argument to asynchronous \ThreeMajority, we  obtain the following multi-step concentration of $\alpha_t(i)$:
If $\alpha_{t-1}$ is an almost balanced configuration (i.e., $\alpha_{t-1}(i)=\Theta(1/k)$ for all $i\in[k]$) and if $\alpha_{t-1}(i) = \norm{\alpha_{t-1}}^2 + \Theta(1/k)$,
the normalized population $\alpha_t(i)$ has an additive drift as
\begin{align}
\E_{t-1}[\alpha_t(i)]&=\alpha_{t-1}(i) \left( 1+\frac{\alpha_{t-1}(i) - \norm{\alpha_{t-1}}^2}{n} \right) \label{eq:additive drift alpha} \\
&= \alpha_{t-1}(i)+\Omega \left( \frac{1}{nk^2} \right). \nonumber
\end{align}
Thus, the sequence $\left( \alpha_t(i) - \frac{ct}{nk^2} \right)_{t\ge 0}$ is a submartingale for some small constant $c>0$ unless the configuration is almost balanced.
In the asynchronous model, in one step at most one vertex updates its opinion.
Therefore, %we have 
$\abs{\alpha_t - \alpha_{t-1}}\le \frac{ 1 }{ n }$.
By the Azuma--Hoeffding inequality for submartingales, we obtain a lower tail that
\begin{align}\label{AH-first}
    \Pr \left[ \alpha_t(i) - \frac{ct}{nk^2} \le \alpha_0(i) - \lambda \right] \le \exp \left( - \frac{n^2\lambda^2}{2t} \right).
\end{align}
For example, suppose we are trying to show that $\alpha_t(i)\ge \alpha_0(i)$ with high probability for $t\approx kn$.
If we set $t=kn$ and $\lambda=\frac{ct}{nk^2}=\frac{c}{k}$, the upper tail is $\exp \left( -\Omega \left( \frac{n}{k^3} \right) \right)$; which again requires $k\ll n^{1/3}$.

\paragraph*{Our solution: The Freedman inequality.}
    The argument in the previous paragraph does not work for our setting of $k\approx n^{1/2}$ since the tail bound from the Azuma--Hoeffding inequality is not sharp enough.
    By exploiting the variance of the martingale difference, we can obtain a sharper bound.
    Such a bound, namely, the Freedman inequality \cite{Fre75}, can be seen as a Bernstein type bound for martingales.
    Roughly speaking, it shows that for any martingale $(X_t)_{t\ge 0}$ such that $\abs{X_t-X_{t-1}}$ is small enough and $\E_{t-1}[(X_t - X_{t-1})^2]\le S$, then we have $X_t = X_0 \pm O(\sqrt{S T \log n})$ for all $t=0,\dots,T$ with high probability.
        Recall that the Azuma--Hoeffding inequality yields that, if $\left| X_t - X_{t-1} \right| \le D$,
        then $X_t=X_0\pm O \left( \sqrt{D^2T\log n} \right)$ with high probability.
    Thus, we can obtain a sharper tail bound than the Azuma-Hoeffding inequality from the Freedman inequality if $S\ll D^2$.
    For $X_t=\alpha_t(i)$, we have $S=\E_{t-1} \left[ \left( \alpha_t(i) - \alpha_{t-1}(i) \right)^2 \right] \le n^{-2}\cdot \Pr_{t-1}[\alpha_t(i)\neq \alpha_{t-1}(i)] = O \left( \frac{1}{n^2k} \right)$ if $\alpha_{t-1}(i)=O(1/k)$.
    This improves the tail bound~\eqref{AH-first} in the previous subsection by a factor of $k$ in the exponent.

\subsubsection{Overview of Gap Emerging Phase} \label{sec:overview gap emerging}
\paragraph*{Weak and strong opinions.}
    We consider two classes of opinions:
    An opinion $i\in[k]$ is \emph{strong} at step $t$ if
        $\alpha_t(i)\geq \frac{7}{8}\norm{\alpha_t}^2$,
        and is \emph{weak} at step $t$ if
        $\alpha_t(i) \leq \frac{3}{4}\norm{\alpha_t}^2$.
    Note that there always exists a strong opinion since $\max_{i\in[k]}\cbra*{\alpha_t(i)} \geq \norm{\alpha_t}^2$.
    We also note that being non-weak just means $\alpha_t(i)\ge \frac{3}{4}\norm{\alpha_t}^2$ (not necessarily being strong).
    %TODO remark on benefit of ell2, not ellinf
    
    Suppose that the initial configuration is balanced (i.e., $\alpha_0(1)=\dots=\alpha_0(k)= 1/k = \norm{\alpha_0}^2$).
    Then, all opinions are initially strong.
    Our key technical lemma is the \emph{unique strong opinion lemma},
    which claims that after $O(kn\log n)$ steps,
        there remains only one strong opinion with high probability.
    To state it more formally, let $\istar{t}\in[k]$ be the opinion such that $\alpha_t(\istar{t})=\max_{i\in[k]}\cbra*{\alpha_t(i)}$.
    Then, the unique strong opinion lemma claims that $\max_{i\neq\istar{T}}\cbra*{\alpha_T(i)} \leq \frac{7}{8}\norm{\alpha_T}^2$ with high probability for some $T=O(kn\log n)$.
    
    The proof of the unique strong opinion lemma consists of two parts:
    \begin{enumerate}
        \item First, we prove that once an opinion $i\in[k]$ becomes weak, then $i$ is very unlikely to become strong in the future (weak cannot become strong). \label{item:first part}
        \item Second, we prove that, for {\em any two non-weak distinct opinions} $i,j\in[k]$, either $i$ or $j$ becomes weak within $O(kn\log n)$ steps with probability $1-O(n^{-2})$. \label{item:second part}
    \end{enumerate}
    By the union bound over $i,j\in[k]$, we have that for any pair of strong opinions, either $i$ or $j$ becomes weak within $O(kn\log n)$ steps with probability $1-O(k^2n^{-2})=1-O(n^{-1})$.
    This proves the claim above.

    In what follows, we present our approach to the proof of \cref{item:first part,item:second part}.
    \paragraph*{Proof outline of \cref{item:first part} (\cref{sec:weak cannot be strong}).}
        This statement can be seen as a relaxation of the “weak remains weak" property of \cite{nearly_tight_analysis}.
        The crucial difference is that we allow a gap between the population of weak and strong opinions.
        We show that if $\alpha_0(i)\leq\frac{3}{4}\norm{\alpha_0}^2$, then $\alpha_t(i)<\frac{7}{8}\norm{\alpha_t}^2$ for all $t=1,2,\dots,n^{10}$
        with high probability.
        The proof relies on the martingale-based argument for the Gambler's ruin problem (see, e.g., \cite[Theorem 4.8.9]{Dur19}).
        To state it more formally, fix a weak opinion~$i$
            and let $\alpharatio_t = \frac{\alpha_t(i)}{\norm{\alpha_t}^2}$.
        Note that $R_0\le \frac{3}{4}$ since the opinion $i$ is weak.
        By calculation, we can show that $\E_{t-1}[\alpharatio_t]\leq \alpharatio_{t-1} - \Omega\rbra*{\frac{1}{kn}}$.
        Let $\tau=\inf\cbra*{t\geq 0\colon \alpharatio_t\not \in \left[ \frac{5}{8},\frac{7}{8} \right]}$.
        Then, $(\alpharatio_t)_{0 \le t < \tau}$ can be seen as an asymmetric random walk over the interval $\left[\frac{5}{8}, \frac{7}{8}  \right] $ starting from $\alpharatio_0\leq \frac{6}{8}$.
        We can show that, for some $\phi=\omega(\log n)$, the sequence $(\exp(\phi \alpharatio_{t\land \tau}))_{t\geq 0}$ is a supermartingale; where $a\land b=\min\{a,b\}$.
        By the optimal stopping theorem (\cref{thm:gambler}), we have $\E[\exp(\phi \alpharatio_\tau)]\leq \exp(\phi \alpharatio_0)$.
        This gives $\Pr\sbra*{\alpharatio_\tau\geq \frac{7}{8}}=n^{-\omega(1)}$.

    \paragraph*{Proof outline of \cref{item:second part} (\cref{sec:pair of non-weak opinions vanishes}).}
    We shall look at the bias of two opinions.
    Fix two distinct non-weak opinions $i,j\in[k]$ at step $0$ and let $\delta_t=\abs{\alpha_t(i) - \alpha_t(j)}$
    and
    $\tauweak$ be the first time that either $i$ or $j$ becomes weak, that is,
    \[
        \tauweak = \inf\cbra*{t\geq 0\colon \min\cbra*{\alpha_t(i), \alpha_t(j)} \leq \frac{3}{4}\norm{\alpha_t}^2}.
    \]
    By calculation, if $t\leq \tauweak$, from \cref{eq:additive drift alpha}, we have
    \begin{align*}
        \E_{t-1}[\delta_t] \geq \abs*{\E_{t-1}[\alpha_t(i) - \alpha_t(j)]} = \delta_{t-1}\rbra*{1+\frac{\alpha_{t-1}(i) + \alpha_{t-1}(j) - \norm{\alpha_{t-1}}^2}{n}} \ge \rbra*{1+\frac{1}{2kn}}\delta_{t-1}.
    \end{align*}
    In other words, the bias $\delta_t$ is multiplied by at least $\rbra*{1+\frac{1}{2kn}}$ in expectation for any $t\leq \tauweak$.
    Using the Freedman inequality with bounding the second moment of the one-step difference $\E_{t-1}[(\delta_t-\delta_{t-1})^2]$,
        we can show that either $\tauweak\leq t$ or otherwise $\delta_t\approx \rbra*{1+\frac{1}{kn}}^t \delta_0$ with probability $1-\exp(-\Omega(n\delta_0^2))$.
    If $\delta_0\ge \sqrt{\log n/n}$, then this tail bound is strong enough.

    To ensure that $\delta_T\geq \sqrt{\log n/n}$ for some $T=O(kn \log n)$, 
%    {\color{red} [Will this imply that one of these two opinions has become weak?]}
    we exploit an additive drift of $\delta_t$.
    By calculation, we can see that
    \[
        \Var_{t-1}\sbra*{\delta_t} \geq \frac{\alpha_{t-1}(i) + \alpha_{t-1}(j) - 5\delta_{t-1}^2}{n^2}.
    \]
    If $\delta_t$ is small enough and $t\leq\tauweak$, then $\delta_t^2$ admits an additive drift of $\Omega\rbra*{\frac{1}{n^2k}}$.
    Then, from the optimal stopping theorem, we can show that $\delta_{t+nk} \ge \delta_t + \Omega\rbra*{\frac{1}{\sqrt{n}}}$ or otherwise $\tauweak\le t+nk$ with constant probability.
    Using drift analysis result of \citet*{Doerr11}, we show that $\delta_t\approx \sqrt{\log n/n}$ with high probability for some $t=O(nk \log n)$.
%    Indeed, using the more sophisticated argument of \cite{DGM+11}, it is straightforward to see that $\tauweak=O(nk\log n)$ with high probability.

\subsection{Organization}
In \cref{sec:preliminaries}, we define notations and our model.
In \cref{sec:unique strong opinion lemma}, we analyze the gap emergence phase by proving the unique strong opinion lemma.
In \cref{sec:from unique strong opinion to consensus}, we consider the gap amplification phase, that is, the dynamics starting with a configuration having a sufficient initial gap.
In \cref{sec:lower bound on consensus time}, we prove lower bounds on the consensus time.
Combining results from \cref{sec:unique strong opinion lemma,sec:from unique strong opinion to consensus,sec:lower bound on consensus time} , we prove \cref{thm:small k} in \cref{sec:lower bound 3-majority}.
In \cref{sec:many opinions}, we consider the dynamics starting with arbitrary number of opinions.
The proof of \cref{thm:many opinions} will be given in \cref{sec:remaining opinions in Voter}.
We present the outline of the proof of \cref{thm:all k} in \cref{fig:diagram}.

\begin{figure}[htbp] %TODO modify the height
    \begin{center}
        \begin{tikzpicture}[
            every node/.style={shape=rectangle,draw,minimum width=2.5cm, minimum height=0.7cm,align=left},
            arrow/.style={-latex}
          ]
          
          % Nodes
          \node (thm1) at (0,0) {\cref{thm:all k}};
          \node (tmp1) [right=of thm1,shape=coordinate] {};
          \node (thm2-1) [right=of tmp1] {\cref{thm:many opinions} \\ (shown in \cref{sec:remaining opinions in Voter})};
          \node (thm2-2) [below=of thm2-1] {\cref{thm:small k} \\ (shown in \cref{sec:lower bound 3-majority})};
          \node (tmp2) [below=of thm2-1,left=of thm2-2,shape=coordinate] {};
          \node (tmp3) [right=of thm2-2,shape=coordinate]{};
          \node (lemma5-1a) [right=of tmp3] {\cref{lem:AfterOneStrong} (gap amplification)};
          \node (lemma5-1b) [below=of lemma5-1a] {\cref{lem:unique strong opinion lemma} (gap emergence)};
          \node (tmp4) [below=of lemma5-1a,left=of lemma5-1b,shape=coordinate] {};
          \node (lemma2-3) [right=of thm2-1] {\cref{lem:coupling} \\ (coupling \Voter and \ThreeMajority)};
          
          % Edges
          \draw [->] (tmp1) -- (thm1);
          \draw [-] (thm2-1) -- (tmp1);
          \draw [-] (thm2-2) -- (tmp2);
          \draw [-] (tmp2) -- (tmp1);
          \draw [-] (lemma5-1a) -- (tmp3);
          \draw [->] (tmp3) -- (thm2-2);
          \draw [-] (lemma5-1b) -- (tmp4) -- (tmp3);
        %  \draw [-] (tmp4) -- (tmp3);
          \draw [->] (lemma2-3) -- (thm2-1);
        \end{tikzpicture}
    \caption{An overview of the proof of \cref{thm:all k}. \label{fig:diagram}}
    \end{center}
\end{figure}
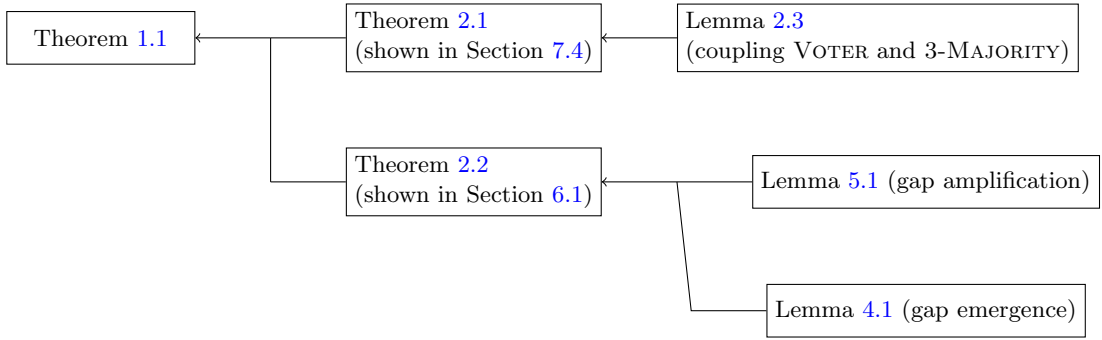
%\fnote{I would move this figure to somewhere before Section 2.2 or later, right now the reader cannot really understand it because it uses theorems and lemmas that are not even mentioned.}

\section{Preliminaries} \label{sec:preliminaries}
Let $\Nat_0=\{0\}\cup\Nat$ denote the set of non-negative integers.
For $n\in\Nat$, let $[n]\defeq\{1,\dots,n\}$.
For a sequence of random variables $(X_t)_{t\in\Nat_0}$,
we write $\E_t[\cdot]$ to denote the conditional expectation $\E\sbra*{\cdot \condition \mathcal{F}_t}$, where $(\mathcal{F}_t)_{t\in\Nat_0}$ is the natural filtration associated with $(X_j)_{0\leq j\leq t}$.
When it is clear from the context, we will omit the term 
“with probability $1$".
%For a random variable $X$ and a real value $r$, we write $X\le r$ (resp.\ $X\ge r$) to denote that $\Pr[X\le r]=1$ (resp.\ $\Pr[X\ge r]=1$).
For convenience, we sometimes use the notation $a\land b=\min\{a,b\}$.
For a vector $x\in\Real^V$ and $p \in \Real_{\geq 0}\cup\{\infty\}$, let $\norm{x}_p=\rbra*{\sum_{v\in V}x_v^p}^{1/p}$ denote the $\ell^p$-norm of $x$,
where we define$\norm{x}_\infty\defeq \max_{v\in V}\{x_v\}$.
In particular, we use the shorthand notation $\norm{x}=\norm{x}_2$ for the $\ell^2$-norm.
We use the standard asymptotic notation, $O(\cdot), o(\cdot), \Omega(\cdot)$, etc.

    %We define \ThreeMajority process and present its basic properties.
    We present the formal definition of the asynchronous \ThreeMajority as follows.
    \begin{definition} \label{def:3Majority}
    Let $n,k\in\Nat$ be parameters such that $1\le k \le n$ and $V=[n]$.
    The \emph{asynchronous 3-Majority process over $k$ opinions on $V$}
    is the discrete-time Markov chain $(\ThreeMajorityProcess_t)_{t\in\Nat_0}$ over the state space $[k]^{V}$, where $\ThreeMajorityProcess_t$ (for $t\geq 1$) is
    obtained from $\ThreeMajorityProcess_{t-1}$ by the following procedure:
    \begin{enumerate}
        \item Select uniformly random $v,w_1,w_2,w_3\in V$, independently and with replacement.
        \item Define $\ThreeMajorityProcess_t\in[k]^V$ by
        \[
            \ThreeMajorityProcess_t(u) = \begin{cases}
                \ThreeMajorityProcess_{t-1}(u) & \text{if $u\neq v$},\\
                \ThreeMajorityProcess_{t-1}(w_1) & \text{if $u=v$ and $\ThreeMajorityProcess_{t-1}(w_1)=\ThreeMajorityProcess_{t-1}(w_2)$}, \\
                \ThreeMajorityProcess_{t-1}(w_3) & \text{otherwise}.
            \end{cases}
        \]
    \end{enumerate}
    For simplicity, we shortly call such process $(\ThreeMajorityProcess_t)_{t\in\Nat_0}$ \emph{asynchronous 3-Majority process}.
    The \emph{consensus time} $\tau_{\mathrm{cons}}$ is the stopping time defined by
    \[
        \taucons = \inf\cbra*{t\geq 0 \colon \text{for some $i\in[k]$ and all $v\in V$, $\ThreeMajorityProcess_t(v)=i$}}.
    \]
    The \emph{fractional population} is the sequence of random vectors $(\alpha_t)_{t\in\Nat_0}$ where each $\alpha_t\in[0,1]^k$ is defined by
    \[
        \alpha_t(i) = \frac{1}{|V|}\abs*{\{v\in V\colon \ThreeMajorityProcess_t(v)=i\}}.
    \]
    For $t\geq 0$ and $i,j\in[k]$,
    the \emph{bias} $\delta_t(i,j)$ is defined as $\delta_t(i,j) \defeq \abs{\alpha_t(i) - \alpha_t(j)}$.
    Let $\alphanorm_t=\norm{\alpha_t}^2$.
    \end{definition}
    
    An element of $V$ is called a \emph{vertex}
    and an element of $[k]$ is called an \emph{opinion}.
    We always use $n$ and $k$ to denote the number of vertices and opinions, respectively.
    We consider the asymptotic behavior of $\tau_{\mathrm{cons}}$ for $n\to\infty$.
    Thus, $k=k(n)$ is usually a function on $n$.
    
    In the following, we will need some technical bounds about the expectation and variance.
    These are collected here for convenience.
    \begin{lemma} \label{lem:basic inequalities}
         Consider the asynchronous \ThreeMajority process over $k$ opinions on $V=[n]$.
         Fix any time step $t\geq 1$ and two distinct opinions $i,j\in[k]$.
        Then, we have the following.
        \begin{enumerate}
            \item \label{item:expectation of alpha}
                $\E_{t-1}[\alpha_t(i)]=\alpha_{t-1}(i)\rbra*{1+\frac{\alpha_{t-1}(i)-\norm{\alpha_{t-1}}^2}{n}}$.
            \item \label{item:difference of alpha}
                $\abs{\alpha_t(i)-\alpha_{t-1}(i)}\leq \frac{1}{n}$.
            \item \label{item:square difference of alpha}
                $\E_{t-1}[(\alpha_t(i)-\alpha_{t-1}(i))^2]\leq \frac{3\alpha_{t-1}(i)}{n^2}$.
                %; thus, $\Var_{t-1}[\alpha_t(i)]=\Var_{t-1}[\alpha_t(i) - \alpha_{t-1}(i)] \leq \frac{3\alpha_{t-1}(i)}{n^2}$.
            \item \label{item:expectation of delta}
                $\E_{t-1}[\delta_t(i,j)]\geq \delta_{t-1}(i,j)\rbra*{1+\frac{\alpha_{t-1}(i)+\alpha_{t-1}(j)-\alphanorm_{t-1}}{n}}$.
            \item \label{item:difference of delta}
                $|\delta_t(i,j)-\delta_{t-1}(i,j)|\leq \frac{2}{n}$.
            \item \label{item:square difference of delta}
                $\E_{t-1}\sbra*{\rbra*{\delta_t(i,j)-\delta_{t-1}(i,j)}^2}\leq \frac{12(\alpha_{t-1}(i)+\alpha_{t-1}(j))}{n^2}$.
            \item \label{item:variance of delta}
                $\Var_{t-1}\sbra*{ \delta_t(i,j) } \geq \frac{\alpha_{t-1}(i)+\alpha_{t-1}(j) - 5\delta_{t-1}(i,j)^2}{n^2}$.
            \item \label{item:expectation of 2norm}
                $\E_{t-1}\sbra*{\alphanorm_{t}}
                =\alphanorm_{t-1}+\frac{2}{n}\rbra*{1-\frac{1}{n}}\rbra*{\norm{\alpha_{t-1}}_3^3-\alphanorm_{t-1}^2}+\frac{2}{n^2}\rbra*{1-\alphanorm_{t-1}}
                \geq \alphanorm_{t-1}$.
            \item \label{item:square difference of 2norm}
                $\E_{t-1}\sbra*{\rbra*{\alphanorm_{t}-\alphanorm_{t-1}}^2}
                \leq \frac{24\norm{\alpha_{t-1}}_3^3}{n^2}$; thus, $\Var_{t-1}\sbra*{\alphanorm_{t}} \leq \frac{24\alphanorm_{t-1}^{1.5}}{n^2}$.
            \item \label{item:difference of 2norm}
                $\abs*{\alphanorm_{t}-\alphanorm_{t-1}}
                \leq \frac{6\max_{j\in[k]}\alpha_{t-1}(j)}{n}\leq \frac{6\sqrt{\alphanorm_{t-1}}}{n}$.
        \end{enumerate}
    \end{lemma}
    \Cref{item:expectation of alpha} is shown by \cite[Lemma 2.1]{simple_dynamics} in the synchronous model; we present the proof here since
    \cref{item:expectation of alpha} is particularly important.
    The other claims (\cref{item:difference of alpha,item:square difference of alpha,item:expectation of delta,item:difference of delta,item:square difference of delta,item:variance of delta,item:expectation of 2norm,item:square difference of 2norm,item:difference of 2norm}) follow from a straightforward calculation, which can be found in \cref{sec:proof of basic facts}.

    \begin{proof}[Proof of \cref{item:expectation of alpha}.]
        Fix the $(t-1)$-th configuration and
        let $g(i) \defeq \alpha_{t-1}(i)\rbra*{ 1+\alpha_{t-1}(i) - \norm{\alpha_{t-1}}^2 }$.
        Fix a vertex $v$.
        Conditioned on $v$ is chosen to update at the $(t-1)$-th step, $v$ holds opinion $i$ at the $t$-th step with probability
            \begin{align*}
                \underbrace{\alpha_{t-1}(i)^2}_{\text{$i$ is chosen twice}}+\underbrace{(1-\norm{\alpha_{t-1}}^2)\alpha_{t-1}(i)}_{\text{no opinion is chosen twice and the third vertex holds $i$}} = g(i).
            \end{align*}
        Since $v$ is activated with probability $1/n$,
        we have $\E_{t-1}[\alpha_t(i)]=\alpha_{t-1}(i)+\frac{(1-\alpha_{t-1}(i))g(i)}{n}-\frac{\alpha_{t-1}(i)(1-g(i))}{n}=\alpha_{t-1}(i)\rbra*{1+\frac{\alpha_{t-1}(i)-\norm{\alpha_{t-1}}^2}{n}}$.
    \end{proof}

    For simplicity, the following \cref{sec:unique strong opinion lemma,sec:from unique strong opinion to consensus,sec:lower bound on consensus time} 
    assume $k = o(\sqrt{n/\log n})$ but our proof works for $k \le c\sqrt{n / \log n}$ for a sufficiently small constant $c>0$ without any modification.
\section{Unique Strong Opinion Lemma} \label{sec:unique 
strong opinion lemma}
Throughout this section, consider the asynchronous \ThreeMajority over $k = o(\sqrt{n/\log n})$ opinions on $V=[n]$.
Let $(\alpha_t)_{t\in\Nat_0}$ be the fractional population.
We say that an opinion $i\in[k]$ is \emph{weak at step $t$}
    if $\alpha_t(i)\leq \frac{3}{4}\alphanorm_t$
and is \emph{strong at step $t$}
    if $\alpha_t(i)\geq \frac{7}{8}\alphanorm_t$.
Note that there can be an opinion that is neither weak nor strong.
Note also that the most popular opinion is always strong since $\alphamax{t}\geq \alphanorm_t$.
It can be the case that there are two or more strong opinions.
In this section, we prove that, after $O(nk\log n)$ steps, there remains only one strong opinion.
\begin{lemma}[Unique Strong Opinion Lemma] \label{lem:unique strong opinion lemma}
    For $t\in\Nat_0$, let $\istar{t}\in[k]$ be an opinion that attains the maximum population, i.e., $\istar{t} \in \argmax_i\{\alpha_t(i)\}$, where ties are broken to the opinion with the smallest index.
    Then, for some
    $T=O(nk\log n)$,
    we have $\Pr\sbra*{\max_{i\neq \istar{T}}\alpha_T(i)\leq \frac{7}{8}\alphanorm_t} \geq 1-O(n^{-8})$.
\end{lemma}
If $\alpha_t(i)=\alpha_j(t)=\norm{\alpha_t}_\infty$ for some $i\neq j$, then $\max_{i\neq \istar{t}}\alpha_t(i) = \norm{\alpha_t}_\infty$.
Therefore, \cref{lem:unique strong opinion lemma} means that a huge gap emerges between the most and the second most popular opinions at the $T$-th step.

The proof of \cref{lem:unique strong opinion lemma} consists of two parts.
First, we prove that, once an opinion becomes weak,
    then the opinion is very unlikely to become strong in the future.
The proof can be found in \cref{sec:weak cannot be strong}.
\begin{lemma}[Weak cannot become strong] \label{lem:weak cannot be strong}
For any weak opinion $i\in[k]$ at step $0$,
we have $\Pr\sbra*{\forall 0\leq t\leq n^{10}, \alpha_t(i) \le \frac{3}{4}\gamma_t}\geq 1-n^{-\omega(1)}$.
\end{lemma}
Second, we prove that, for any two distinct opinions $i,j\in[k]$,
    either $i$ or $j$ becomes weak within $\tilde{O}(nk)$ steps with high probability.
The proof can be found in \cref{sec:pair of non-weak opinions vanishes}.
\begin{lemma}[Pair of Non-Weak Opinions Vanishes] \label{lem:pair of non-weak opinions vanishes}
    There exists a universal constant $C>0$ that satisfies the following:
    For any distinct non-weak opinions $i,j\in[k]$ at step $0$ (i.e., $\alpha_0(i),\alpha_0(j)> \frac{3}{4}\gamma_0$), with probability $1-O(n^{-9})$,
    either $i$ or $j$ becomes weak within $Cnk\log n$ steps.
\end{lemma}

\Cref{lem:unique strong opinion lemma} immediately follows from the two results above.
\begin{proof}[Proof of \cref{lem:unique strong opinion lemma} assuming \cref{lem:weak cannot be strong,lem:pair of non-weak opinions vanishes}.]
    From \cref{lem:pair of non-weak opinions vanishes}, the following event occurs with probability $1-O(n^{-9})$:
    For any pair of non-weak opinions, at least one of them becomes weak within $Cnk\log n$ steps, where $C>0$ is some universal constant.
    Moreover, with probability $1-n^{-\omega(1)}$, any weak opinion cannot become strong during $n^{10}$ steps after being weak.
    By the union bound, with probability $1-O(k^2n^{-9})=1-O(n^{-8})$,
    the following event occurs:
        For any pair of distinct non-weak opinions $i,j\in[k]$, either $i$ or $j$ is weak at step $Cnk\log n$.
    In other words,
        there exists only one strong opinion at step $Cnk\log n$.
    This completes the proof.
\end{proof}

\subsection{Weak cannot become strong} \label{sec:weak cannot be strong}
This subsection is devoted to proving \cref{lem:weak cannot be strong}, which asserts that once an opinion becomes weak, then the opinion is very unlikely to become strong in the future.
Indeed, we prove the following general result, which contains \cref{lem:weak cannot be strong} as a special case.
Throughout this subsubsection, let $\alpharatio_t = \frac{\alpha_t(i)}{\norm{\alpha_t}^2}=\frac{\alpha_t(i)}{\alphanorm_t}$.
\begin{lemma} \label{lem:weak cannot be strong general constant}
    Let $k=o(\sqrt{n/\log n})$ and $0<L<U<1$ be any constants.
    Fix an opinion $i\in[k]$ and let $\alpharatio_t=\frac{\alpha_t(i)}{\alphanorm_t}$.
    If $\alpharatio_0\leq L$, 
    then we have $\Pr\sbra*{\forall 0\leq t\leq n^{10},\alpharatio_t\leq U}\geq 1-n^{-\omega(1)}$.
\end{lemma}
In what follows, we prove \cref{lem:weak cannot be strong general constant},
    which immediately implies \cref{lem:weak cannot be strong} by setting
        $L=\frac{3}{4}$ and $U=\frac{7}{8}$.
We first observe that $\alpharatio_t$ has an additive drift towards negative direction conditioned on $L\leq \alpharatio_{t-1} \leq U$ and the second moment of one-step difference is always small.
\begin{lemma}
\label{lem:ai_L2norm}
Let $k=o(\sqrt{n})$ and $0<L<U<1$ be any constants.
Fix an opinion $i\in[k]$.
For $t\in\Nat_0$, let $\alpharatio_t=\frac{\alpha_t(i)}{\alphanorm_t}$ 
and let $\mathcal{E}_t$ be the event that $L \leq \alpharatio_t \leq U$.
Then, we have the following:
    \begin{enumerate}[label=(\roman*)]
        \item \label{lab:ail2norm1}
            $\E_{t-1}\sbra*{\alpharatio_t \condition \mathcal{E}_{t-1}}\leq \alpharatio_{t-1}-\frac{(1-U)L-o(1)}{kn}$.
        \item \label{lab:ail2norm2}
            $\E_{t-1}\sbra*{\rbra*{\alpharatio_t - \alpharatio_{t-1}}^2 \condition \mathcal{E}_{t-1}}\leq \frac{(24+o(1))k}{n^2}$.
        \item \label{lab:ail2norm3}
            $\abs*{\alpharatio_t - \alpharatio_{t-1}}\leq \frac{14k}{n}$.
    \end{enumerate}
\end{lemma}
\cref{lem:ai_L2norm} holds under the weaker condition $k=o(\sqrt{n})$ than the condition $k=o(\sqrt{n\log n})$ of \cref{lem:unique strong opinion lemma}.
The proof of \cref{lem:ai_L2norm} relies on the approximation of $\E\sbra*{\frac{X}{Y}}$ in terms of $\E[X]$ and $\E[Y]$.
The proof can be found in \cref{sec:proof of basic facts}.

Next, we show that \cref{lem:ai_L2norm} 
implies that $(\alpharatio_t)$ is likely to decrease to $\alpharatio_t\leq \frac{L+U}{2}$ before 
growing to $\alpharatio_t\geq U$
if $L\leq \alpharatio_0\leq \frac{L+U}{2}$.
Here, the condition $k=o(\sqrt{n\log n})$ (or indeed $k\le 0.01\sqrt{n\log n})$) is crucial.
\begin{lemma} \label{lem:drop before rise}
Suppose $k=o(\sqrt{n/\log n})$ and let $0<L<U<1$ be constants.
Fix an opinion $i\in[k]$ and let $\alpharatio_t=\frac{\alpha_t(i)}{\alphanorm_t}$ for $t\geq 0$.
Suppose that $\alpharatio_0\leq\frac{L+U}{2}$
and let
    $\tau = \inf\cbra{t\geq 0\colon \alpharatio_t\leq L\text{ or }\alpharatio_t\geq U}$ be the stopping time.
    Then, $\Pr\sbra*{\alpharatio_\tau\leq L}\geq 1-n^{-\omega(1)}$.
\end{lemma}
\begin{proof}
    From \cref{lem:ai_L2norm}, for sufficiently large $n$, we have
\begin{align*}
    &|\alpharatio_{t+1}-\alpharatio_t|\leq \frac{14k}{n},\\
    &\indicator_{\tau>t-1}\E_{t-1}\sbra*{\alpharatio_t-\alpharatio_{t-1} \condition \mathcal{E}_{t-1}} + \frac{(1-U)L}{2kn} \le 0, \\
    &\indicator_{\tau>t-1}\E_{t-1}[\rbra{\alpharatio_t-\alpharatio_{t-1}}^2] \le \frac{25k}{n^2}.
\end{align*}
    We apply \cref{thm:gambler} with $X_t = \alpharatio_t$ and
    \begin{align*}
        &X_0=\alpharatio_0\le \frac{L+U}{2}, \hspace{2em} D=\frac{14k}{n},\hspace{2em} \theta=\frac{(1-U)L}{2kn}, \hspace{2em} S=\frac{25k}{n^2}.
    \end{align*}
    Since $k=o(\sqrt{n/\log n})$, we have $\phi = \frac{6\theta}{3S+2D\theta} = \omega(\log n)$ and thus
    \begin{align*}
        \Pr\sbra{\alpharatio_\tau \geq U} \leq \frac{\e^{\phi X_0} - \e^{\phi(L-D)}}{\e^{\phi U} - \e^{\phi(L-D)}} = \e^{-\Omega(\phi)} = n^{-\omega(1)}.
    \end{align*}
\end{proof}

\begin{proof}[Proof of \cref{lem:weak cannot be strong general constant}.]
It suffices to show that $\alpharatio_t< U$ for all $t=0,\dots,n^{10}$ with probability $1-n^{-\omega(1)}$ conditioned on $\alpharatio_0\leq L$.

For each $s=0,\dots,n^{10}$, let
\begin{align*}
    \tau^{\uparrow}_s = \inf\cbra*{t>s\colon \alpharatio_t\geq U}, & &
    \tau^{\downarrow}_s = \inf\cbra*{t>s\colon \alpharatio_t\leq L}
\end{align*}
and $\mathcal{B}_s$ be the bad event that $\tau^{\uparrow}_s<\tau^{\downarrow}_s$.
Since we start with $\alpharatio_0\leq L$ and $\abs{\alpharatio_t-\alpharatio_{t-1}}\leq \frac{14k}{n}=o(1)$,
    there must exist some $s\in\{0,\dots,\tau^{\uparrow}_0\}$
        such that both $\mathcal{B}_s$ and $\alpharatio_s\leq \frac{L+U}{2}$ occurs.
From \cref{lem:drop before rise}, we have $\Pr\sbra*{\mathcal{B}_s\condition \alpharatio_s\leq L}\leq n^{-\omega(1)}$ for every $s$.
Therefore, we have
\begin{align*}
    \Pr\sbra*{\tau^{\uparrow}_0\leq n^{10}}\leq \Pr\sbra*{\bigcup_{s=0}^{n^{10}} \rbra*{\mathcal{B}_s\text{ and }\alpharatio_s\le \frac{L+U}{2}}} \leq \sum_{s=0}^{n^{10}} \Pr\sbra*{\mathcal{B}_s \condition \alpharatio_s\leq \frac{L+U}{2}} \leq n^{-\omega(1)}.
\end{align*}
\end{proof}

\subsection{Pair of non-weak opinions} \label{sec:pair of non-weak opinions vanishes}
This part is devoted to proving \cref{lem:pair of non-weak opinions vanishes}.
We have two key components in the proof of \cref{lem:pair of non-weak opinions vanishes}.

First, for any distinct non-weak opinions $i,j\in[k]$, we show that $\delta_t(i,j)=\abs{\alpha_t(i) - \alpha_t(j)}$ (\cref{def:3Majority}) grows by a multiplicative factor or otherwise one of $i,j$ becomes weak within $O(nk)$ steps with high probability provided that $\delta_0(i,j)$ is sufficiently large.
In this part, for simplicity, we use $\delta_t$ to denote $\delta_t(i,j)$.
\begin{lemma}\label{lem:deltaUPorWeak}
    Let $i,j$ be two distinct initially non-weak opinions (i.e., $\alpha_0(i),\alpha_0(j)> \frac{3}{4}\alphanorm_0$).
    Then, there exists a universal constant $C>0$ such that,
    for
    $T=\frac{n}{8(\alpha_0(i)+\alpha_0(j))} \leq \frac{kn}{12}$,
    with probability $1-\e^{-C\alpha_0(i)^2n}-\e^{-C\alpha_0(j)^2n} - \e^{-C\delta_0^2n}$, at least one of the following events occurs:
    \begin{itemize}
        \item For some $0\leq t\leq T$, we have $\min\{\alpha_t(i), \alpha_t(j)\}\leq \frac{3}{4}\alphanorm_t$ (i.e., either $i$ or $j$ becomes weak).
        \item For some $0\leq t\leq T$, we have $\delta_t \geq 1.01\delta_0$.
    \end{itemize}
\end{lemma}

If $\delta_0\geq \sqrt{\log n/n}$, then \cref{lem:deltaUPorWeak} implies that $\delta_t$ grows by a multiplicative factor or otherwise, either $i$ or $j$ becomes weak within $O(kn)$ steps with high probability.

Second, we prove that  $\delta_t\geq\sqrt{\log n/n}$ with high probability for some $t=O(kn\log n)$ even if we start with $\delta_0=0$\footnote{Note that \cref{lem:deltaUPorWeak additive} does not directly imply the gap between the most and the second most popular opinions (we need \cref{lem:weak cannot be strong} to ensure it: See the proof of \cref{lem:unique strong opinion lemma}).}.
\begin{lemma}\label{lem:deltaUPorWeak additive}
    Let $i,j\in[k]$ be two distinct initially non-weak opinions
    and let $C_0>0$ be an arbitrary constant.
    Then, there exists a constant $C>0$ that, for $T=Cnk\log n$, with probability $1-O(n^{-10})$, at least one of the following events occurs:
    \begin{itemize}
        \item For some $0\leq t\leq T$, we have $\min\cbra*{\alpha_t(i),\alpha_t(j)}\leq\frac{3}{4}\alphanorm_t$.
        \item For some $0\leq t\leq T$, we have $\delta_t \geq C_0\sqrt{\frac{\log n}{n}}$ for some $0\leq t\leq T$.
    \end{itemize}
\end{lemma}

Below, we show how to prove \cref{lem:pair of non-weak opinions vanishes} assuming \cref{lem:deltaUPorWeak,lem:deltaUPorWeak additive}.
\begin{proof}[Proof of \cref{lem:pair of non-weak opinions vanishes} from \cref{lem:deltaUPorWeak,lem:deltaUPorWeak additive}.]
    Let $i,j\in[k]$ be two non-weak opinions.
    In particular, note that
$\alpha_0(i),\alpha_0(j)\geq \frac{3}{4}\alphanorm_0$.
    Let
        \begin{align}
            \tauweak=\inf\cbra*{t\geq 0:\alpha_t(i)\leq \frac{3}{4}\alphanorm_{t} \textrm{ or } \alpha_t(j)\leq \frac{3}{4}\alphanorm_{t} } \label{eq:tauweak}
        \end{align}
    be the first time that either $i$ or $j$ becomes weak.
    Let $C_0 = \sqrt{\frac{10}{C_{\hbox{\tiny{\ref{lem:deltaUPorWeak}}}}}}$ and $T=C_{\hbox{\tiny{\ref{lem:deltaUPorWeak additive}}}} nk\log n$,
    where
    $C_{\hbox{\tiny \ref{lem:deltaUPorWeak}}}, C_{\hbox{\tiny{\ref{lem:deltaUPorWeak additive}}}}>0$ are the universal constants denoted by $C$ in \cref{lem:deltaUPorWeak,lem:deltaUPorWeak additive}, respectively.
    Let $\tau_0,\tau_1,\dots$ be stopping times where $\tau_0=0$ and for $\ell\geq 1$,
    \[
    \tau_\ell = \inf\cbra*{t \geq \tau_{\ell-1}\colon t\geq\tauweak\text{ or }\delta_t\geq \max\cbra*{ C_0\sqrt{\frac{\log n}{n}}, 1.01\delta_{\tau_{\ell-1}} }}.
    \]
    Note that $\tau_\ell=\tau_{\ell-1}$ if $\tau_{\ell-1}\geq\tauweak$.
    From \cref{lem:deltaUPorWeak additive}, we have
    \[
    \Pr\sbra*{\tau_1 > T} = O(n^{-10})
    \]
    and from \cref{lem:deltaUPorWeak}, for $\ell\geq 2$, we have
    \begin{align*}
        \Pr\sbra*{\tau_\ell - \tau_{\ell-1} > \frac{kn}{12}} \leq O(n^{-10}).
%         \e^{-C_{\hbox{\tiny \ref{lem:deltaUPorWeak}}}\alpha_{\tau_{\ell-1}}(i)^2n} + \e^{-C_{\hbox{\tiny \ref{lem:deltaUPorWeak}}}\alpha_{\tau_{\ell-1}}(j)^2n} + O(n^{-10}) = O(n^{-10}).
    \end{align*}
    Here, $\delta_{\tau_1}\geq C_0\sqrt{\frac{\log n}{n}}$ and thus $\exp(-C_{\hbox{\tiny \ref{lem:deltaUPorWeak}}}\delta_{\tau_{\ell-1}}^2)=O(n^{-10})$ and $\e^{-C_{\hbox{\tiny \ref{lem:deltaUPorWeak}}}\alpha_{\tau_{\ell-1}}(i)^2n} = \exp(-\Omega(n/k^2))=n^{-\omega(1)}$ since $k=o(\sqrt{n/\log n})$.

    Note that, for $L\defeq \log_{1.01}n$, we have
    $\tauweak\leq\tau_L$ (otherwise, we would have $\delta_{\tau_L}>1$).
    Therefore, with probability $1-O(n^{-10}/\log n)$, we have $\tauweak\leq\tau_L\leq O(T+kn\log n)=O(nk\log n)$.
\end{proof}
To prove \cref{lem:deltaUPorWeak}, fix two initially non-weak opinions $i\neq j$
and
consider $\tauweak$ defined in \cref{eq:tauweak}.
Define the following stopping times:
\begin{align*}
    &\tau_{\delta}^\uparrow=\inf\cbra*{t\geq 0:\delta_t\geq 1.01\delta_0},  \\
    & \tau_i^\uparrow=\inf\cbra*{t\geq 0:\alpha_t(i)\geq 2\alpha_0(i)}, \\
    & \tau_i^\downarrow=\inf\cbra*{t\geq 0:\alpha_t(i)\leq \frac{1}{2}\alpha_0(i)}.
\end{align*}
We also define $\tau^{\uparrow}_j,\tau^{\downarrow}_j$ in the same way.
We prove two technical lemmas that follow from our multiplicative drift lemma (\cref{lem:multipricative_drift_Freedman}).
\begin{lemma}
\label{lem:alphaUD}
    Let $i\in[k]$ be any non-weak opinion.
    Then, we have the following:
    \begin{enumerate}[label=(\Roman*)]
        \item \label{item:alphaU} 
        There is a constant $C>0$ such that 
        %For any constant $c_\alpha^\uparrow>1$, 
        $\Pr\sbra*{\tau_i^\uparrow\leq \frac{n}{8\alpha_0(i)}}\leq \exp\rbra*{-C\alpha_0(i)^2n}$.
        %holds for some positive constants $C$ and $C'$ depending only on $c_\alpha^\uparrow>1$.
        \item \label{item:alphaD} 
        %For any constants $0<c_{w}<1/2$, $c_\alpha^\uparrow>1$, and $0<c_\alpha^\downarrow<1$, 
        There is a constant $C>0$ such that 
        $\Pr\sbra*{\tau_i^\downarrow\leq \min\cbra*{\frac{n}{8\alpha_0(i)},\tau_i^\uparrow,\tauweak}}\leq \exp\rbra*{-C\alpha_0(i)^2n}$.
        %for some positive constants $C$ and $C'$ depending only on $c_\alpha^\uparrow$, $c_\alpha^\downarrow$, and $c_w$.
    \end{enumerate}
\end{lemma}
\begin{proof}[Proof of \cref{item:alphaU} of \cref{lem:alphaUD}]
    If $\tau_i^\uparrow>t-1$, 
    % \begin{align}
    %    \indicator_{\tau>t-1}\E_{t-1}\sbra*{\alpha_{t}(i)}
    %    &=\indicator_{\tau>t-1}\alpha_{t-1}(i)\rbra*{1+\frac{\alpha_{t-1}(i)-\alphanorm_{t-1}}{n}}
    %    \leq \indicator_{\tau>t-1}\alpha_{t-1}(i)\rbra*{1+\frac{2\alpha_{0}(i)}{n}}. \label{eq:alphaUPc}
    %\end{align}
    \begin{align}
        \E_{t-1}\sbra*{\alpha_{t}(i)}
        &=\alpha_{t-1}(i)\rbra*{1+\frac{\alpha_{t-1}(i)-\alphanorm_{t-1}}{n}}
        \leq \alpha_{t-1}(i)\rbra*{1+\frac{2\alpha_{0}(i)}{n}} \label{eq:alphaUPc}
    \end{align}
    holds from \cref{item:expectation of alpha} of \cref{lem:basic inequalities}. 
    Let $X_t=-\alpha_{t}(i)$. 
    Write $\tau=\tau_i^\uparrow$ for convenience.
    Then, we have the following inequalities:
    \begin{itemize}%[label=(\alph*)]
        \item \label{item:alphaU1} 
            $\indicator_{\tau>t-1}\E_{t-1}\sbra*{X_t-\rbra*{1+\frac{2\alpha_{0}(i)}{n}}X_{t-1}}\geq 0$ (From \cref{eq:alphaUPc}).
        \item \label{item:alphaU2} 
            $|X_t-X_{t-1}|\leq \frac{1}{n}$ (From \cref{item:difference of alpha} of \cref{lem:basic inequalities}).
        \item \label{item:alphaU3}
            $\indicator_{\tau>t-1}\E_{t-1}\sbra*{\rbra*{X_t-X_{t-1}}^2}
            \leq
            \indicator_{\tau>t-1}\cdot \frac{3\alpha_{t-1}(i)}{n^2}
            \le 
            \frac{6\alpha_{0}(i)}{n^2}$
            (From \cref{item:square difference of alpha} of \cref{lem:basic inequalities}).
        \item \label{item:alphaU4}
            $\indicator_{\tau>t-1}|X_{t-1}|\leq 2\alpha_0(i)$ (From the definition of $\tau_i^\uparrow$).
    \end{itemize}
    Let $T=\frac{n}{8\alpha_0(i)} \leq \frac{\log\rbra{2}}{4}\cdotp \frac{n}{\alpha_0(i)}$.
    Then, applying \cref{lem:multipricative_drift_Freedman} with 
    $a=1+\frac{2\alpha_{0}(i)}{n}$, 
    $A=\sum_{t=1}^T\rbra*{1+\frac{2\alpha_{0}(i)}{n}}^{-2t}\leq \frac{n}{2\alpha_{0}(i)}$, 
    $B=1$,
    $D=1/n$, 
    $S=\frac{6\alpha_{0}(i)}{n^2}$, 
    $U=2\alpha_0(i)$, and $\lambda=\rbra*{\sqrt{2}-1}\alpha_0(i)$, we obtain
    \begin{align*}
        \Pr\sbra*{\tau\leq T}
        &=\Pr\sbra*{\tau\leq \min\cbra{T,\tau}}
        =\Pr\sbra*{\bigvee_{t=0}^{T\wedge \tau}\cbra*{\alpha_t(i)\geq 2\alpha_0(i)}}
        =\Pr\sbra*{\bigvee_{t=0}^{T\wedge \tau}\cbra*{\alpha_{t\wedge \tau}(i)\geq 2\alpha_0(i)}}\\
        &\leq \Pr\sbra*{\bigvee_{t=0}^{T\wedge \tau}\cbra*{\alpha_{t\wedge \tau}(i)\geq \sqrt{2}\rbra*{1+\frac{2\alpha_{0}(i)}{n}}^{t\wedge \tau}\alpha_0(i)}}\\
        &\leq  \Pr\sbra*{\bigvee_{t=0}^{T}\cbra*{\alpha_{t\wedge \tau}(i)\geq \rbra*{1+\frac{2\alpha_{0}(i)}{n}}^{t\wedge \tau}\rbra*{\alpha_0(i)+\rbra*{\sqrt{2}-1}\alpha_0(i)}}}\\
        &= \Pr\sbra*{\bigvee_{t=0}^{T}\cbra*{X_{t\wedge \tau}\leq \rbra*{1+\frac{2\alpha_{0}(i)}{n}}^{t\wedge \tau}\rbra*{X_0-\rbra*{\sqrt{2}-1}\alpha_0(i)}}}\\
        &\leq \exp\rbra*{-C\alpha_0(i)^2n}
    \end{align*}
    for some constant $C>0$.
    Note that we use $2 \geq \sqrt{2}\rbra*{1+\frac{2\alpha_{0}(i)}{n}}^T$ in the first inequality.
\end{proof}
\begin{proof}[Proof of \cref{item:alphaD} of \cref{lem:alphaUD}]
    If $\min\cbra*{\tauweak,\tau_i^\uparrow}>t-1$, 
    \begin{align}
        \E_{t-1}\sbra*{\alpha_{t}(i)}
        =\alpha_{t-1}(i)\rbra*{1+\frac{\alpha_{t-1}(i)-\alphanorm_{t-1}}{n}}
        \geq \alpha_{t-1}(i)\rbra*{1-\frac{\alpha_{t-1}(i)}{3n}} 
        \geq \alpha_{t-1}(i)\rbra*{1-\frac{2\alpha_{0}(i)}{3n}} \label{eq:alphaDownc}
    \end{align}
    holds from \cref{item:expectation of alpha} of \cref{lem:basic inequalities}.
    Note that $\alphanorm_{t-1}\leq \frac{4}{3}\alpha_{t-1}(i)$ holds for $\tauweak>t-1$.
    Write $\tau=\min\cbra*{\tauweak,\tau_i^\uparrow}$ 
    %and $c=\frac{c_\alpha^\uparrow c_{w}}{1-c_{w}}$ 
    for convenience.
    Then,  we have the following inequalities:
    \begin{itemize}%[label=(\alph*')]
        \item \label{item:alphaD1} $\indicator_{\tau>t-1}\E_{t-1}\sbra*{\alpha_t(i)-\rbra*{1-\frac{2\alpha_{0}(i)}{3n}}\alpha_{t-1}(i)}\geq 0$ (From \cref{eq:alphaDownc}).
        \item \label{item:alphaD2} $|\alpha_t(i)-\alpha_{t-1}(i)|\leq \frac{1}{n}$ (From \cref{item:difference of alpha} of \cref{lem:basic inequalities}).
        \item \label{item:alphaD3} $\indicator_{\tau>t-1}\E_{t-1}\sbra*{\rbra*{\alpha_t(i)-\alpha_{t-1}(i)}^2}\leq \frac{6\alpha_{0}(i)}{n^2}$ (From \cref{item:square difference of alpha} of \cref{lem:basic inequalities}).
        \item \label{item:alphaD4} $\indicator_{\tau>t-1}|\alpha_{t-1}(i)|\leq 2\alpha_0(i)$ (From the definition of $\tau_i^\uparrow$).
    \end{itemize}
    Let $T= \frac{n}{8\alpha_0(i)}\leq \frac{\log(2)}{4(2/3)}\cdotp \frac{n}{\alpha_0(i)}$.
    From \cref{lem:multipricative_drift_Freedman} with
%
%    \begin{align*}
%        & X_t=\alpha_t(i), & & A=\sum_{t=1}^T\rbra*{1-\frac{2\alpha_{0}(i)}{3n}}^{-2t}\leq %\frac{3n}{4 \alpha_0(i)}, \\
%        & D=\frac{1}{n}, & & S=\frac{6\alpha_{0}(i)}{n^2}, \\
%        & U=2\alpha_0(i), & & \lambda=\rbra*{1-\frac{1}{\sqrt{2}}}\alpha_0(i),
%    \end{align*}
%
    $X_t=\alpha_t(i)$, 
    $A=\sum_{t=1}^T\rbra*{1-\frac{2\alpha_{0}(i)}{3n}}^{-2t}\leq \frac{3n}{4 \alpha_0(i)}$,
    $B=\rbra*{1-\frac{2\alpha_{0}(i)}{3n}}^{-2T}\leq 2$,
    $D=\frac{1}{n}$,
    $S=\frac{6\alpha_{0}(i)}{n^2}$, $U=2\alpha_0(i)$, and $\lambda=\rbra*{1-\frac{1}{\sqrt{2}}}\alpha_0(i)$,
    we obtain
    \begin{align*}
        \Pr\sbra*{\tau_i^\downarrow\leq \min\cbra*{T,\tau}}
        &=\Pr\sbra*{\bigvee_{t=0}^{T\wedge \tau}\cbra*{\alpha_t(i)\leq \frac{1}{2}\alpha_0(i)}}
        =\Pr\sbra*{\bigvee_{t=0}^{T\wedge \tau}\cbra*{\alpha_{t\wedge\tau}(i)\leq \frac{1}{2}\alpha_0(i)}}\\
        &\leq \Pr\sbra*{\bigvee_{t=0}^{T\wedge \tau}\cbra*{\alpha_{t\wedge\tau}(i)\leq \rbra*{1-\frac{2\alpha_0(i)}{3n}}^{t\wedge\tau}\frac{1}{\sqrt{2}}\alpha_0(i)}}\\
        &\leq \Pr\sbra*{\bigvee_{t=0}^{T}\cbra*{\alpha_{t\wedge\tau}(i)\leq \rbra*{1-\frac{2\alpha_0(i)}{3n}}^{t\wedge\tau}\rbra*{\alpha_0(i)-\rbra*{1-\frac{1}{\sqrt{2}}}\alpha_0(i)}}}\\
         &\leq \exp\rbra*{-C\alpha_0(i)^2n}
    \end{align*}
    for some positive constant $C$.
    Note that we use $\frac{1}{2}\leq \rbra*{1-\frac{2\alpha_0(i)}{3n}}^T\frac{1}{\sqrt{2}}$ in the first inequality.
\end{proof}
\begin{lemma}
\label{lem:deltaUP}
%For any constants $c_\alpha^\uparrow>1$, $0<c_\alpha^\downarrow<1$, $0<c_w<1/2$, and $C>0$, there are positive constants $c_\delta^\uparrow>1$ and 
There is a constant $C>0$ such that 
\begin{align*}
    \Pr\sbra*{\min\cbra*{\tau_\delta^\uparrow,\tauweak,\tau_i^\uparrow,\tau_j^\uparrow,\tau_i^\downarrow,\tau_j^\downarrow}>\frac{n}{8\rbra*{\alpha_0(i)+\alpha_0(j)}}}\leq \exp\left(-C\delta_0^2n\right).
\end{align*}
\end{lemma}
\begin{proof}
Let $\tau=\min\cbra*{\tau_\delta^\uparrow,\tauweak,\tau_i^\uparrow,\tau_j^\uparrow,\tau_i^\downarrow,\tau_j^\downarrow}$.
If $\tau>t-1$,
\begin{align}
    \E_{t-1}\sbra*{\delta_t}
    &\geq \delta_{t-1}\rbra*{1+\frac{\alpha_{t-1}(i)+\alpha_{t-1}(j)-\alphanorm_{t-1}}{n}}  \nonumber\\
    &\geq \delta_{t-1}\rbra*{1+\frac{\alpha_{t-1}(i)+\alpha_{t-1}(j)}{3n}}   \nonumber\\
    &\geq \delta_{t-1}\rbra*{1+\frac{\alpha_{0}(i)+\alpha_{0}(j)}{6n}}  \label{eq:deltaUPc}
\end{align}
holds from \cref{item:expectation of delta} of \cref{lem:basic inequalities}.
The second inequality follows from $\alphanorm_{t-1}=\frac{\alphanorm_{t-1}}{2}+\frac{\alphanorm_{t-1}}{2}\leq \frac{2}{3}\rbra*{\alpha_{t-1}(i)+\alpha_{t-1}(j)}$.
Note that both $\alphanorm_{t-1}\leq \frac{4}{3}\alpha_{t-1}(i)$ and $\alphanorm_{t-1}\leq \frac{4}{3}\alpha_{t-1}(j)$ hold for $\tauweak>t-1$.
Furthermore, we observe that the following:
\begin{itemize}%[label=(\roman*)]
    \item \label{item:deltaUP1} 
    $\indicator_{\tau>t-1} \E_{t-1}\sbra*{\delta_t-  \rbra*{1+\frac{\alpha_{0}(i)+\alpha_{0}(j)}{6n}}\delta_{t-1}(i)}\geq 0$ (from \cref{eq:deltaUPc}).
    \item \label{item:deltaUP2} $|\delta_t-\delta_{t-1}|\leq \frac{2}{n}$ (from \cref{item:difference of delta}  of \cref{lem:basic inequalities}).
    \item \label{item:deltaUP3} 
    $\indicator_{\tau>t-1}\E_{t-1}\sbra*{\rbra*{\delta_t-\delta_{t-1}}^2}\leq \frac{24(\alpha_{0}(i)+\alpha_0(j))}{n^2}$ (from \cref{item:square difference of delta} of \cref{lem:basic inequalities}).
    \item \label{item:deltaUP4} 
    $\indicator_{\tau>t-1}|\delta_{t-1}|\leq 1.01\delta_0$ (from the definition of $\tau_\delta^\uparrow$).
\end{itemize}
%$c=\frac{(1-2c_w)c_\alpha^\downarrow}{2(1-c_w)}$ and 
Let $T=\frac{n}{8\rbra*{\alpha_0(i)+\alpha_0(j)}}$.
%Note that $c$ is a positive constant depending only on $c_w$ and $c_\alpha^\downarrow$.
%Let $c_\delta^\uparrow=\rbra*{1+\frac{\alpha_0(i)+\alpha_0(j)}{6n}}^T\mathrm{e}^{-\frac{1}{192}}$.
%Note that $c_\delta^\uparrow\geq \mathrm{e}^{\frac{1}{96}}$ is a positive constant strictly larger than $1$.
%
From
\cref{lem:multipricative_drift_Freedman} with
$X_t=\delta_t$,
$A=\sum_{t=1}^T\rbra*{1+\frac{\alpha_0(i)+\alpha_0(j)}{6n}}^{2t}\leq \frac{6n}{\alpha_0(i)+\alpha_0(j)}$,
$B=1$,
$D = \frac{2}{n}$,
$S=\frac{24(\alpha_{0}(i)+\alpha_0(j))}{n^2}$,
$U=1.01\delta_0$,
and
$\lambda=\rbra*{1-\mathrm{e}^{-\frac{1}{192}}}\delta_0$,
%\begin{align*}
%& X_t=\delta_t,
%& &
%A=\sum_{t=1}^T\rbra*{1+\frac{\alpha_0(i)+\alpha_0(j)}{6n}}^{2t}\leq \frac{6n}{\alpha_0(i)+\alpha_0(j)}, \\
%& D = \frac{2}{n},
%& &
%S=\frac{24(\alpha_{0}(i)+\alpha_0(j))}{n^2},
%\\
%& U=1.01\delta_0,
%& &
%\lambda=\rbra*{1-\mathrm{e}^{-\frac{1}{192}}}\delta_0,
%\end{align*}
we have
\begin{align*}
    \Pr\sbra*{\tau>T}
    &=\Pr\sbra*{\tau>T,\delta_{T}\leq 1.01\delta_0}\\
    &\leq  \Pr\sbra*{\tau>T,\delta_{T}\leq \rbra*{1+\frac{\alpha_0(i)+\alpha_0(j)}{6n}}^T\mathrm{e}^{-\frac{1}{192}}\delta_0}\\
    &=\Pr\sbra*{\tau>T,\delta_{T\wedge\tau}\leq \rbra*{1+\frac{\alpha_0(i)+\alpha_0(j)}{6n}}^{T\wedge \tau}\rbra*{\delta_0-\rbra*{1-\mathrm{e}^{-\frac{1}{192}}}\delta_0}}\\
    &\leq \Pr\sbra*{\delta_{T\wedge\tau}\leq \rbra*{1+\frac{\alpha_0(i)+\alpha_0(j)}{6n}}^{T\wedge \tau}\rbra*{\delta_0-\rbra*{1-\mathrm{e}^{-\frac{1}{192}}}\delta_0}}\\
    &\leq \exp\rbra*{-C\delta_0^2n}
\end{align*}
for some constant $C>0$.
Note that we use $1.01<\mathrm{e}^{\frac{1}{96}}\leq \rbra*{1+\frac{\alpha_0(i)+\alpha_0(j)}{6n}}^T\mathrm{e}^{-\frac{1}{192}}$ in the first inequality.
 % depending only on $c_w$, $c_\alpha^\uparrow$, $C$ and $c_\alpha^\downarrow$.
%Note that $\sum_{t=1}^T\rbra*{1+c\frac{\alpha_0(i)+\alpha_0(j)}{n}}^{-2t}\leq \frac{n}{c(\alpha_{0}(i)+\alpha_{0}(j))}$.
\end{proof}
\begin{proof}[Proof of \cref{lem:deltaUPorWeak}]
Let $T=\frac{n}{8\rbra*{\alpha_0(i)+\alpha_0(j)}}$, $T_i=\frac{n}{8\alpha_0(i)}$, and $T_j=\frac{n}{8\alpha_0(j)}$.
It suffices to prove
    \begin{align*}
        &\Pr\sbra*{\min\cbra*{\tau_\delta^\uparrow,\tauweak}>T}
        \leq \e^{-C\alpha_0(i)^2n}+\e^{-C\alpha_0(j)^2n}+\e^{-C\delta_0^2n}.
    \end{align*}
Note that
    \begin{align}
        \Pr\sbra*{\tau_\delta^\uparrow>T,\tauweak>T}
        &=\Pr\sbra*{\tau_\delta^\uparrow>T,\tauweak>T,\cbra*{\tau_i^\uparrow\leq T \textrm{ or }\tau_j^\uparrow\leq T}} \label{eq:taudeltaweak1}\\
        &+\Pr\sbra*{\tau_\delta^\uparrow>T,\tauweak>T,\tau_i^\uparrow> T,\tau_j^\uparrow> T, \cbra*{\tau_i^\downarrow\leq T \textrm{ or }\tau_j^\downarrow\leq T}} \label{eq:taudeltaweak2}\\
        &+\Pr\sbra*{\tau_\delta^\uparrow>T,\tauweak>T,\tau_i^\uparrow> T,\tau_j^\uparrow> T, \tau_i^\downarrow> T,\tau_j^\downarrow>T} \label{eq:taudeltaweak3}
    \end{align}
    For the first term, from the union bound and \cref{item:alphaU} of \cref{lem:alphaUD}, 
    \begin{align*}
        \cref{eq:taudeltaweak1}
        &\leq \Pr\sbra*{\tau_i^\uparrow\leq T}+\Pr\sbra*{\tau_j^\uparrow\leq T}\\
        &\leq \Pr\sbra*{\tau_i^\uparrow\leq T_i}+\Pr\sbra*{\tau_j^\uparrow\leq T_j}\\
        &\leq \exp\rbra*{-C_1\alpha_0(i)^2n}+\exp\rbra*{-C_1\alpha_0(j)^2n}
    \end{align*}
    for some positive constant $C_1$.
    Note that $T\leq T_i$ and $T\leq T_j$ hold.
    For the second term, from the union bound and \cref{item:alphaD} of \cref{lem:alphaUD}, 
    \begin{align*}
        \cref{eq:taudeltaweak2}
        &\leq \Pr\sbra*{\tau_i^\downarrow\leq T,\tauweak>T,\tau_i^\uparrow> T}
        +\Pr\sbra*{\tau_j^\downarrow\leq T,\tauweak>T,\tau_j^\uparrow> T}\\
        &\leq \Pr\sbra*{\tau_i^\downarrow\leq \min\cbra*{T,\tauweak,\tau_i^\uparrow}}
        +\Pr\sbra*{\tau_j^\downarrow\leq \min\cbra*{T,\tauweak,\tau_j^\uparrow}}\\
        &\leq \Pr\sbra*{\tau_i^\downarrow\leq \min\cbra*{T_i,\tauweak,\tau_i^\uparrow}}
        +\Pr\sbra*{\tau_j^\downarrow\leq \min\cbra*{T_j,\tauweak,\tau_j^\uparrow}}\\
        &\leq \exp\rbra*{-C_2\alpha_0(i)^2n}+\exp\rbra*{-C_2\alpha_0(j)^2n}
    \end{align*}
    for some positive constant $C_2$.
    For the third term, from \cref{lem:deltaUP}, 
    \begin{align*}
        \cref{eq:taudeltaweak3}
        &=\Pr\sbra*{\min\cbra*{\tau_\delta^\uparrow,\tauweak,\tau_i^\uparrow,\tau_j^\uparrow,\tau_i^\downarrow,\tau_j^\downarrow}>T}
        \leq \exp\left(-C_3\delta_0^2n\right).
    \end{align*}
    for some positive constant $C_3$.
    Combining the above, we obtain the claim.
\end{proof}
\begin{proof}[Proof of \cref{lem:deltaUPorWeak additive}.]
The proof consists of two parts.
First, we invoke the additive drift of $\delta_t^2$ and apply \cref{thm:OST_mar} to show that $\delta_t$ changes by order $\Omega(1/\sqrt{n})$ within $O(kn)$ steps with constant probability.
Second, we use the result of drift analysis from \cite{Doerr11} to prove \cref{lem:deltaUPorWeak additive}.

\paragraph*{Step 1. Additive Drift of $\delta_t^2$.}
Let
$\tau^+_\delta = \inf\cbra*{t\geq 0\colon \delta_t\ge \frac{C_0}{ \sqrt{n} }}$.
We claim
\begin{align}
    \Pr\sbra*{\min\{\tauweak,\tau^+_\delta\}\geq 3C_0^2nk}\leq \frac{1}{3}-o(1)<\frac{1}{2}. \label{eq:constant change of delta}
\end{align}
In other words, for some $t=O(kn)$, either $\tauweak\le t$ or $\delta_t\ge \frac{C_0}{\sqrt{n}}$ occurs (for any initial condition) with probability $1/2$.

To prove \cref{eq:constant change of delta}, we apply \cref{thm:OST_mar} for $X_t=\delta_t$ and $\tau=\min\cbra*{\tauweak,\tau^+_\delta}$.
Note that $\tau<\infty$ almost surely.
In what follows, we check that the two conditions of \cref{thm:OST_mar} hold.
From \cref{item:variance of delta} of \cref{lem:basic inequalities}, we have 
\begin{align*}
    \indicator_{\tau > t-1}\cdot \rbra*{\Var_{t-1}[\delta_t] - \frac{1}{n^2 k}} &\geq
    \indicator_{\tau > t-1}\cdot \rbra*{\frac{\alpha_{t-1}(i) + \alpha_{t-1}(j) - 5\delta_{t-1}^2}{n^2} - \frac{1}{n^2 k}}\\
    &\geq \indicator_{\tau > t-1}\cdot \frac{1}{n^2}\rbra*{ \frac{3}{2}\alphanorm_t - O\rbra*{\frac{\log n}{n}} - \frac{1}{k} } & & (t-1 < \tau) \\
    &\geq 0.   & & (\alphanorm_t\geq 1/k)
\end{align*}
Also, since $\alpha_{t-1}(i)+\alpha_{t-1}(j)\geq \frac{3}{2}\alphanorm_{t-1}$ if $t-1<\tau$, we have
\begin{align*}
    \indicator_{\tau > t-1}\cdot\rbra*{\E_{t-1}\sbra*{\delta_t}^2 - \delta_{t-1}^2} 
    &\geq \delta_{t-1}^2\cdot \rbra*{\rbra*{ 1 + \frac{\alpha_{t-1}(i)+\alpha_{t-1}(j)-\alphanorm_{t-1}}{n}}^2 - 1} \\
    &\geq \delta_{t-1}^2\cdot \rbra*{\rbra*{ 1 + \frac{\alphanorm_{t-1}}{2n}}^2 - 1} \\
    &\geq 0.
\end{align*}
Therefore, from \cref{thm:OST_mar}, we obtain 
\begin{align*}
    \E[\tau] &\leq n^2k \E[\delta_{\tau}^2] 
    = n^2k\E[(\delta_{\tau} - \delta_{\tau-1} + \delta_{\tau-1})^2] \\
    &\leq n^2k\rbra*{ \frac{4}{n^2} + \frac{2}{n}\cdot \frac{C_0}{\sqrt{n}} + \frac{C_0}{n}  } & & \text{$\delta_{\tau-1}<\frac{C_0}{\sqrt{n}}$ and \cref{item:difference of delta} of \cref{lem:basic inequalities}}\\
    &\leq (C_0^2 + o(1))nk.
\end{align*}
The claim \cref{eq:constant change of delta} follows by the Markov inequality.

\paragraph*{Apply the Drift Analysis result of \cite{Doerr11}.}
We invoke the result of drift analysis credited to \citet{Doerr11}.
We use a slightly modified version from \cite[Lemma 21]{fast_convergence_undecided}.
\begin{lemma} \label{lem:doerr11_lemma}
Let $(W(s))_{s\in\Nat_0}$ be a sequence of random variables that takes value in $\Nat_0$ and satisfies $W(0)=0$, $W(s)\le W(s-1)$, and
\begin{align*}
    &\Pr \left[ W(s)=1 \condition W(s-1)=0 \right] \ge \frac{1}{2},\\
    &\Pr \left[ W(s)=w+1 \condition W(s-1)=w \right] \ge 1-\exp\rbra*{-(1.01)^{w}},\\
    &\Pr \left[ W(s)=0 \condition W(s-1)=w \right] \le \exp\rbra*{-(1.01)^{w}}.
\end{align*}
Fix any constant $C^*>0$ and
define a stopping time $\tau_W$ by $\tau_W=\inf \left\{ s\ge 0\colon W(s) \ge C^*\log\log n \right\}$, i.e., the first time that $W(s)$ reaches $\left\lceil C^*\log\log n \right\rceil$.
Then, there exists a constant $C'$ such that $\tau_W\le C'\log n$ with probability $1-O(n^{-10})$.
\end{lemma}

To prove \cref{lem:deltaUPorWeak additive}, we need more notation.
Define disjoint intervals $E_0,E_1,\dots$ by
\begin{align*}
    E_s = \begin{cases}
        \left[ 0,\frac{C_0}{\sqrt{n}} \right) & \text{if $\ell=0$},\\
        \left[ \frac{C_0(1.01)^{\ell-1}}{\sqrt{n}}, \frac{C_0(1.01)^{\ell}}{\sqrt{n}} \right) & \text{if $\ell\ge 1$}.
    \end{cases}
\end{align*}
(Recall that $C_0$ is the constant in the statement of \cref{lem:deltaUPorWeak additive}.)
For time step $t\in\Nat_0$, let $\mathsf{label}_t\in\Nat_0$ be the index $\ell$ of the interval such that $\delta_t\in E_\ell$.
Fix a sufficiently large constant $C>0$ and
define stopping times $\tau(0),\tau(1),\dots$ by $\tau(0)=0$ and for $s\ge 1$,
\begin{align*}
    \tau(s) = \inf \left\{ t\ge \tau(s-1) \colon t\ge \tauweak \text{ or }\mathsf{label}_t\neq \mathsf{label}_{\tau(s-1)} \text{ or }t-\tau(s-1)>Ckn \right\}.
\end{align*}
We say that the time interval $[\tau(s-1),\tau(s))$ \emph{good} if either $\tau(s)\ge\tauweak$ or $\mathsf{label}_{\tau(s)} = \mathsf{label}_{\tau(s-1)}+1$ occurs.
In other words, the time interval $[\tau(s-1), \tau(s))$ is good if either $\tauweak$ comes or $\delta_t \ge (1.01)\delta_{\tau(s-1)}$.
Note that each time interval $[\tau(s-1),\tau(s))$ has length at most $Ckn$ because of the definition of $\tau(s)$.

Using \cref{lem:doerr11_lemma} (with $C^*=\frac{\log 2}{\log 1.01}$), we claim that, for $L=\left\lceil \log_{1.01}\log n \right\rceil$ and for some $s^*=O(\log n)$, the intervals $[\tau(0),\tau(1)),\dots,[\tau(s^*-1),\tau(s^*))$ contain $L$ consecutive good intervals with high probability.
If this occurs and intervals $[\tau(s_0),\tau(s_0+1)),\dots,[\tau(s_0+L-1),\tau(s_0+L))$ are good, then either $\tauweak$ or $\delta_{\tau(s_0+L)} \ge \frac{C_0(1.01)^L}{\sqrt{n}} \ge C_0\sqrt{\frac{\log n}{n}}$.
This proves \cref{lem:deltaUPorWeak additive} since $\tau(s_0+L)\le Cnk\cdot (s_0+L)=O(nk\log n)$.

To prove the claim above, define random variables $(W(s))_{s\in\Nat_0}$ by $W(0)=0$ and for $s\ge 1$,
\begin{align*}
    W(s) = \begin{cases}
        W(s-1) + 1 & \text{if the time interval $[\tau(s-1),\tau(s))$ is good},\\
        0 & \text{otherwise}.\\
    \end{cases}
\end{align*}
In other words, $W(s)$ records the number of consecutive good intervals before step $\tau(s)$.
To check the three conditions of \cref{lem:doerr11_lemma}, observe that
\begin{itemize}
    \item If $\mathsf{label}_t=0$ (i.e., $\delta_t < \frac{C_0}{\sqrt{n}}$), then from \cref{eq:constant change of delta}, we have $\mathsf{label}_{t+t'}=1$ with probability $1/2$ for some $t'=O(kn)$. In other words, we have $\Pr[W(s)=1|W(s-1)=0]\ge \frac{1}{2}$.
    \item If $W(s-1)=w$, then we have either $\tauweak\le \tau(s-1)$ or $\tau_{\tau(s-1)} \ge \frac{C_0(1.01)^w}{\sqrt{n}}$. By \cref{lem:deltaUPorWeak}, the time interval $[\tau(s-1),\tau(s))$ is good with probability $1-n^{-\omega(1)}-\exp \left( - C_{\hbox{\tiny{\ref{lem:deltaUPorWeak}}}}\cdot C_0^2 \cdot (1.01)^w \right) \ge \max \left\{ \frac{1}{2}, 1-\exp \left( -(1.01)^w \right) \right\}$ (the constant $C_0$ is sufficiently large).
    Note that if this occurs, then $W(s)=W(s-1)+1$.
    \item If $W(s-1)=w$, then we have
        \begin{align*}
            \Pr[W(s)=0 | W(s-1)=w] \le \Pr[\text{the interval $[\tau(s-1),\tau(s))$ is not good}] \le \exp(-(1.01)^w).
        \end{align*}
\end{itemize}
These three observations imply that our $(W(s))_{s\in\Nat_0}$ satisfies the three conditions of \cref{lem:doerr11_lemma}.
Therefore, from \cref{lem:doerr11_lemma}, with probability $1-O(n^{-10})$, we have $W(s)$ reaches the target value $L= \left\lceil \log_{1.01}\log n \right\rceil$ for some $s = O(\log n)$.
This completes the proof of \cref{lem:deltaUPorWeak additive}.
\end{proof}
\section{From Unique Strong Opinion to Consensus} \label{sec:from unique strong opinion to consensus}
Let $I_t\in [k]$ be the opinion such that $\alpha_{t}(I_t)=\norm{\alpha_t}_\infty$,
where ties are broken down to the opinion with the smallest index.
Note that $\max_{i\neq I_0}\alpha_0(i)$ is the second largest population.
This section is devoted to proving that
once the second most popular opinion becomes weak,
    then the process reaches consensus
    within $O(kn\log n)$ steps with high probability.
\begin{lemma}\label{lem:AfterOneStrong}
    Suppose $\max_{i\neq \istar{0}}\alpha_0(i)\leq \frac{7}{8}|\alpha_0\|_\infty$ and $k=o(\sqrt{n/\log n})$. 
    Then, for some $T=O(nk\log n)$, we have $\Pr\sbra*{\taucons\leq T}\geq 1-O(n^{-1})$.
\end{lemma}
We use the following notation in this section:
\begin{align*}
    &\tau_{\ell_\infty}^\uparrow=\inf\cbra*{t\geq 0:\norm{\alpha_t}_\infty\geq 2\norm{\alpha_0}_\infty}, 
    &&\tau_{\ell_\infty}^\downarrow=\inf\cbra*{t\geq 0:\norm{\alpha_t}_\infty\leq \frac{9}{10}\norm{\alpha_0}_\infty}, \\
    &\tau_{\ell_\infty}^+=\inf\cbra*{t\geq 0:\norm{\alpha_t}_\infty\geq \frac{2}{3}}, 
    &&\tau_{\ell_\infty}^-=\inf\cbra*{t\geq 0:\norm{\alpha_t}_\infty\leq \frac{3}{5}}, \\
    &\tau_{\mathrm{bad}}=\inf\cbra*{t\geq 0:\max_{i\neq I_t}\alpha_t(i)\geq \frac{15}{16}\norm{\alpha_t}^2}.
\end{align*}
%Note that $\tau_{\mathrm{bad}}>n^{10}$ with probability $1-n^{-\omega(1)}$ from \cref{lem:weak cannot be strong general constant}.
The key tool is the following lemma obtained from the multiplicative drift lemma (\cref{lem:multipricative_drift_Freedman}).
\begin{lemma}
    \label{lem:LinfUD}
    We have the following:
    \begin{enumerate}[label=(\Roman*)]
    \item \label{item:LinfUP} 
    %For any constants $0<c_w<1/2$, $c_\infty^\uparrow>1$, $0<c_\infty^\downarrow<1$, and $0<c_\infty^+<1$,
    $\Pr\sbra*{\min\cbra*{\tau_{\ell_\infty}^\uparrow,\tau_{\mathrm{bad}},\tau_{\ell_\infty}^\downarrow,\tau_{\ell_\infty}^+}>\frac{160n}{\norm{\alpha_0}_\infty}}\leq \exp\rbra*{-C\norm{\alpha_{0}}_\infty^2n}
    $ holds for some  constant $C>0$.
    \item \label{item:LinfDown} 
    %For any constants $0<c_w<1/2$, $c_\infty^\uparrow>1$, and $0<c_\infty^\downarrow<1$ and $T>0$, 
    For any $T>0$, 
    $\Pr\sbra*{\tau_{\ell_\infty}^\downarrow\leq \min\cbra*{T,\tau_{\mathrm{bad}},\tau_{\ell_\infty}^\uparrow}}
     \leq \exp\rbra*{-C\frac{\norm{\alpha_{0}}_\infty n^2}{T+n}}$
    holds for some constant $C>0$.
    \end{enumerate} 
\end{lemma}
\begin{proof}[Proof of \cref{item:LinfUP} of \cref{lem:LinfUD}]
    Let $\tau=\min\cbra*{\tau_{\ell_\infty}^\uparrow,\tau_{\mathrm{bad}}, \tau_{\ell_\infty}^\downarrow,\tau_{\ell_\infty}^+}$.
    %and $c=c_wc_\infty^\downarrow\rbra*{1-c_\infty^+}>0$.
    From \cref{item:expectation of alpha} of \cref{lem:basic inequalities},
    \begin{align}
    \indicator_{\tau>t-1}\E_{t-1}\sbra*{\norm{\alpha_{t}}_\infty}
    &=\indicator_{\tau>t-1}\E_{t-1}\sbra*{\alpha_{t}(I_t)} \nonumber\\
    &=\indicator_{\tau>t-1}\E_{t-1}\sbra*{\alpha_{t}(I_{t-1})} \nonumber\\
    &=\indicator_{\tau>t-1}\alpha_{t-1}(I_{t-1})\rbra*{1+\frac{\alpha_{t-1}(I_{t-1})-\alphanorm_{t-1}}{n}} \label{eq:Linfdownc} \\
    &\geq \indicator_{\tau>t-1}\norm{\alpha_{t-1}}_\infty\rbra*{1+\frac{\norm{\alpha_{t-1}}_\infty\rbra*{1-\norm{\alpha_{t-1}}_\infty}}{16n}} \label{eq:Linfsub} \\
    &\geq \indicator_{\tau>t-1}\norm{\alpha_{t-1}}_\infty\rbra*{1+\frac{3\norm{\alpha_{0}}_\infty}{160n}}. \label{eq:Linfupc}
    \end{align}
    Note that $\norm{\alpha_{t-1}}^2=\norm{\alpha_{t-1}}^2_\infty+\sum_{i\neq I_{t-1}}\alpha_{t-1}(i)^2\leq \norm{\alpha_{t-1}}^2_\infty+\frac{15}{16}\norm{\alpha_{t-1}}_\infty\rbra*{1-\norm{\alpha_{t-1}}_\infty}$
    and $I_t=I_{t-1}$ for $\tau_{\mathrm{bad}}>t-1$.
    Then, we have the following:
    \begin{itemize}%[label=(\roman*)]
        \item \label{item:LinfU1} $\indicator_{\tau>t-1}\E_{t-1}\sbra*{\norm{\alpha_{t}}_\infty-\rbra*{1+\frac{3\norm{\alpha_{0}}_\infty}{160n}}\norm{\alpha_{t-1}}_\infty}\geq 0$ (From \cref{eq:Linfupc}).
        \item \label{item:LinfU2} $\indicator_{\tau>t-1}|\norm{\alpha_{t}}_\infty-\norm{\alpha_{t-1}}_\infty|\leq \frac{1}{n}$ (From \cref{item:difference of alpha} of \cref{lem:basic inequalities}).
        \item \label{item:LinfU3} $\indicator_{\tau>t-1}\E_{t-1}\sbra*{\rbra*{\norm{\alpha_{t}}_\infty-\norm{\alpha_{t-1}}_\infty}^2}\leq \frac{6\norm{\alpha_{0}}_\infty}{n^2}$ (From \cref{item:square difference of alpha} of \cref{lem:basic inequalities}).
        \item \label{item:LinfU4} $\indicator_{\tau>t-1}|\norm{\alpha_{t-1}}_\infty|\leq 2\norm{\alpha_{0}}_\infty$ (From the definition of $\tau_{\ell_\infty}^\uparrow$).
    \end{itemize}
    Note that $\norm{\alpha_{t}}_\infty-\norm{\alpha_{t-1}}_\infty=\alpha_{t}(I_{t-1})-\alpha_{t-1}(I_{t-1})$ for $\tau_{\mathrm{bad}}>t-1$.
    
    Let $T= \frac{160n}{\norm{\alpha_{0}}_\infty} \geq \frac{4\log(2)}{3/160}\frac{n}{\norm{\alpha_{0}}_\infty}$.
    Then, we have $2\leq \frac{1}{2}\rbra*{1+\frac{3\norm{\alpha_{0}}_\infty}{160n}}^T$.
    Thus, applying \cref{lem:multipricative_drift_Freedman} with $X_t=\norm{\alpha_{t-1}}_\infty$, 
    $A=\sum_{t=1}^T\rbra*{1+\frac{3\norm{\alpha_{0}}_\infty}{160n}}^{2t}\leq \frac{n}{c\norm{\alpha_{0}}_\infty}$, 
    $B=1$,
    $D=1/n$, 
    $S=\frac{6\norm{\alpha_{0}}_\infty}{n^2}$, $U=2\norm{\alpha_{0}}_\infty$, and $h=\norm{\alpha_{0}}_\infty/2$, we obtain
\begin{align*}
    \Pr\sbra*{\tau>T}
    &=\Pr\sbra*{\tau>T,\norm{\alpha_{T}}_\infty\leq 2\norm{\alpha_{0}}_\infty}\\
    &\leq \Pr\sbra*{\tau>T,\norm{\alpha_{T}}_\infty\leq \frac{1}{2}\rbra*{1+\frac{3\norm{\alpha_{0}}_\infty}{160n}}^T\norm{\alpha_{0}}_\infty}\\
    &=\Pr\sbra*{\tau>T,\norm{\alpha_{T\wedge \tau}}_\infty\leq \rbra*{1+\frac{3\norm{\alpha_{0}}_\infty}{160n}}^{T\wedge \tau}\rbra*{\norm{\alpha_{0}}_\infty-\frac{\norm{\alpha_{0}}_\infty}{2}}}\\
    &\leq \Pr\sbra*{\norm{\alpha_{T\wedge \tau}}_\infty\leq \rbra*{1+\frac{3\norm{\alpha_{0}}_\infty}{160n}}^{T\wedge \tau}\rbra*{\norm{\alpha_{0}}_\infty-\frac{\norm{\alpha_{0}}_\infty}{2}}}\\
    &\leq \exp\rbra*{-C\norm{\alpha_{0}}_\infty^2n}
\end{align*}
for a positive constant $C$.% depending only on $c_w$, $c_\infty^\uparrow$, $c_\infty^\downarrow$ and $c_\infty+$.
\end{proof}
\begin{proof}[Proof of \cref{item:LinfDown} of \cref{lem:LinfUD}]
    Let $\tau=\min\cbra*{\tau_{\mathrm{bad}},\tau_{\ell_\infty}^\uparrow}$.
    We have the following:
    \begin{itemize}%[label=(\roman*)]
        \item \label{item:LinfD1} $\indicator_{\tau>t-1}\E_{t-1}\sbra*{\norm{\alpha_{t}}_\infty-\norm{\alpha_{t-1}}_\infty}\geq 0$ (From \cref{eq:Linfdownc}).
        \item \label{item:LinfD2} $\indicator_{\tau>t-1}|\norm{\alpha_{t}}_\infty-\norm{\alpha_{t-1}}_\infty|\leq \frac{1}{n}$ (From \cref{item:difference of alpha} of \cref{lem:basic inequalities}).
        \item \label{item:LinfD3} $\indicator_{\tau>t-1}\E_{t-1}\sbra*{\rbra*{\norm{\alpha_{t}}_\infty-\norm{\alpha_{t-1}}_\infty}^2}\leq \frac{6\norm{\alpha_{0}}_\infty}{n^2}$ (From \cref{item:square difference of alpha} of \cref{lem:basic inequalities}).
        \item \label{item:LinfD4} $\indicator_{\tau>t-1}|\norm{\alpha_{t-1}}_\infty|\leq 2\norm{\alpha_{0}}_\infty$ (From the definition of $\tau_{\ell_\infty}^\uparrow$).
    \end{itemize}
    Note that $\norm{\alpha_{t}}_\infty-\norm{\alpha_{t-1}}_\infty=\alpha_{t}(I_{t-1})-\alpha_{t-1}(I_{t-1})$ for $\tau_{\mathrm{bad}}>t-1$.
    Then, applying \cref{lem:multipricative_drift_Freedman} for $X_t=\norm{\alpha_{t}}_\infty$, 
    $A=\sum_{t=1}^T1^{-2t}=T$, 
    $B=1$,
    $D=1/n$, 
    $S=\frac{6\alpha_{0}(i)}{n^2}$, 
    $U=c_\alpha^\uparrow\alpha_0(i)$, and $h=\norm{\alpha_{0}}_\infty/10$, we have
    \begin{align*}
        \Pr\sbra*{\tau_{\ell_\infty}^\downarrow\leq \min\cbra*{T,\tau}}
        &=\Pr\sbra*{\bigvee_{t=0}^{T\wedge \tau}\cbra*{\norm{\alpha_{t}}_\infty\leq \frac{9}{10}\norm{\alpha_{0}}_\infty}}\\
        &=\Pr\sbra*{\bigvee_{t=0}^{T\wedge \tau}\cbra*{\norm{\alpha_{t\wedge \tau}}_\infty\leq \frac{9}{10}\norm{\alpha_{0}}_\infty}}\\
        &\leq \Pr\sbra*{\bigvee_{t=0}^{T}\cbra*{\norm{\alpha_{t\wedge \tau}}_\infty\leq \rbra*{\norm{\alpha_{0}}_\infty-\frac{\norm{\alpha_{0}}_\infty}{10}}}}\\
         &\leq \exp\rbra*{-C\frac{\norm{\alpha_{0}}_\infty n^2}{T+n}}
    \end{align*}
    for some positive constant $C$. % depending only on $c_w$, $c_\infty^\uparrow$, and $c_\infty^\downarrow$.
\end{proof}
\Cref{lem:LinfUD} assures the growth of $\norm{\alpha_t}_\infty$ in the following way. 
\begin{lemma}
\label{lem:PhaseAnalysis}
For $k=o(\sqrt{n/\log n})$, we have the following:
\begin{enumerate}[label=(\Roman*)]
\item \label{item:final_phase:1} $\Pr\sbra*{\min\cbra*{\tau_{\mathrm{bad}}, \tau_{\ell_\infty}^+}>320nk\log n}=n^{-\omega(1)}$.
\item \label{item:final_phase:2} Suppose $\|\alpha_0\|_\infty\geq 2/3$. Then, $\Pr\left[\min\{\tau_{\mathrm{bad}},\taucons\}>100n\log n\right]=O(n^{-1})$.
\end{enumerate}
\end{lemma}
\begin{proof}[Proof of \cref{item:final_phase:1} of \cref{lem:PhaseAnalysis}]
Let $T=\frac{160n}{\norm{\alpha_0}_\infty}$.
From \cref{item:LinfUP,item:LinfDown} of \cref{lem:LinfUD}, we have
    \begin{align}
        &\Pr\left[\tau_{\mathrm{bad}}>T, \tau_{\ell_\infty}^+>T, \tau_{\ell_\infty}^\uparrow>T\right] \nonumber\\
        &=\Pr\left[\tau_{\mathrm{bad}}>T, \tau_{\ell_\infty}^+>T, \tau_{\ell_\infty}^\uparrow>T,\tau_{\ell_\infty}^\downarrow>T\right]
        +\Pr\left[\tau_{\mathrm{bad}}>T, \tau_{\ell_\infty}^+>T, \tau_{\ell_\infty}^\uparrow>T, \tau_{\ell_\infty}^\downarrow\leq T \right] \nonumber\\
        &\leq  \Pr\sbra*{\min\cbra*{\tau_{\ell_\infty}^\uparrow,\tau_{\mathrm{bad}},\tau_{\ell_\infty}^\downarrow,\tau_{\ell_\infty}^+}>T}+\Pr\left[\tau_{\ell_\infty}^\downarrow\leq \min\{T,\tau_{\mathrm{bad}},\tau_{\ell_\infty}^\uparrow\} \right]\nonumber\\
        %&\leq \exp\left(- C'n \norm{\alpha_0}_\infty^2 \right)+\exp\rbra*{-\frac{C''\norm{\alpha_{0}}_\infty n^2}{T+n}}\\
        &\leq \exp\left(- C_1n \norm{\alpha_0}_\infty^2 \right)+\exp\left(- C_2n \norm{\alpha_0}_\infty^2 \right)\nonumber \\
        &=n^{-\omega(1)}.
        \label{eq:PhaseI_1}
    \end{align}
    Here, we use the assumption of $k=o(\sqrt{n/\log n})$.
    Now, consider a sequence of stopping times $\tau_0$, $\tau_1$, $\ldots$ defined as $\tau_0\defeq 0$ and 
    \begin{align*}
        \tau_i\defeq \inf\cbra*{t\geq \tau_{i-1}:\tau_{\mathrm{bad}}\leq t\textrm{ or }\tau_{\ell_{\infty}}^+\leq t\textrm{ or }\norm{\alpha_t}_\infty\geq 2\norm{\alpha_0}_\infty}
    \end{align*}
    for $i\geq 1$. Note that $\tau_i=\tau_{i-1}$ if $\tau_{\mathrm{bad}}\leq \tau_{i-1}$ or $\tau_{\ell_{\infty}}^+\leq \tau_{i-1}$.
    \Cref{eq:PhaseI_1} implies that
    \begin{align*}
        \Pr\sbra*{\tau_i-\tau_{i-1}>T}=n^{-\omega(1)}
    \end{align*}
    holds for all $i\geq 1$. 
    Let $L=\log_{2}(n)$.
    Then, $\min\cbra*{\tau_{\mathrm{bad}},\tau_{\ell_{\infty}}^+}\leq \tau_{L}$ holds since $\min\cbra*{\tau_{\mathrm{bad}},\tau_{\ell_{\infty}}^+}> \tau_{L}$ contradicts the assumption of $\norm{\alpha_t}_\infty\leq 1$.
   Hence, for $H=320nk\log n\geq TL$,
    \begin{align*}
        \Pr\cbra*{\min\cbra*{\tau_{\mathrm{bad}},\tau_{\ell_{\infty}}^+}>H}
        \leq \Pr\cbra*{\tau_{L}>TL}
        \leq \Pr\cbra*{\bigvee_{i=1}^L\cbra{\tau_i-\tau_{i-1}>T}}=n^{-\omega(1)}.
    \end{align*}
\end{proof}
\begin{proof}[Proof of \cref{item:final_phase:2} of \cref{lem:PhaseAnalysis}]
Let $\tau=\min\cbra*{\tau_{\mathrm{bad}},\tau_{\ell_\infty}^{-},\taucons}$,
    $r=1-\frac{9}{400n}$,
    $X_t=r^{-t}(1-\norm{\alpha_t}_\infty)$,
    and $Y_t=X_{t\wedge \tau}$.
    Then, from \cref{eq:Linfsub},
    \begin{align*}
        \E_{t-1}\sbra*{Y_t-Y_{t-1}}
        &=\indicator_{\tau>t-1}\E_{t-1}\sbra*{X_t-X_{t-1}}\\
        &=\indicator_{\tau>t-1}r^{-t}\rbra*{\E_{t-1}\sbra*{1-\norm{\alpha_t}_\infty}-r\rbra*{1-\norm{\alpha_{t-1}}_\infty}}\\
        &\leq \indicator_{\tau>t-1}r^{-t}\left(\left(1-\|\alpha_{t-1}\|_\infty\left(1+\frac{\|\alpha_{t-1}\|_\infty(1-\|\alpha_{t-1}\|_\infty)}{16n}\right)\right)-r(1-\|\alpha_{t-1}\|_\infty)\right)\\
        &=\indicator_{\tau>t-1}r^{-t}(1-\|\alpha_{t-1}\|_\infty)\left(1-\frac{\|\alpha_{t-1}\|_\infty^2}{16n}-r\right)\\
        %&\leq \indicator_{\tau>t-1}r^{-t}(1-\|\alpha_{t-1}\|_\infty)\left(1-\frac{9}{200n}-r\right)\\
        &\leq 0,
    \end{align*}
    i.e., $(Y_t)_{t\in \mathbb{N}}$ is a supermartingale. Hence, we have
    \begin{align*}
        \E[X_{T}\mid \tau>T]\Pr[\tau>T]=\E[X_{T\wedge \tau}\mid \tau>T]\Pr[\tau>T]\leq \E[X_{T\wedge \tau}]=\E[Y_T]\leq Y_0=X_0.
    \end{align*}
    Furthermore, 
    \begin{align*}
        \E[X_{T}\mid \tau>T]
        &=r^{-T}\E[(1-\|\alpha_{T}\|_\infty)\mid \tau>T]
        \geq r^{-T}\frac{1}{n}
    \end{align*}
    holds. Note that $\taucons>T$ implies $1-\norm{\alpha_T}_\infty\geq 1/n$. 
    Thus, letting $T=100n\log n>\frac{800}{9}n\log n$, we obtain
    \begin{align}
        \Pr[\tau>T]\leq r^Tn(1-\|\alpha_0\|_\infty)\leq \exp\left(-\frac{9}{400n}T\right)n\leq n^{-1}. \label{eq:key_small_linf}
    \end{align}
Hence, from \cref{eq:key_small_linf} and \cref{item:LinfDown} of \cref{lem:LinfUD},
\begin{align*}
    &\Pr\left[\tau_{\mathrm{bad}}>T,\taucons>T\right]\\
    &=\Pr\left[\tau_{\mathrm{bad}}>T,\taucons>T,\tau_{\ell_\infty}^\uparrow>T\right]\\
    &\leq \Pr\left[\tau_{\mathrm{bad}}>T,\taucons>T,\tau_{\ell_\infty}^\uparrow>T,\tau^\downarrow_{\ell_\infty}>T\right]
    +\Pr\left[\tau_{\mathrm{bad}}>T,\tau_{\ell_\infty}^\uparrow>T,\tau^\downarrow_{\ell_\infty}\leq T\right]\\
    &\leq \Pr\left[\tau_{\mathrm{bad}}>T,\taucons>T,\tau^-_{\ell_\infty}>T\right]
    +\Pr\left[\tau^\downarrow_{\ell_\infty}\leq \min\{T,\tau_{\mathrm{bad}},\tau_{\ell_\infty}^\uparrow\}\right]\\
    &\leq n^{-1}+\exp\left(-C\frac{n\|\alpha_0\|_\infty}{\log n}\right)
\end{align*}
holds for some positive constant $C$.
Note that the assumption of $\norm{\alpha_0}_\infty\geq 2/3$ implies $\tau_{\ell_\infty}^\uparrow=\infty$ and $\tau_{\ell_\infty}^\downarrow>T \Rightarrow \tau_{\ell_\infty}^->T$.
\end{proof}

\begin{proof}[Proof of \cref{lem:AfterOneStrong}]
Let $H=320kn\log n$ and $M=100n\log n$.
Note that $\Pr\sbra*{\tau_{\mathrm{bad}}\leq H+M}=n^{-\omega(1)}$ holds from \cref{lem:weak cannot be strong general constant}.
Hence, from \cref{item:final_phase:1,item:final_phase:2} of \cref{lem:PhaseAnalysis},
\begin{align*}
    \Pr\cbra*{\taucons>H+M}
    &=\Pr\cbra*{\taucons>H+M,\tau_{\mathrm{bad}}\leq H+M}\\
    &+\Pr\cbra*{\taucons>H+M,\tau_{\mathrm{bad}}> H+M,\tau_{\ell_\infty}^+>H}\\
    &+\Pr\cbra*{\taucons>H+M,\tau_{\mathrm{bad}}> H+M,\tau_{\ell_\infty}^+\leq H}\\
    &\leq \Pr\cbra*{\tau_{\mathrm{bad}}\leq H+M}
    +\Pr\cbra*{\min\cbra*{\tau_{\mathrm{bad}},\tau_{\ell_\infty}^+}>H}\\
    &+\Pr\cbra*{\taucons>H+M,\tau_{\mathrm{bad}}> H+M,\tau_{\ell_\infty}^+\leq H}\\
    &=O(n^{-1}).
\end{align*}
\end{proof}

\section{Lower Bounds} \label{sec:lower bound on consensus time}
\subsection{3-Majority} \label{sec:lower bound 3-majority}
We prove that the consensus time is $\Omega(kn)$ with high probability if $k = o(\sqrt{n/\log n})$.
\begin{lemma} \label{lem:consensus time lower bound}
    Consider the asynchronous \ThreeMajority over $k = o(\sqrt{n/\log n})$ opinions on an $n$-vertex complete graph
    starting with the initial fractional population $\alpha_0\in[0,1]^{[k]}$ satisfies $\alpha_0(i)=\frac{1}{k}$ for all $i\in[k]$.
    Then
    we have 
    \[
    \Pr\sbra*{\taucons\geq \frac{nk}{4}}\geq 1-n^{-\omega(1)}.
    \]
\end{lemma}
\begin{proof}
    Suppose $\alpha_0(i)=\frac{1}{k}$ for all $i\in[k]$.
    From \cref{item:alphaU} of \cref{lem:alphaUD}
    and the union bound over $i\in[k]$, there exists a universal constant $C>0$ such that
    \begin{align*}
        \Pr\sbra*{\taucons\leq \frac{kn}{8}} 
        \leq \Pr\sbra*{\exists i\in[k],\tau^{\uparrow}_i \leq \frac{kn}{8}} \leq n\exp\rbra*{-\frac{Cn}{k^2}} = n^{-\omega(1)}.
    \end{align*}
    This proves the claim.
\end{proof}

\begin{proof}[Proof of \cref{thm:small k}.]
    We first prove the upper bound.
    From \cref{lem:unique strong opinion lemma}, with a probability $1-O(n^{-8})$,
        we have
        $\max_{i\neq \istar{T}}\alpha_T(i)\leq \frac{7}{8}\norm{\alpha_T}^2\leq \frac{7}{8}\norm{\alpha_T}_\infty$ within $T=O(nk\log n)$ steps.
    Starting with this configuration, from \cref{lem:AfterOneStrong}, the process reaches consensus with additional $O(nk\log n)$ steps with a probability at least $1-O(n^{-1})$.

    The lower bound directly follows from \cref{lem:consensus time lower bound} (our proof works for $k\le c\sqrt{n/\log n}$ for a sufficiently small constant $c>0$).
\end{proof}

\subsection{2-Choices}\label{sec:lower bo2}
Here we show the asynchronous \TwoChoices requires $\Omega(kn)$ steps to reach a consensus for some initial configuration if $k=o(n/\log n)$.
In an update of \TwoChoices rule,
(1) a vertex $v\in V$ randomly picks two vertices $u_1$ and $u_2$ with replacement, and
(2) if $u_1$ and $u_2$ have the same opinion, $v$ adopts it.
For the asynchronous \TwoChoices, we have
\begin{align}
    \E_{t-1}\sbra*{\alpha_{t}(i)-\alpha_{t-1}(i)}
    &=\frac{1}{n}(1-\alpha_{t-1}(i))\alpha_{t-1}(i)^2-\frac{1}{n}\alpha_{t-1}(i)\rbra*{\alphanorm_{t-1}-\alpha_{t-1}(i)^2} \nonumber \\
    &=\frac{\alpha_{t-1}(i)\rbra*{\alpha_{t-1}(i)-\alphanorm_{t-1}}}{n} \label{eq:expected value of alpha for 2Choices}
\end{align}
and
\begin{align}
    \E_{t-1}\sbra*{\rbra*{\alpha_{t}(i)-\alpha_{t-1}(i)}^2}
    &\leq \frac{1}{n^2}\Pr_{t-1}\sbra*{\alpha_{t}(i)\neq \alpha_{t-1}(i)}\nonumber \\
    &=\frac{\rbra*{1-\alpha_{t-1}(i)}\alpha_{t-1}(i)^2+\alpha_{t-1}(i)\rbra*{\alphanorm_{t-1}-\alpha_{t-1}(i)^2}}{n^2}\nonumber \\
    &\leq \frac{\alpha_{t-1}(i)\rbra*{\alpha_{t-1}(i)+\alphanorm_{t-1}}}{n^2}. \label{eq:square difference of alpha for 2Choices}
\end{align}
for any $i\in [k]$ and $t>0$.
In other words, \TwoChoices has the same expected value as \ThreeMajority ($\E_{t-1}\sbra*{\alpha_t(i)}
    = \alpha_{t-1}(i)\rbra{1+\frac{\alpha_{t-1}(i)-\alphanorm_{t-1}}{n}}$), while it has a smaller value of $\E_{t-1}\sbra{\rbra{\alpha_{t}(i)-\alpha_{t-1}(i)}^2}$.

\begin{theorem} \label{lem:consensus time lower bound_Bo2}
    Consider the asynchorous \TwoChoices over $k$ opinions on $V$ with $|V|=n$ and $k=o(n/\log n)$.
    If the initial fractional population $\alpha_0\in[0,1]^{[k]}$ satisfies $\alpha_0(i)=\frac{1}{k}$ for all $i\in[k]$, then
    we have 
    \[
    \Pr\sbra*{\taucons\geq \frac{nk}{8}}\geq 1-n^{-\omega(1)}.
    \]
\end{theorem}
\begin{comment}
Taking $k=n/\log^2 n$, \cref{lem:consensus time lower bound_Bo2} implies the following corollary.
\begin{corollary} \label{coro:consensus time lower bound_Bo2}
    For the asynchronous \TwoChoices, there is an initial configuration of opinions such that $\taucons=\tilde{\Omega}(n^2)$ w.h.p.
\end{corollary}
\end{comment}
\begin{proof}
Let $\tau=\inf\{t\geq 0:\|\alpha_t\|_\infty\geq 2/k\}$.
Obviously, $\taucons \geq \tau$ holds almost surely.
If $\tau>t-1$,
\begin{align}
    \E_{t-1}\sbra*{\alpha_t(i)}
    = \alpha_{t-1}(i)\rbra*{1+\frac{\alpha_{t-1}(i)-\alphanorm_{t-1}}{n}}
    \leq \rbra*{1+\frac{2}{kn}}\alpha_{t-1}(i) \label{eq:2choices lower 1}
\end{align}
holds from \cref{eq:expected value of alpha for 2Choices}.
%
%Let $\tau_{\ell_2}^\uparrow=\inf\{t\geq 0:\alphanorm_t\geq 2\alphanorm_0\}$, and $\tau=\min\{\tau_i^\uparrow,\tau_{\ell_2}^\uparrow\}$.
Let $X_t=-\alpha_t(i)$.
Then, we have
\begin{itemize}
    \item $\indicator_{\tau>t-1}\E_{t-1}\sbra*{X_t-\rbra*{1+\frac{2}{kn}}X_{t-1}}\geq 0$ (from \cref{eq:2choices lower 1}),
    \item $\indicator_{\tau>t-1}\abs{X_t-X_{t-1}}\leq \frac{1}{n}$,
    \item $\indicator_{\tau>t-1}\E_{t-1}\sbra*{(X_t-X_{t-1})^2}\leq \indicator_{\tau>t-1}\frac{\alpha_{t-1}(i)\rbra*{\alpha_{t-1}(i)+\alphanorm_{t-1}}}{n^2}\leq \frac{8}{k^2n^2}$ (from \cref{eq:square difference of alpha for 2Choices}),
    \item $\indicator_{\tau>t-1}\abs{X_{t-1}}\leq 2/k$.
\end{itemize}
Apply \cref{lem:multipricative_drift_Freedman} with the parameters
$T=nk/8$, $a=1+\frac{2}{kn}$, $A=\sum_{t=1}^T(1+\frac{2}{kn})^{-2t}\leq kn/2$, $B=1$, $D=1/n$, $S=\frac{8}{k^2n^2}$, $U=2/k$, and $\lambda =(\sqrt{2}-1)\alpha_0(i)$, we have
\begin{align*}
    &\Pr\sbra*{\bigvee_{t=0}^{T}\cbra*{\alpha_{t\wedge \tau}(i)\geq \left(1+\frac{2}{kn}\right)^{t\wedge \tau}\rbra*{\alpha_0(i)+\rbra*{\sqrt{2}-1}\alpha_0(i)}}}\\
    &=\Pr\sbra*{\bigvee_{t=0}^{T}\cbra*{X_{t\wedge \tau}\leq \left(1+\frac{2}{kn}\right)^{t\wedge \tau}\rbra*{X_0+\rbra*{\sqrt{2}-1}\alpha_0(i)}}}\\
    &\leq \exp\left(-Cn/k\right) &&(\textrm{for some constant $C>0$})\\
    &\leq n^{-\omega(1)} &&(\textrm{From $k=o(n/\log n)$}).
\end{align*}
Hence, we obtain
\begin{align*}
    \Pr\sbra*{\tau\leq T}
    &=\Pr\sbra*{\tau\leq \min\{T,\tau\}}\\
    &=\Pr\sbra*{\bigvee_{t=0}^{T\wedge \tau}\bigvee_{i\in [k]}\cbra*{\alpha_t(i)\geq \frac{2}{k}}}\\
    &=\Pr\sbra*{\bigvee_{t=0}^{T\wedge \tau}\bigvee_{i\in [k]}\cbra*{\alpha_{t\wedge \tau}(i)\geq 2\alpha_0(i)}}\\
    &\leq \Pr\sbra*{\bigvee_{t=0}^{T\wedge \tau}\bigvee_{i\in [k]}\cbra*{\alpha_{t\wedge \tau}(i)\geq \sqrt{2}\left(1+\frac{2}{kn}\right)^{t\wedge \tau}\alpha_0(i)}}\\
    &\leq \sum_{i\in [k]}\Pr\sbra*{\bigvee_{t=0}^{T}\cbra*{\alpha_{t\wedge \tau}(i)\geq \left(1+\frac{2}{kn}\right)^{t\wedge \tau}\rbra*{\alpha_0(i)+\rbra*{\sqrt{2}-1}\alpha_0(i)}}}\\
    &\leq n^{-\omega(1)}.
\end{align*}
Note that $2\geq \sqrt{2}\left(1+\frac{2}{kn}\right)^T$ holds for $T=nk/8\leq \frac{\log (2)}{4} kn$.
\end{proof}

\newcommand{\taukv}{\tau^{\mathrm{voter}}_{\kappa}}
\newcommand{\taukmaj}{\tau^{\mathrm{3maj}}_{\kappa}}

\renewcommand{\d}{\mathbf{d}}
\newcommand{\dt}{\widetilde{\mathbf{d}}}
\renewcommand{\it}{\widetilde{i}}
\newcommand{\jt}{\widetilde{j}}
\newcommand{\rt}{\widetilde{r}}
\newcommand{\lt}{\widetilde{\ell}}

\newcommand{\bu}{\mathbf{u}}
\newcommand{\bv}{\mathbf{v}}

\newcommand{\C}{\mathcal{C}}
\newcommand{\p}{\mathbf{p}}
\newcommand{\lest}{\le_{\mathrm{st}}}
\newcommand{\gest}{\ge_{\mathrm{st}}}
\newcommand{\Add}{\mathsf{Add}}
\newcommand{\Del}{\mathsf{Del}}

\section{Starting with Many Opinions} \label{sec:many opinions}
This section is devoted to proving \cref{thm:many opinions}.
To this end, we borrow the theory of majorization \cite{MOA11}.
\paragraph*{Configuration and majorization.}
%We view the vertex set as $V=[n]=\cbra{1,\dots,n}$ and the opinion set as $\Sigma=[k]$.
%We always assume $n \ge k \ge 2$.
For $i\in[k]$, let $e_i$ be the unit vector.
That is,
\[
  (e_i)_j = \begin{cases}
    1 & \text{if }i=j,\\
    0 & \text{otherwise}.
  \end{cases}
\]
A \emph{configuration} is a vector $\c=(c_1,\dots,c_k) \in \Int^k$ such that $c_1 \ge \dots \ge c_k \ge 0$ and $\sum_{i\in[k]} c_i=n$.
Let $\C_{n,k}$ be the set of all possible configurations, i.e.,
\[
  \C_{n,k} = \cbra*{ \c = (c_1,\dots,c_k) \in \Int^k \colon \sum_{i\in[k]} c_i = n \text{ and }c_1 \ge \dots \ge c_k \ge 0}.
\]
Unlike the normalized population $\alpha_t \in [0,1]^k$,
  we always keep the non-increasing order in elements of $\c \in \C_{n,k}$.
If $n$ and $k$ are clear from the context, we omit the subscription and simply write $\C=\C_{n,k}$.

For $i\in \{0,1,\dots,k\}$, let $\c^{\le 0}=0$ and $\c^{\le i}\defeq c_1 + \dots + c_i$ be the partial sum.
Given a configuration $\c=(c_1,\dots,c_k)$, we think that vertices $1,\dots,c_1$ hold opinion $1$,
vertices $c_1+1,\dots,c_1+c_2$ hold opinion $2$, and so on.
Specifically, the opinion of a vertex $u$ is the index $i\in[k]$ such that $\c^{\le i-1} < u \le \c^{\le i}$.
\begin{definition}[Majorization] \label{def:majorization}
  For two configurations $\c,\ct \in \C$, we say $\c$ \emph{majorizes} $\ct$, denoted by $\c \succeq \ct$, if $\c^{\le i} \ge \ct^{\le i}$ for all $i=1,\dots,k$.
\end{definition}

We invoke the following useful characterization of majorization.
Recall that a matrix $P\in[0,1]^{n\times n}$ is \emph{doubly stochastic} if every row and column sum up to $1$.
\begin{lemma}[Theorem B.2 in Section 2 of \cite{MOA11}] \label{lem:doubly stochastic}
  For any two configurations $\c,\ct \in \C$, $\c \succeq \ct$ if and only if there exists a doubly stochastic matrix $P \in [0,1]^{n\times n}$ such that $\ct = \c P$.
\end{lemma}
Using this lemma, it is easy to see that concatenating preserves majorization ordering.
For two vectors $\bu=(u_1,\dots,u_m)$ and $\bv = (v_1,\dots,v_n)$, let $[\bu,\bv] \defeq (u_1,\dots,u_m,v_1,\dots,v_n)$ be the vector obtained by concatenating $\bu$ and $\bv$ as a vector.

\begin{lemma} \label{lem:concatenation}
  Let $\c_1,\ct_1 \in \C_{n_1,k_1}$ and $\c_2,\ct_2 \in \C_{n_2,k_2}$ be configurations such that $\c_i \succeq \ct_i$ for $i=1,2$.
  Let $\c \in \C_{n_1+n_2,k_1+k_2}$ (and $\ct \in \C_{n_1+n_2,k_1+k_2}$) be the configuration obtained by concatenating $\c_1$ and $\c_2$ (resp.\ $\ct_1$ and $\ct_2$) as a vector and then rearranging elements in descending order.
  Then, we have $\c \succeq \ct$.
\end{lemma}

\begin{proof}
  Let $[\c_1,\c_2] \in \Int^{n_1+n_2}$ be the vector obtained by concatenating $\c_1$ and $\c_2$ (without rearranging).
  Similarly, define $[\ct_1,\ct_2]$ as the concatenation of $\ct_1$ and $\ct_2$.
  From \cref{lem:doubly stochastic}, there are two doubly stochastic matrices $P_1,P_2$ such that $\ct_i = \c_i P_i$.
  Let $P \in[0,1]^{(n_1+n_2)\times (n_1+n_2)}$ be the doubly stochastic matrix defined by
  \[
    P \defeq \begin{bmatrix}
      P_1 & O \\
      O & P_2
    \end{bmatrix}.
  \]
  Then, we have $[\ct_1,\ct_2] = [\c_1,\c_2] P$.
  Since $\c$ and $\ct$ are obtained by rearranging $[\c_1,\c_2]$ and $[\ct_1,\ct_2]$, respectively,
  we can write $\c = [\c_1,\c_2]\,\Pi$ and $\ct=[\ct_1,\ct_2]\,\widetilde{\Pi}$ for some permutation matrices $\Pi,\widetilde{\Pi}$.
  Therefore, we have
  \[
    \ct = [\ct_1,\ct_2]\, \widetilde{\Pi} = [\c_1,\c_2] \, P \, \widetilde{\Pi} = \c \Pi^{-1} P \, \widetilde{\Pi}.
  \]
  Since $\Pi^{-1}\,P\,\widetilde{\Pi}$ is doubly stochastic, from \cref{lem:doubly stochastic}, we obtain the claim.
\end{proof}
We also use the following useful lemma.
\begin{lemma} \label{lem:Muirhead}
    Let $\c \in \C$ be a configuration and $i,j\in[k]$ be any opinions such that $c_i > c_j$.
    Then, we have $\c \succeq \c - e_i + e_j$.
\end{lemma}
\cref{lem:Muirhead} is known in the literature (see, e.g., \cite[Lemma D.1 of Section 5]{MOA11}) but we prove it here for completeness.
\begin{proof}
    Consider the doubly stochastic matrix $P \defeq \lambda I + (1-\lambda)\Pi$, where $\lambda = \frac{c_i-c_j-1}{c_i-c_j} \in [0,1]$, $I$ is the identity matrix, and $\Pi$ is the permutation matrix that swaps $i$-th and $j$-th element.
    Then, we have
    \begin{align*}
        (\c P)_{\ell} &= \begin{cases}
            c_\ell & \text{if }\ell\not\in\{i,j\},\\
            \lambda c_i + (1-\lambda)c_j & \text{if }\ell=i,\\
            \lambda c_j + (1-\lambda)c_i & \text{if }\ell=j
        \end{cases} \\
        &= (\c - e_i + e_j)_\ell.
    \end{align*}
    Therefore, from \cref{lem:doubly stochastic}, we obtain the claim.
\end{proof}

\paragraph*{Voter and 3-Majority.}
For a configuration $\c\in \C$ and unit vectors $e_i$, we define $\c - e_i$ (and $\c + e_i$) as the configuration
obtained by computing $\c - e_i$ (resp.\ $\c+e_i$) as a vector and then rearranging the components in descending order.
We promise that $c_i >0$ and $c_j < n$ whenever we consider $\c - e_i$ or $\c+e_j$.

For a configuration $\c=(c_1,\dots,c_k) \in \C$, let $\p_{\c} = (p_1,\dots,p_k) \in [0,1]^k$ be the relative population sizes defined as $p_i = \frac{c_i}{n}$.
We omit the subscript $\c$ and simply write $\p$ if $\c$ is clear from the context.
Note that $\p$ specifies a distribution over $[k]$ as $\sum_{i\in[k]}p_i = 1$ and $p_i \ge 0$.
By $i\sim \p$ we mean that an opinion $i\in[k]$ is chosen according to the distribution $\p$.

For a configuration $\c$, let $\Pull(\c)$ be a $\C$-valued random variable obtained by simulating one step of the asynchronous Pull Voting on the configuration $\c$.
That is, $\Pull(\c)$ is the random variable obtained by choosing $i,j\sim \p_{\c}$ and then
outputting $\c - e_i + e_j$, an event of probability $(c_i/n)(c_j/n)$.
In this case, we call $i$ a \emph{loosing opinion} and $j$ a \emph{gaining opinion}.
Note that the components of $\Pull(\c)$ are rearranged in descending order.
The same value of $\Pull(\c)$ can arise in several ways.
For example, it can be the case that $\c-e_1+e_3 = \c-e_2+e_4$, which occurs if $c_1=c_2$ and $c_3=c_4$.
\begin{remark} \label{rem:opinion size is important}
Generally, for any fixed $i,j\in[k]$, we have
\[
\Pr\sbra*{\Pull(\c) = \c - e_i + e_j} = \frac{c_i k_i}{n} \cdot \frac{c_j k_j}{n},
\]
where $k_i$ denotes the number of opinions of size equal to $c_i$ (in the configuration $\c$).
Therefore, to sample the distribution $\Pull(\c)$ for given $\c$, it suffices to determine the \emph{sizes} of loosing opinion $i$ and gaining opinion $j$.
It does not matter specifically which opinion with the selected size is chosen.
%For example, for randomly chosen $i,j\sim \p_{\c}$, let $a,b \in[k]$ be the minimum index such that $c_{a} = c_i$ and $c_b = c_j$.
%Then, the marginal distributions of $\c - e_i + e_j$ and $\c - \e_a + \e_b$ are both $\Pull(\c)$.
\end{remark}

Similarly, let $\ThreeMaj(\c)$ be the random variable obtained by simulating one step of the 3-Majority starting from $\c$.
To state it more formally,
for $\c \in \C$, let $f\colon \C \to \C$ be a function defined by $f(\c) = (f_1(\c),\dots,f_k(\c))$, where
\begin{align}
  f_i(\c) \defeq c_i\rbra*{ 1 + \frac{c_i}{n} - \frac{\norm{\c}^2}{n^2} } = np_i\rbra*{1 + p_i - \norm{\p}^2}
   \label{eq:3MajFunc}
\end{align}
and $\norm{\cdot}$ denotes the $\ell^2$ norm (cf.\ \cref{item:expectation of alpha} of \cref{lem:basic inequalities}).
It is not hard to see that $f(\c) \in \C$ if $\c \in \C$.
Then, $\ThreeMaj(\c)$ is defined as a $\C$-valued random variable obtained by
choosing $i \sim \p_\c, j \sim \frac{f(\c)}{n}$ and then outputting $\c - e_i + e_j$.
Note that
$\Pr\sbra*{\ThreeMaj(\c) = \c - e_i + e_j} = \frac{k_ic_i}{n}\cdot \frac{k_jf_j(\c)}{n}$, where $k_i$ denotes the number of opinions of size equal to $c_i$.
\begin{lemma}[Lemma 3.2 of \cite{ignore_or_comply}] \label{lem:3maj f}
  Let $f$ be the function defined in \cref{eq:3MajFunc}.
  Then, for any configuration $\c$, we have $f(\c) \succeq \c$.
\end{lemma}

\paragraph*{Stochastic majorization.}
  Let $X,Y$ be $\C$-valued random variables.
  We say that $Y$ \emph{stochastically majorizes} $X$, denoted by $X \lest Y$, if
    there exists a coupling $D$ of $X$ and $Y$ such that $X \preceq Y$ with probability $1$ over $D$.
  It is known that this relation is transitive.
  \begin{lemma} \label{lem:lest transitive}
    If $X\lest Y$ and $Y\lest Z$, then $X \lest Z$.
  \end{lemma}
  For completeness, we include a proof here.
  \begin{proof}
    Let $D_{XY}$ and $D_{YZ}$ be the couplings that certify $X\lest Y$ and $Y\lest Z$, respectively.
    Consider the following coupling of $X$ and $Z$:
      Sample $y\sim Y$ and then $(x,y) \sim D_{XY}|_{Y=y}$, and $(y,z) \sim D_{YZ}|_{Y=y}$. Output $(x,z)$.
      Since $\Pr_{(x,y)\sim D_{XY}}\sbra*{x \preceq y} = 1$ and $\Pr_{(y,z)\sim D_{YZ}}\sbra*{y \preceq z}=1$, we have $x\preceq z$.
    This proves the claim.
  \end{proof}

  To prove \cref{lem:coupling}, we present a one-step coupling of $\Pull$ and $\ThreeMaj$ that preserves the majorization order.
  \begin{lemma} \label{lem:pull and 3-majority}
    For any configurations $\c\succeq \ct$,
    we have $\Pull(\ct) \lest \ThreeMaj(\c)$.
  \end{lemma}

  \subsection{Step 1. Compare Voter and 3-Majority}
  We first show that $\ThreeMaj$ majorizes $\Pull$ if they start from the same initial configuration.
  \begin{lemma} \label{lem:pull and 3maj}
    For any configuration $\c\in \C$,
    we have $\Pull(\c)\lest \ThreeMaj(\c)$.
  \end{lemma}
  \begin{proof}
  Consider the following coupling:
  \begin{enumerate}
    \item Draw $i\sim \p_\c$ (which will be the losing class) and $X\sim[0,1]$ uniformly at random.
    \item Let $j\in[k]$ be the index such that $\frac{1}{n}\sum_{\ell=1}^{j-1} c_\ell < X \le \frac{1}{n}\sum_{\ell=1}^j c_\ell$.
    \item Similarly, let $j'\in[k]$ be the index such that $\frac{1}{n}\sum_{\ell=1}^{j'-1} f(c_\ell) < X \le \frac{1}{n}\sum_{\ell=1}^{j'} f(c_\ell)$.\\
    If $c_j = c_{j'}$, redefine $j' \defeq j$.
    \item Output $\c_{p}\defeq \c - e_i + e_j$ as $\Pull(\c)$ and $\c_{maj} \defeq \c - e_i + e_{j'}$ as $\ThreeMaj(\c)$.
  \end{enumerate}
 Note that the marginal distribution of $\c_p$ (and $\c_{maj}$) is $\Pull(\c)$ (resp.\ $\ThreeMaj(\c)$) since the distributions of the losing opinion and gaining opinion are identical to that of $\Pull(\c)$ (resp.\ $\ThreeMaj(\c)$); see \cref{rem:opinion size is important}.

  We prove the correctness, that is, $\c_p \preceq \c_{maj}$ with probability $1$.
  Since $f(\c) \succeq \c$ (from \cref{lem:3maj f}), we have
  \[
    \frac{1}{n}\sum_{\ell=1}^j c_\ell \ge X > \frac{1}{n}\sum_{\ell=1}^{j'-1} f(c_\ell) \ge \frac{1}{n}\sum_{\ell=1}^{j'-1} c_\ell
  \]
  and $j\ge j'$; thus $c_{j} \le c_{j'}$.
  If $c_{j'}=c_j$, then $\c_p = \c_{maj}$ since $j=j'$.
  If $c_{j'}>c_j$, then $\c_p$ can be obtained from $\c_{maj}$ by decreasing an opinion of size $\c_{j'}$ by losing $j'$ and gaining $j$, i.e., $\c_p = \c_{maj} - e_{j'} + e_j$.
  Therefore, from \cref{lem:Muirhead}, we have $\c_{maj} \succeq \c_p$.
\end{proof}

  \subsection{Step 2. Decomposition of \texorpdfstring{$\Pull$}{Voter}}
  We prove that
  $\Pull$ preserves the majorization order.
  \begin{lemma} \label{lem:pull vs pull}
    For any configurations $\c\succeq \ct$, we have $\Pull(\c) \gest \Pull(\ct)$.
  \end{lemma}
  To prove \cref{lem:pull vs pull},
  recall that $\Pull(\c)$ can be obtained by first choosing $u,v\sim V$ independently uniformly at random and then replacing the opinion of $u$ with the opinion of $v$.
  If $u=v$, then $\c = \Pull(\c)$.
  Conditioned on $u\neq v$,
  $\Pull(\c)$ can be seen as a randomized algorithm that consists of the following two parts.
  \begin{enumerate}
  \item Choose $u\sim V$ uniformly at random and delete $u$.
    Let $\c'=\Del(\c) \defeq \c-e_i$ be the configuration after the removal, where $i\in[k]$ is the opinion of $u$.
  \item Choose $v\sim V\setminus\{u\}$ and let $j$ be its opinion.
  Add another vertex with opinion $j$.
  Let $\c''=\Add(\c') \defeq \c'+e_j$.
  \end{enumerate}
  Let $\Pull'(\c)$ be the random variable $\Pull(\c)$ conditioned on $u \neq v$.
  Then, we have $\Pull'(\c) = \Add(\Del(\c))$.

  The most technical part is to show that both $\Del$ and $\Add$ preserve the majorization order.
  \begin{lemma} \label{lem:del preserve}
    For any configurations $\c\succeq \ct$, we have $\Del(\c)\gest \Del(\ct)$.
  \end{lemma}
  \begin{lemma} \label{lem:add preserve}
    For any configurations $\c\succeq \ct$, we have $\Add(\c)\gest \Add(\ct)$.
  \end{lemma}

  \begin{proof}[Proof of \cref{lem:pull vs pull} using \cref{lem:del preserve,lem:add preserve}.]
    From \cref{lem:del preserve,lem:add preserve}, we have $\Pull'(\c) \gest \Pull'(\ct)$.
    To see this, we can couple $\Pull'(\c)$ and $\Pull'(\ct)$ by first simulating the coupling of \cref{lem:del preserve} to obtain $\c'=\Del(\c)$ and $\ct'=\Del(\ct)$ and then simulating the coupling of \cref{lem:add preserve} to obtain $\c''=\Add(\c')$ and $\ct''=\Add(\ct')$.
    Then, we have $\c'\succeq \ct'$ and thus $\c'' \succeq \ct''$.
    Note that the marginal distributions of $\c''$ and $\ct''$ are $\Pull'(\c)$ and $\Pull'(\ct)$, respectively.
    In other words, $\Pull'(\c) \gest \Pull'(\ct)$.

    Using this coupling,
    we can couple $\Pull(\c)$ and $\Pull(\ct)$ as follows:
    With probability $1/n$, output $\c$ and $\ct$ as $\Pull(\c)$ and $\Pull(\ct)$, respectively.
    With probability $1-1/n$, run the coupling and output $\c''$ and $\ct''$.
    Let $\d,\dt$ be the resulting configurations.
    Since $\Pull(\c)$ is obtained by outputting $\c$ with probability $1/n$ and $\Pull'(\c)$ with probability $1-1/n$, the marginal distribution of $\d$ is $\Pull(\c)$ (and so does $\Pull(\ct)$).
    Moreover, since $\Pull'(\c) \succeq \Pull'(\ct)$, we have $\d \succeq \dt$.
    This proves the claim.
  \end{proof}

\subsection{Step 3. Deletion and Addition Preserves Majorization}
In this subsection, we prove \cref{lem:del preserve}.
The proof of \cref{lem:add preserve} is symmetry; We defer the proof in \cref{sec:addition proof} for completeness.

\paragraph*{Block structure.}
To prove \cref{lem:add preserve,lem:del preserve}, we introduce the notion of \emph{block structure}.
Given configurations $\c \succeq \ct$, we can arrange them into a $2\times k$ matrix $M(\c,\ct)$ as
\begin{align*}
  M = M(\c,\ct) = \begin{bmatrix}
    c_1 & c_2 & \dots & c_k \\
    \widetilde{c}_1 & \widetilde{c}_2 & \dots & \widetilde{c}_k
  \end{bmatrix}.
\end{align*}
Let $(S_1,S_2,\ldots,S_\ell)$ be a partition of the columns of $M$ into  sub-matrices $S_i$. We call these $S_i$ {\em minimal blocks}, if each $S$ has the property that the sum of the first and second rows of $S$ are equal and $S$ contains no smaller sub-block with this property. For example let
\[
S_1=\begin{bmatrix}
c_1&c_2&\cdots&c_i\\
\tc_1&\tc_2&\cdots&\tc_i
\end{bmatrix}.
\]
Here $S_1$ is a block of length $i>1$,
where $\c^{\le 1} > \ct^{\le 1}, \c^{\le 2} > \ct^{\le 2},...,\c^{\le i-1}> \ct^{\le i-1}$ but $\c^{\le i}=\ct^{\le i}$ so that both rows add up to the same number.
That is to say $c_1> \tc_1$, $c_1+c_2> \tc_1+\tc_2$ etc., but $c_1+\cdots c_i=\tc_1+ \cdots \tc_i$.

As an example take $\c=(5,4,3,2,2,0)$ and $\ct=(4,4,4,2,1,1)$.
Using $|$ as a block delimiter, the block structure is
\[
\left|\begin{array}{ccc}
5&4&3\\
4&4&4
\end{array}
\right|
\left|\begin{array}{c}
2\\2
\end{array}
\right|
\left|\begin{array}{cc}
2&0\\1&1
\end{array}
\right|.
\]
The partial sum structure of the first block is $(5>4, 9>8, 12=12)$, and the last is $(2>1,2=2)$.

An important observation is that, for any vertex $u\in V$,
the opinion $i$ of $u$ at configuration $\c$
and the opinion $j$ of $u$ at configuration $\ct$
belong to the same block.

\begin{proof}[Proof of \cref{lem:del preserve}.]
  Let $\c\succeq \ct$ be the initial configurations.
  We first prove \cref{lem:del preserve} in the special case that $M(\c,\ct)$ consists of a single minimal block.
Then, we reduce the general multiple block case to the single block case.

  \paragraph*{Case I. Single block case.}
  Suppose the matrix $M(\c,\ct)$ consists of a single minimal block.
  That is, $\c^{\le i} > \ct^{\le i}$ for all $i=1,\dots,k-1$ and $\c^{\le k} = \ct^{\le k} = n$.
  Consider the following coupling:
  \begin{enumerate}
  \item Choose a uniformly random vertex $u\sim V$ and let $i,\it$ be the opinions of $u$ in the configuration $\c,\ct$, respectively.
  \item Output $\c'\defeq \c-e_i$ and $\ct'\defeq \ct-e_{\it}$.
  \end{enumerate}
  Note that $\c'$ is the configuration obtained by first computing $\c-e_i$ as a vector and then rearranging elements in descending order.
  Therefore, $\c'$ is equal (as a vector) to $\c - e_r$,
    where $r\in[k]$ is the right-most opinion whose value is equal to the value of $u$'s opinion in the configuration $\c$ (this preserves the descending order).
  Similarly, let $\rt$ be the rightmost opinion in $\ct$ so that $\ct'=\c - e_{\rt}$.

  Compare the partial sums of $\c'$ and $\ct'$.
  For any $\ell$, we have
  \begin{align*}
    \c'^{\le \ell} &\ge \c^{\le \ell} - 1 \\
    &\ge (\ct^{\le \ell} - 1)+1 & & \text{$\because$$M(\c,\ct)$ consists of a single block} \\
    &\ge \ct'^{\le \ell} & & \text{$\because$Deletion does not increase the partial sum}
  \end{align*}
  and thus $\c'\succeq \ct'$.

  \paragraph*{Case II. Multiple block case.}
  Suppose the matrix $M(\c,\ct)$ consists of two or more blocks $S_1,\dots,S_\ell$ (for $\ell \ge 2$).
  Then, $\c$ and $\ct$ can be seen as the concatenation of subvectors $\c[S_1],\dots,\c[S_\ell]$ and $\ct[S_1],\dots,\ct[S_\ell]$, where  $\c[S_b],\ct[S_b]$ are configurations.
  Note that $\c[S_b]\succeq \ct[S_b]$ since $\c \succeq \ct$.

  Consider the following coupling:
  \begin{enumerate}
  \item Choose a block index $b\in[\ell]$ randomly, where the probability of choosing a specific block $S_a$ is in proportional to the sum of $c_i$s belonging to $S_a$.
  \item Simulate the coupling of Case I on configurations $\c[S_b]$ and $\ct[S_b]$. This yields configurations $\c'[S_b]$ and $\ct'[S_b]$ such that $\c'[S_b] \succeq \ct'[S_b]$.
  \item Output the concatenations 
  \begin{align*}
  &\c'\defeq \sbra*{\c[S_1],\dots,\c[S_{b-1}], \c'[S_b], \c[S_{b+1}],\dots,\c[S_\ell] }, \\
  &\ct' \defeq \sbra*{\ct[S_1],\dots,\ct[S_{b-1}], \ct'[S_b], \ct[S_{b+1}],\dots,\ct[S_\ell] }.
  \end{align*}
  (note that we replace $b$-th subvector by the vector of Step 2.)
  \end{enumerate}

As an example take $\c=(4,4,2,2,2)$ and $\ct=(3,3,3,3,2)$.
Using $|$ as a block delimiter, the block structure is
\[
\left|\begin{array}{cccc}
4&4&2&2\\
3&3&3&3
\end{array}
\right|
\left|\begin{array}{c}
2\\2
\end{array}
\right|.
\]
In Step 1, the block $S_1$ is chosen with probability $\frac{12}{14}$ and $S_2$ is chosen with probability $\frac{2}{14}$.
Suppose the block $S_1$ is selected and we apply the coupling of the single block case to this block.
We choose a uniformly random vertex $u$ from this block and choose the right most opinion among those whose size is equal to the size of $u$'s opinion.
If the selected vertex has  label 9 (or any entry in $\{ 9,10,11,12\}$)  then it has opinion $2$ in the first row and $3$ in the second row.
After applying the coupling within the block, we end with the configuration
\[
\left|\begin{array}{cccc}
4&4&2&1\\
3&3&3&2
\end{array}
\right|
\left|\begin{array}{c}
2\\2
\end{array}
\right|.
\]
Then, we have $\c'[S_1]=(4,\,4,\,2,\,1) \succeq \ct'(3,\,3,\,3,\,2) = \ct'[S_1]$ and $\c'[S_2]=\ct'[S_2]=(2)$.
Since taking concatenation preserves the majorization ordering (\cref{lem:concatenation}), we have $\c'=(4,\,4,\,2,\,1,\,2) \succeq \ct'=(3,\,3,\,3,\,2,\,2)$.
Finally on permuting the entries into descending order we have $\c''=(4,\,4,\,2,\,2,\,1) \succeq \ct''=(3,\,3,\,3,\,2,\,2)$.

More formally,
  note that the probability that this coupling changes $b$-th block $S_b$ is equal to to probability that a uniformly random vertex $u\sim V$ holds an opinion in the block $S_b$.
  Therefore, $(\c',\ct')$ is a coupling of $\Del(\c)$ and $\Del(\ct)$.
  Since each corresponding subvectors satisfy the majorization (i.e., $\c'[S_a] \succeq \ct'[S_a]$ for all $a\in[\ell]$),
  from \cref{lem:concatenation}, we have $\c'\succeq \ct'$.
\end{proof}

Now we can prove \cref{lem:pull and 3-majority}.
\begin{proof}[Proof of \cref{lem:pull and 3-majority}.]
  Let $\c\succeq \ct$ be any configurations.
  From \cref{lem:pull and 3maj,lem:pull vs pull}, we have $\ThreeMaj(\c) \gest \Pull(\c) \gest \Pull(\ct)$.
  From \cref{lem:lest transitive}, we obtain the claim.
\end{proof}

\subsection{Remaining Opinions in Voter} \label{sec:remaining opinions in Voter}
In this subsection, we prove the number of remaining opinions in \Voter is likely to decrease to $\kappa$ within $O(n^2\log n/\kappa)$ steps.
\begin{lemma} \label{lem:voter}  
  There exists a universal constant $C>0$ such that
for any $\kappa\ge 1$, the number of remaining opinions in the asynchronous \Voter is at most $\kappa$ within $Cn^2\log n/\kappa$ steps with high probability.
\end{lemma}
\Cref{lem:voter} is a direct consequence of Section 14.3 of \cite{AF02}.
Here, we give a proof of \cref{lem:voter} for completeness. 

First, consider the coalescing random walk process (\textrm{CRW}) defined as follows.
%Initially, there is one particle at each vertex.
At each time step, two vertices $v\in V$ and $u\in V$ are chosen uniformly at random with replacement.
If there is a particle on $v$, then the particle moves to $u$.
If two or more particles are on the same vertex, they coalesce into a cluster, which moves as a single particle.

For any $C\subseteq V$ and $u\in V$, the following relation between \Voter and \textrm{CRW} is well known (Section 14.3 of \cite{AF02}):
\begin{align*}
    &\Pr\rbra*{\textrm{In \Voter, at time $t$, all vertex $v\in C$ have the same opinion initially held by $u$}}\\
    &=\Pr\rbra*{\textrm{In \textrm{CRW}, all particles initially on $C$ have coalesced before $t$ and their cluster is on $u$ at $t$}}.
\end{align*}

Let $\taukv$ denote the time taken for the asynchronous \Voter to reach a configuration in which the number of remaining opinions is at most $\kappa$.
Similarly,
let $\tau_\kappa^{\mathrm{crw}}$ denote the time taken for \textrm{CRW} to reach a configuration where there are $\kappa$ clusters.
From the above relation, we have 
\begin{align*}
    \Pr(\taukv \leq t)=\Pr(\tau_\kappa^{\mathrm{crw}}\leq t)
\end{align*}
for any $\kappa\in [n]$ and $t$.
Furthermore, since
$\tau_{i-1}^{\mathrm{crw}}-\tau_i^{\mathrm{crw}}$ follows the geometric distribution with success probability $\frac{(i-1)i}{n^2}$ (Section 14.3.3 of \cite{AF02}), 
we have
\begin{align*}
    \E[\tau_\kappa^{\mathrm{crw}}]
    =\sum_{i=\kappa+1}^n\E[\tau_{i-1}^{\mathrm{crw}}-\tau_i^{\mathrm{crw}}]
    =\sum_{i=\kappa+1}^n\frac{n^2}{(i-1)i}
    =n^2\left(\frac{1}{\kappa}-\frac{1}{n}\right)\leq \frac{n^2}{\kappa},
\end{align*}
i.e., $\Pr(\tau_\kappa^{\mathrm{crw}}\geq \mathrm{e}n^2/\kappa)\leq 1/\mathrm{e}$ from the Markov inequality.
Combining all, we obtain
\begin{align*}
     \Pr\sbra*{\taukv \geq \frac{\mathrm{e}n^2}{\kappa}h\log n}
     &=\Pr\sbra*{\tau_\kappa^{\mathrm{crw}} \geq \frac{\mathrm{e}n^2}{\kappa}h\log n}
     \leq \frac{1}{n^{h}}
\end{align*}
for any positive integer $h$.
\qed

Now, we are ready to prove \cref{lem:coupling,thm:many opinions}.
\begin{proof}[Proof of \cref{lem:coupling}.]
  Let $\c$ be the initial configuration and $\c^{\mathrm{pull}}_t,\c^{\mathrm{maj}}_t$ be the configuration of the asynchronous \Voter and \ThreeMajority at the $t$-th time step.
  By applying \cref{lem:pull and 3-majority} for $t=1,2,\dots$, we can couple $(\c^{\mathrm{maj}}_t)_{t\ge 0}$ and $(\c^{\mathrm{pull}}_t)_{t\ge 0}$ such that $\c^{\mathrm{pull}}_t \succeq \c^{\mathrm{maj}}_t$ for all $t\ge 0$.
  If $\c^{\mathrm{pull}}_t$ contains at most $\kappa$ remaining opinions, then so does $\c^{\mathrm{maj}}_t$ since $\c^{\mathrm{maj}}_t \succeq \c^{\mathrm{pull}}_t$.
  This proves the claim.
\end{proof}
\begin{proof}[Proof of \cref{thm:many opinions}.]
  This directly follows from \cref{lem:coupling,lem:voter}.
\end{proof}

\printbibliography

@ARTICLE{Condon2020,
  title     = "Approximate majority analyses using tri-molecular chemical reaction networks",
  author    = "Condon, Anne and Hajiaghayi, Monir and Kirkpatrick, David and Maňuch, Ján",
  journal   = "Natural Computing",
  publisher = "Springer",
  volume    =  19,
  pages     = "249--270",
  month     =  aug,
  year      =  2020,
  doi       =  "10.1007/s11047-019-09756-4"
}

@INPROCEEDINGS{Angluin2007,
   title     = "A Simple Population Protocol for Fast Robust Approximate
                Majority",
   booktitle = "Distributed Computing",
   author    = "Angluin, Dana and Aspnes, James and Eisenstat, David",
   publisher = "Springer Berlin Heidelberg",
   pages     = "20--32",
   year      =  2007
 }

@ARTICLE{randomized_gossip_algorithms_Boyd,
  title     = "Randomized gossip algorithms",
  author    = "Boyd, S and Ghosh, A and Prabhakar, B and Shah, D",
  journal   = "IEEE Trans. Inf. Theory",
  publisher = "IEEE",
  volume    =  52,
  number    =  6,
  pages     = "2508--2530",
  month     =  jun,
  year      =  2006,
  issn      = "0018-9448, 1557-9654",
  doi       = "10.1109/TIT.2006.874516"
}

@INPROCEEDINGS{undecided_chemical,
  title     = "Simplifying Analyses of Chemical Reaction Networks for
               Approximate Majority",
  booktitle = "{DNA} Computing and Molecular Programming",
  author    = "Condon, Anne and Hajiaghayi, Monir and Kirkpatrick, David and
               Ma{\v n}uch, J{\'a}n",
  publisher = "Springer International Publishing",
  pages     = "188--209",
  year      =  2017,
  doi       = "10.1007/978-3-319-66799-7\_13"
}

@ARTICLE{Becchetti2015,
   title     = "Plurality consensus in the gossip model",
    journal = {In Proceedings of the 26th annual ACM-SIAM Symposium on Discrete Algorithms (SODA)},
   author    = "Becchetti, Luca and Clementi, Andrea and Natale, Emanuele and Pasquale, Francesco and
                Silvestri, Riccardo",
   pages     = "371--390",
   year      =  2015,
 }

@BOOK{Ver18,
  title     = "{High-Dimensional} Probability: An Introduction with
               Applications in Data Science",
  author    = "Vershynin, Roman",
  publisher = "Cambridge University Press",
  month     =  sep,
  year      =  2018,
  doi       = "10.1017/9781108231596"
}

@ARTICLE{Fre75,
  title     = "On Tail Probabilities for Martingales",
  author    = "Freedman, David A",
  journal   = "aop",
  publisher = "Institute of Mathematical Statistics",
  volume    =  3,
  number    =  1,
  pages     = "100--118",
  month     =  feb,
  year      =  1975,
  language  = "en",
  issn      = "0091-1798, 2168-894X",
  doi       = "10.1214/aop/1176996452"
}

@inproceedings{EFKMT17,
author = {Els\"{a}sser, Robert and Friedetzky, Tom and Kaaser, Dominik and Mallmann-Trenn, Frederik and Trinker, Horst},
title = {Brief Announcement: Rapid Asynchronous Plurality Consensus},
year = {2017},
isbn = {9781450349925},
publisher = {Association for Computing Machinery},
address = {New York, NY, USA},
url = {https://doi.org/10.1145/3087801.3087860},
doi = {10.1145/3087801.3087860},
booktitle = {Proceedings of the ACM Symposium on Principles of Distributed Computing},
pages = {363–365},
numpages = {3},
location = {Washington, DC, USA},
series = {PODC '17}
}

@book{Dur19,
	author = {Rick Durrett},
	publisher = {Campridge University Press},
	title = {Probability: Theory and Examples},
	year = {2019}}

@article{HP01,
	author = {Y. Hassin and D. Peleg},
	journal = {Information and Computation},
	number = {2},
	pages = {248--268},
	title = {Distributed probabilistic polling and applications to proportionate agreement},
	volume = {171},
	year = {2001}}

@ARTICLE{nearly_tight_analysis,
  title     = "{Nearly-Tight} Analysis for 2-Choice and 3-Majority Consensus
               Dynamics",
  author    = "Ghaffari, Mohsen and Lengler, Johannes",
  journal   = "Proceedings of Symposium on Principles of Distributed Computing
               (PODC)",
  publisher = "ACM",
  pages     = "305--313",
  month     =  jul,
  year      =  2018,
  address   = "New York, NY, USA",
  doi       = "10.1145/3212734.3212738"
}

@article{CNS19,
	author = {E. Cruciani and E. Natale and G. Scornavacca},
	journal = {In Proceedings of the 33rd AAAI conference on artificial intelligence (AAAI)},
	pages = {6046--6053},
	title = {Distributed community detection via metastability of the 2-choices dynamics},
	year = {2019}}

@ARTICLE{linear_voting,
  title   = "The Linear Voting Model",
  author  = "Cooper, Colin and Rivera, Nicol{\'a}s",
  journal = "Proceedings of International Colloquium on Automata, Languages,
             and Programming (ICALP)",
  year    =  2016,
  doi     = "10.4230/LIPIcs.ICALP.2016.144"
}

@ARTICLE{twochoice_expander_DISC17,
  title   = "Fast Plurality Consensus in Regular Expanders",
  author  = "Cooper, Colin and Radzik, Tomasz and Rivera, Nicol{\'a}s and
             Shiraga, Takeharu",
  journal = "Proc. Int. Symp. High Perform. Distrib. Comput.",
  year    =  2017,
  doi     = "10.4230/LIPIcs.DISC.2017.13"
}

@ARTICLE{twochoice_ICALP14,
  title   = "The Power of Two Choices in Distributed Voting",
  author  = "Cooper, Colin and Els{\"a}sser, Robert and Radzik, Tomasz",
  journal = "Proceedings of International Colloquium on Automata, Languages,
             and Programming (ICALP)",
  pages   = "435--446",
  year    =  2014,
  doi     = "10.1007/978-3-662-43951-7\_37"
}

@INPROCEEDINGS{quasi-majority,
  title     = "{Quasi-Majority} Functional Voting on Expander Graphs",
  booktitle = "47th International Colloquium on Automata, Languages, and
               Programming ({ICALP} 2020)",
  author    = "Shimizu, Nobutaka and Shiraga, Takeharu",
  publisher = "Schloss-Dagstuhl - Leibniz Zentrum f{\"u}r Informatik",
  pages     = "97:19",
  year      =  2020,
  doi       = "10.4230/LIPIcs.ICALP.2020.97"
}

@InProceedings{voter_dynamic_graph,
  author =	{Berenbrink, Petra and Giakkoupis, George and Kermarrec, Anne-Marie and Mallmann-Trenn, Frederik},
  title =	{{Bounds on the Voter Model in Dynamic Networks}},
  booktitle =	{43rd International Colloquium on Automata, Languages, and Programming (ICALP 2016)},
  pages =	{146:1--146:15},
  series =	{Leibniz International Proceedings in Informatics (LIPIcs)},
  ISBN =	{978-3-95977-013-2},
  ISSN =	{1868-8969},
  year =	{2016},
  volume =	{55},
  editor =	{Chatzigiannakis, Ioannis and Mitzenmacher, Michael and Rabani, Yuval and Sangiorgi, Davide},
  publisher =	{Schloss Dagstuhl -- Leibniz-Zentrum f{\"u}r Informatik},
  address =	{Dagstuhl, Germany},
  URL =		{https://drops.dagstuhl.de/entities/document/10.4230/LIPIcs.ICALP.2016.146},
  URN =		{urn:nbn:de:0030-drops-62901},
  doi =		{10.4230/LIPIcs.ICALP.2016.146}
}

@unpublished{AF02,
	author = {D. Aldous and J. Fill},
	note = {http://statwww.berkeley.edu/pub/users/aldous/RWG/book.html},
	title = {Reversible {M}arkov chains and random walks on graphs}}

@ARTICLE{consensus_ER_Schoenebeck18,
  title   = "Consensus of Interacting Particle Systems on {Erd{\"o}s-R{\'e}nyi}
             Graphs",
  author  = "Schoenebeck, Grant and Yu, Fang-Yi",
  journal = "Proceedings of Symposium on Discrete Algorithms (SODA)",
  pages   = "1945--1964",
  month   =  jan,
  year    =  2018,
  doi     = "10.1137/1.9781611975031.127"
}

@INPROCEEDINGS{NIY99,
  title     = "Probabilistic Local Majority Voting for the Agreement Problem on
               Finite Graphs",
  booktitle = "Computing and Combinatorics",
  author    = "Nakata, Toshio and Imahayashi, Hiroshi and Yamashita, Masafumi",
  publisher = "Springer Berlin Heidelberg",
  pages     = "330--338",
  year      =  1999,
  doi       = "10.1007/3-540-48686-0\_33"
}

@INPROCEEDINGS{hierarchy_Berenbrink,
  title     = "On the hierarchy of distributed majority protocols",
  booktitle = "Proceedings of International Conference on Principles of
               Distributed Systems ({OPODIS})",
  author    = "Berenbrink, Petra and Coja-Oghlan, Amin and Gebhard, Oliver and
               Hahn-Klimroth, Max and Kaaser, Dominik and Rau, Malin",
  publisher = "Schloss Dagstuhl - Leibniz-Zentrum f{\"u}r Informatik",
  pages     = "23:1--23:19",
  month     =  feb,
  year      =  2023,
  doi       = "10.4230/LIPICS.OPODIS.2022.23"
}

@ARTICLE{Shimizu_Shiraga_SBM,
  title   = "Phase transitions of Best‐of‐two and Best‐of‐three on stochastic
             block models",
  author  = "Shimizu, Nobutaka and Shiraga, Takeharu",
  journal = "Random Structures \& Algorithms",
  volume  =  59,
  number  =  1,
  pages   = "96--140",
  month   =  aug,
  year    =  2021,
  issn    = "1042-9832",
  doi     = "10.1002/rsa.20992"
}

@INCOLLECTION{Becchetti_minority,
  title     = "The Minority Dynamics and the Power of Synchronicity",
  booktitle = "Proceedings of the 2024 Annual {ACM-SIAM} Symposium on Discrete
               Algorithms ({SODA})",
  author    = "Becchetti, Luca and Clementi, Andrea and Pasquale, Francesco and
               Trevisan, Luca and Vacus, Robin and Ziccardi, Isabella",
  publisher = "Society for Industrial and Applied Mathematics",
  pages     = "4155--4176",
  series    = "Proceedings",
  month     =  jan,
  year      =  2024,
  doi       = "10.1137/1.9781611977912.144"
}

@BOOK{MOA11,
  title     = "Inequalities: Theory of Majorization and Its Applications",
  author    = "Marshall, Albert W and Olkin, Ingram and Arnold, Barry C",
  publisher = "Springer New York",
  doi       = "10.1007/978-0-387-68276-1"
}

@ARTICLE{drift_analysis_concentration_Kotzing,
  title   = "Concentration of First Hitting Times Under Additive Drift",
  author  = "K{\"o}tzing, Timo",
  journal = "Algorithmica",
  volume  =  75,
  number  =  3,
  pages   = "490--506",
  month   =  jul,
  year    =  2016,
  issn    = "0178-4617, 1432-0541",
  doi     = "10.1007/s00453-015-0048-0"
}

@ARTICLE{Pel02,
  title   = "Local majorities, coalitions and monopolies in graphs: a review",
  author  = "Peleg, David",
  journal = "Theor. Comput. Sci.",
  volume  =  282,
  number  =  2,
  pages   = "231--257",
  month   =  jun,
  year    =  2002,
  issn    = "0304-3975",
  doi     = "10.1016/S0304-3975(01)00055-X"
}

@ARTICLE{undecided_MFCS,
  title   = "A Tight Analysis of the Parallel {Undecided-State} Dynamics with
             Two Colors",
  author  = "Clementi, Andrea and Ghaffari, Mohsen and Gual{\`a}, Luciano and
             Natale, Emanuele and Pasquale, Francesco and Scornavacca, Giacomo",
  journal = "Proceedings of International Symposium on Mathematical Foundations
             of Computer Science (MFCS)",
  year    =  2018,
  doi     = "10.4230/LIPIcs.MFCS.2018.28"
}

@book{Lig85,
	author = {T. M. Liggett},
	publisher = {Springer-Verlag},
	title = {Interacting particle systems},
	year = {1985}}

@INPROCEEDINGS{CRN_SODA14,
  title     = "Timing in chemical reaction networks",
  booktitle = "Proceedings of the twenty-fifth annual {ACM-SIAM} symposium on
               Discrete algorithms",
  author    = "Doty, David",
  publisher = "Society for Industrial and Applied Mathematics",
  pages     = "772--784",
  series    = "SODA '14",
  month     =  jan,
  year      =  2014,
  address   = "USA",
  location  = "Portland, Oregon",
  isbn      = "9781611973389"
}

@INPROCEEDINGS{stabilizing_consensus,
  title     = "Stabilizing consensus with many opinions",
  booktitle = "Proceedings of the twenty-seventh annual {ACM-SIAM} symposium on
               Discrete algorithms",
  author    = "Becchetti, Luca and Clementi, Andrea and Natale, Emanuele and
               Pasquale, Francesco and Trevisan, Luca",
  publisher = "Society for Industrial and Applied Mathematics",
  pages     = "620--635",
  series    = "SODA '16",
  month     =  jan,
  year      =  2016,
  address   = "USA",
  location  = "Arlington, Virginia",
  isbn      = "9781611974331"
}

@INPROCEEDINGS{Doerr11,
  title     = "Stabilizing consensus with the power of two choices",
  booktitle = "Proceedings of the twenty-third annual {ACM} symposium on
               Parallelism in algorithms and architectures",
  author    = "Doerr, Benjamin and Goldberg, Leslie Ann and Minder, Lorenz and
               Sauerwald, Thomas and Scheideler, Christian",
  publisher = "Association for Computing Machinery",
  pages     = "149--158",
  series    = "SPAA '11",
  month     =  jun,
  year      =  2011,
  address   = "New York, NY, USA",
  location  = "San Jose, California, USA",
  isbn      = "9781450307437",
  doi       = "10.1145/1989493.1989516"
}

@INPROCEEDINGS{fast_convergence_undecided,
  title     = "Fast Convergence of k-Opinion Undecided State Dynamics in the
               Population Protocol Model",
  booktitle = "Proceedings of the 2023 {ACM} Symposium on Principles of
               Distributed Computing",
  author    = "Amir, Talley and Aspnes, James and Berenbrink, Petra and
               Biermeier, Felix and Hahn, Christopher and Kaaser, Dominik and
               Lazarsfeld, John",
  publisher = "Association for Computing Machinery",
  pages     = "13--23",
  series    = "PODC '23",
  month     =  jun,
  year      =  2023,
  address   = "New York, NY, USA",
  location  = "Orlando, FL, USA",
  doi       = "10.1145/3583668.3594589"
}

@INPROCEEDINGS{twochoice_expander_DISC15,
  title     = "Fast Consensus for Voting on General Expander Graphs",
  booktitle = "Distributed Computing",
  author    = "Cooper, Colin and Els{\"a}sser, Robert and Radzik, Tomasz and
               Rivera, Nicol{\'a}s and Shiraga, Takeharu",
  publisher = "Springer Berlin Heidelberg",
  pages     = "248--262",
  year      =  2015,
  doi       = "10.1007/978-3-662-48653-5\_17"
}

@article{CEO+13,
	author = {C. Cooper and R. Els{\"a}sser and H. Ono and T. Radzik},
	journal = {SIAM Journal on Discrete Mathematics},
	number = {4},
	pages = {1748--1758},
	title = {Coalescing random walks and voting on connected graphs},
	volume = {27},
	year = {2013}}

@INPROCEEDINGS{ignore_or_comply,
  title     = "Ignore or Comply? On Breaking Symmetry in Consensus",
  booktitle = "Proceedings of the {ACM} Symposium on Principles of Distributed
               Computing",
  author    = "Berenbrink, Petra and Clementi, Andrea and Els{\"a}sser, Robert
               and Kling, Peter and Mallmann-Trenn, Frederik and Natale,
               Emanuele",
  publisher = "Association for Computing Machinery",
  pages     = "335--344",
  series    = "PODC '17",
  month     =  jul,
  year      =  2017,
  address   = "New York, NY, USA",
  location  = "Washington, DC, USA",
  isbn      = "9781450349925",
  doi       = "10.1145/3087801.3087817"
}

@ARTICLE{consensus_dynamics_SIGACT,
  title   = "Consensus Dynamics",
  author  = "Becchetti, Luca and Clementi, Andrea and Natale, Emanuele",
  journal = "ACM SIGACT News",
  volume  =  51,
  number  =  1,
  pages   = "58--104",
  month   =  mar,
  year    =  2020,
  issn    = "0163-5700",
  doi     = "10.1145/3388392.3388403"
}

@article{BCNPT17,
	author = {L. Becchetti and A. Clementi and E. Natale and F. Pasquale and L. Trevisan},
	journal = {In Proceedings of the 28th annual ACM-SIAM Symposium on Discrete Algorithms (SODA)},
	pages = {940--959},
	title = {Find your place: {S}imple distributed algorithms for community detection},
	year = {2017}}

@ARTICLE{simple_dynamics,
  title   = "Simple dynamics for plurality consensus",
  author  = "Becchetti, Luca and Clementi, Andrea and Natale, Emanuele and
             Pasquale, Francesco and Silvestri, Riccardo and Trevisan, Luca",
  journal = "Distrib. Comput.",
  volume  =  30,
  number  =  4,
  pages   = "293--306",
  month   =  aug,
  year    =  2017,
  issn    = "0178-2770, 1432-0452",
  doi     = "10.1007/s00446-016-0289-4"
}

@ARTICLE{Moran_process14,
  title   = "Approximating Fixation Probabilities in the Generalized Moran
             Process",
  author  = "D{\'\i}az, Josep and Goldberg, Leslie Ann and Mertzios, George B
             and Richerby, David and Serna, Maria and Spirakis, Paul G",
  journal = "Algorithmica",
  volume  =  69,
  number  =  1,
  pages   = "78--91",
  month   =  may,
  year    =  2014,
  issn    = "0178-4617, 1432-0541",
  doi     = "10.1007/s00453-012-9722-7"
}

@ARTICLE{DW83,
  title     = "Finite particle systems and infection models",
  author    = "Donnelly, Peter and Welsh, Dominic",
  journal   = "Math. Proc. Cambridge Philos. Soc.",
  publisher = "Cambridge University Press",
  volume    =  94,
  number    =  1,
  pages     = "167--182",
  month     =  jul,
  year      =  1983,
  issn      = "1469-8064, 0305-0041",
  doi       = "10.1017/S0305004100060989"
}

@INPROCEEDINGS{biology,
  title     = "Theoretical Distributed Computing Meets Biology: A Review",
  booktitle = "Distributed Computing and Internet Technology",
  author    = "Feinerman, Ofer and Korman, Amos",
  publisher = "Springer Berlin Heidelberg",
  pages     = "1--18",
  year      =  2013,
  doi       = "10.1007/978-3-642-36071-8\_1"
}
\appendix
\section{Tools} \label{sec:tools}
Recall that a sequence of random variables $(X_t)_{t\in\Nat_0}$ is \emph{supermartingale} if $\E_{t-1}\sbra*{X_t}\leq X_{t-1}$
and
is
\emph{submartingale} if $\E_{t-1}\sbra*{X_t}\ge X_{t-1}$.
\begin{theorem}[Optimal stopping theorem. See e.g.~Theorem 4.8.5 of \cite{Dur19}]\label{lem:OST}
    Let $(X_t)_{t\in \mathbb{N}_0}$ be a submartingale (resp.\ supermartingale) such that $\E_{t-1}\sbra*{|X_t-X_{t-1}|}<\infty$.
    and $\tau$ be a stopping time such that $\E[\tau]<\infty$.
    Then, $\E[X_\tau]\geq \E[X_0]$ (resp.\ $\E[X_\tau]\leq \E[X_0]$).
\end{theorem}
\begin{theorem}[Freedman's inequality. See Theorem 4.1 of \cite{Fre75}]\label{thm:Freedman}
    Let $(X_t)_{t\in\Nat_0}$ be a real-valued supermartingale
    such that $X_t - X_{t-1}\leq D$ for all $t\geq 1$.
    Then, for any $\lambda,W>0$, we have
    \begin{align*}
        \Pr\sbra*{X_T-X_0\geq \lambda \textrm{ and } \sum_{t=1}^T\E_{t-1}\sbra*{(X_t-X_{t-1})^2}\leq W \textrm{ for some $T\in\Nat_0$}}
        \leq \exp\left(-\frac{\lambda^2}{2W+2D\lambda}\right).
    \end{align*}
\end{theorem}
In the following,
    we show that if a sequence of real-valued random variables $(X_t)_{t\geq 0}$
    has a negative drift, small variance, and bounded difference,
    then $(X_t)$ is likely to decrease to $X_t\leq L$ before growing up to $X_t\geq U$.
Here, we are interested in the value of $X_\tau$, rather than the bound of $\tau$.
\begin{theorem}[Gambler's Ruin under Negative Drift]
    \label{thm:gambler}
    Let $(X_t)_{t\in \mathbb{N}_0}$ be a sequence of real-valued random variables.
    Let $0\leq L \leq U$ be parameters.
    Let $\tau=\inf\{t\geq 0: X_t\leq L \textrm{ or } X_t\geq U\}$ be the stopping time and suppose $\E[\tau]<\infty$.
    Let $S,D,\theta>0$ be constants that do not depend on $t$ and suppose that the following condition holds for all $t\geq 1$:
    \begin{enumerate}[label=(\roman*)]
        \item $\indicator_{\tau>t-1}\rbra*{\E_{t-1}\sbra*{X_t-X_{t-1}}+\theta}\leq 0$,
        \item $\indicator_{\tau>t-1}\E_{t-1}\sbra*{(X_t-X_{t-1})^2}\leq S$,
        \item $\indicator_{\tau>t-1}|X_t-X_{t-1}|\leq D$.
        %\item $\indicator_{\tau>t-1}\E\left[\left|X_t-X_{t-1}\right|\mid \mathcal{F}_{t-1}\right]\leq K$.
        %\item $\indicator_{\tau>t-1}\left|X_t-\E[X_t\mid \mathcal{F}_{t-1}]\right|\leq K$.
    \end{enumerate}
    Let $\phi=\frac{6\theta}{3S+2D\theta}$. Then, for any $L<X_0<U$, we have
    \begin{align*}
        \Pr[X_\tau\geq U]
    \leq \frac{\mathrm{e}^{\phi X_0}-\mathrm{e}^{\phi (L-D)}}{\mathrm{e}^{\phi U}-\mathrm{e}^{\phi (L-D)}}.
    \end{align*}
\end{theorem}
\begin{proof}
    Let $Y_t=X_{\tau \wedge t}$ for convenience. 
    Then, $Y_{t}-Y_{t-1}=\indicator_{\tau>t-1}(X_t-X_{t-1})$.
    Hence, we have $|Y_{t}-Y_{t-1}|\leq D$ and $\E_{t-1}[Y_{t}-Y_{t-1}]=\indicator_{\tau>t-1}(\E_{t-1}[X_{t}]-X_{t-1})\leq -\indicator_{\tau>t-1}\theta$.
    Since $0\leq \phi< 3/D$,
    from the well-known fact $\mathrm{e}^z\leq 1+z+\frac{z^2/2}{1-|z|/3}$ for any $|z|<3$ (see e.g., p.39 of \cite{Ver18})
    with $z=\phi\cdot (Y_{t}-Y_{t-1})$,
    we have
    %applying \cref{lem:Bernsteinlem} yields
    \begin{align*}
         \E_{t-1}\sbra*{\mathrm{e}^{\phi\cdot (Y_t-Y_{t-1})}}
         &\leq \exp\left(\phi \E_{t-1}[Y_t-Y_{t-1}]+\frac{\phi^2/2}{1-\phi D/3}\E_{t-1}[(Y_t-Y_{t-1})^2]\right)\\
         &\leq \exp\left(-\phi\indicator_{\tau>t-1}\theta +\indicator_{\tau>t-1}\frac{\phi^2S/2}{1-\phi D/3}\right)\\
         &= \exp\left(-\phi\indicator_{\tau>t-1}\theta+\indicator_{\tau>t-1}\phi \frac{3\theta S}{3S+2D\theta}\frac{3S+2D\theta}{3S}\right)\\
         &=1,
    \end{align*}
    i.e., $Z_t=\mathrm{e}^{\phi Y_t}=\mathrm{e}^{\phi X_{t\wedge\tau}}$ is a supermartingale. 
    Furthermore, we have $\abs{Z_t - Z_{t-1}} < \infty$.
    %Since $|Z_{t}-Z_{t-1}|\leq | \indicator_{\tau>t-1}(X_t-\E[X_t\mid \mathcal{F}_{t-1}]-r)|K+r$, $W_t\leq \mathrm{e}^{\phi (U+K+r)}$ is bounded.
    Hence, we can apply \cref{lem:OST} and  $\E\left[Z_\tau\right]\leq \E[Z_0]\leq \mathrm{e}^{\phi X_0}$ holds.
    Then, from
    \begin{align*}
    \E\left[Z_\tau\right]
    &=\E\left[Z_\tau\mid X_\tau\geq U\right]\Pr[X_\tau\geq U]+\E\left[Z_\tau\mid X_\tau\leq L\right]\Pr[X_\tau\geq L]\\
    &=\E\left[\mathrm{e}^{\phi X_\tau}\mid X_\tau\geq U\right]\Pr[X_\tau\geq U]+\E\left[\mathrm{e}^{\phi X_\tau}\mid X_\tau\leq L\right](1-\Pr[X_\tau\geq U])\\
    &\geq \mathrm{e}^{\phi U}\Pr[X_\tau\geq U]+\mathrm{e}^{\phi (L-D)}(1-\Pr[X_\tau\geq U])\\
    &=(\mathrm{e}^{\phi U}-\mathrm{e}^{\phi (L-D)})\Pr[X_\tau\geq U]+\mathrm{e}^{\phi (L-D)},
    \end{align*}
we obtain
\begin{align*}
    \Pr[X_\tau\geq U]
    \leq \frac{\mathrm{e}^{\phi X_0}-\mathrm{e}^{\phi (L-D)}}{\mathrm{e}^{\phi U}-\mathrm{e}^{\phi (L-D)}}.
\end{align*}
\end{proof}

\begin{theorem}
    \label{thm:OST_mar}
    Let $(X_t)_{t\in \mathbb{N}_0}$ be a sequence of random variables such that $\abs{X_t}<\infty$ for every $t\geq 0$ and $\tau$ be a stopping time
    such that $\E[\tau]<\infty$.
    Suppose there exists a constant $S>0$ such that, for all $t\geq 1$,
    \begin{enumerate}[label=(\roman*)]
        \item $\indicator_{\tau>t-1}\left(\Var_{t-1}[X_t]-S\right)\geq 0$,
        \item $\indicator_{\tau>t-1}\left(\E_{t-1}[X_t]^2-X_{t-1}^2\right)\geq 0$.
%        \item $\indicator_{\tau>t-1}\E_{t-1}\sbra*{\abs*{X_t^2-X_{t-1}^2}}\leq K$.
    \end{enumerate}
    Then, we have $\E[\tau]\leq \frac{\E\sbra*{X_\tau^2}-\E\sbra*{X_0^2}}{S}$.
\end{theorem}
\begin{proof}
    Let $Y_t\defeq X_t^2-St$ and $Z_t=Y_{t\wedge \tau}$.
    Then, $Z_t$ is a submartingale since
    \begin{align*}
        \E_{t-1}\left[Z_t-Z_{t-1}\right]
        &=\E_{t-1}\left[\indicator_{\tau>t-1}(Y_t-Y_{t-1})\right]\\
        &=\indicator_{\tau>t-1}\E_{t-1}\left[X_t^2-X_{t-1}^2-S\right]\\
        &=\indicator_{\tau>t-1}\left(\E_{t-1}\left[X_t^2\right]-\E_{t-1}\left[X_t\right]^2-S+\E_{t-1}\left[X_t\right]^2-X_{t-1}^2\right)\\
        &\geq 0.
    \end{align*}
    Furthermore, $\E_{t-1}[Z_t-Z_{t-1}]<\infty$.
    Hence, applying \cref{lem:OST} yields
    \begin{align*}
        \E[X_\tau^2]-S\E[\tau]
        =\E[Z_\tau]
        \geq \E[Z_0]
        =\E[X_0^2].
    \end{align*}
\end{proof}

\begin{lemma}[Multiplicative drift]
    \label{lem:multipricative_drift_Freedman}
    Let $(X_t)_{t\in\Nat_0}$ be a sequence of real-valued random variables and $\tau$ be a stopping time.
    %, $\tau^\uparrow(r) \defeq \inf\{t\geq 0: X_t\geq rX_0\}$, and $\tau=\tau^c\wedge \tau^\uparrow(r)$.
    Suppose the following holds for some positive constants $a$, $D$, $S$, $U$ and for all $t\geq 1$:
    \begin{enumerate}[label=(\roman*)]
        \item $\indicator_{\tau>t-1}\E_{t-1}\left[X_t-aX_{t-1}\right]\geq 0$,
        \item $\indicator_{\tau>t-1}|X_t-X_{t-1}|\leq D$,
        \item $\indicator_{\tau>t-1}\E_{t-1}\left[(X_t-X_{t-1})^2\right]\leq S$,
        \item $\indicator_{\tau>t-1}|X_{t-1}|\leq U$.
    \end{enumerate}
    Then, for any $\lambda,T\geq 0$,
    \begin{align*}
    \Pr\left[\bigvee_{t=0}^T\left\{X_{t\wedge \tau}\leq a^{t\wedge \tau}(X_0-\lambda)\right\}\right] 
    %\leq \exp\left(-\frac{\lambda^2/2}{SA+(\lambda+A|a-1|U)(|a-1|U+2D)}\right)
    \leq \exp\left(-\frac{\lambda^2/2}{2A\left(S+(a-1)^2U^2\right)+\lambda B\left(D+|a-1|U\right)}\right)
    \end{align*}
    holds, where $A=\sum_{t=1}^Ta^{-2t}$ and $B=\max\{1,a^{-2T}\}$.
\end{lemma}
\begin{comment}
\begin{remark}
    From the Lebesgue monotone convergence theorem,
        if $a>1$ and $h\geq 0$ are fixed and
        $T\to\infty$, we obtain
        \begin{align*}
            \Pr\sbra*{ X_t \leq a^t(X_0-h)\text{ for some $t = 0,\dots,\tau$} } \leq \exp\rbra*{ - \frac{h^2/2}{SA + (h+A\abs{a-1}U)(\abs{a-1}U + 2D)}},
        \end{align*}
        where $A=\sum_{t=1}^\infty a^{-2t} = \frac{1}{a^2-1}$.
\end{remark}
\end{comment}
\begin{proof}
    Let $Y_t=-a^{-t}X_t$ and $Z_t=Y_{t\wedge \tau}$.
    First, we observe
    \begin{align*}
        Z_t-Z_{t-1}=\indicator_{\tau>t-1}(Y_t-Y_{t-1})=a^{-t}\indicator_{\tau>t-1}\left(aX_{t-1}-X_{t}\right)
    \end{align*}
    holds for any $t\geq 1$.
    Hence, 
    \begin{align*}
        \E_{t-1}\left[Z_t-Z_{t-1}\right]
        &=a^{-t}\indicator_{\tau>t-1}\E_{t-1}\left[aX_{t-1}-X_{t}\right]
        \leq 0
    \end{align*}
    holds, i.e., $(Z_t)_{t\in \mathbb{N}_0}$ is a supermartingale.
    Furthermore, we have
    \begin{align*}
        \left|Z_t-Z_{t-1}\right|
        &=a^{-t}\indicator_{\tau>t-1}\left|X_{t-1}-X_t+(a-1) X_{t-1}\right|
        %\leq \indicator_{\tau>t-1}\left(\left|X_{t-1}-X_t\right|+\epsilon\left|X_{t-1}\right|\right)
        \leq B(D+|a-1| U)
    \end{align*}
    and
    \begin{align}
        \E_{t-1}\left[(Z_t-Z_{t-1})^2\right] 
        &=a^{-2t}\indicator_{\tau>t-1}\E_{t-1}\left[\left(X_{t-1}-X_t+(a-1) X_{t-1}\right)^2\right] \nonumber \\
        &\leq a^{-2t}\indicator_{\tau>t-1}\E_{t-1}\left[2\left(X_{t-1}-X_t\right)^2
        +2(a-1)^2X_{t-1}^2
        %+2(a-1) X_{t-1}(X_{t-1}-X_t)
        \right] \nonumber \\
        &\leq 2a^{-2t}(S+(a-1)^2U^2). \label{eq:squared_X_G}
    \end{align}
    Note that $(x+y)^2\leq 2x^2+2y^2$ holds for any $x,y$.
    For any $T\geq 0$, \cref{eq:squared_X_G} implies
    \begin{align*}
        \sum_{t=1}^T\E_{t-1}[(Z_t-Z_{t-1})^2]
        &\leq 2A(S+(a-1)^2U^2).
        %&= SA+A|a-1|U(|a-1|U+2D).
        %= \frac{S+\epsilon rX_0(\epsilon rX_0+2D)}{\epsilon^2+2\epsilon}.
    \end{align*}
    Here, $A=\sum_{t=1}^{T}a^{-2t}$. 
    Now, we apply \cref{thm:Freedman} to $(Z_t)_{t\in \mathbb{N}_0}$ 
    with $D\leftarrow B(D+|a-1|U)$ and  $W \leftarrow 2A(S+(a-1)^2U^2)$.
    Since $Z_t-Z_0=X_0-a^{-(t\wedge \tau)}X_{t\wedge \tau}$,
    \begin{align*}
        \Pr\left[\bigvee_{t=0}^T\left\{X_{t\wedge \tau}\leq a^{t\wedge \tau}(X_0-\lambda)\right\}\right] 
        &=\Pr\left[\bigvee_{t=0}^T\left\{Z_t-Z_0\geq \lambda\right\}\right] \\
        &= \Pr\left[\bigvee_{t=0}^T\left\{Z_t-Z_0\geq \lambda \textrm{ and } \sum_{i=1}^t\E_{t-1}[(Z_t-Z_{i-1})^2] \leq W \right\}\right]\\ 
        &\leq \exp\left(-\frac{\lambda^2}{2W+2Dh}\right)\\
        &\leq \exp\left(-\frac{\lambda^2/2}{2A(S+(a-1)^2U^2)+\lambda B(D+|a-1|U)}\right)%\\
        %&\leq \exp\left(-\frac{\lambda^2/2}{SA+(\lambda+A|a-1|U)(|a-1|U+2D)}\right)
    \end{align*}
    holds for any $\lambda,T\geq 0$.
\end{proof}
\section{Deferred Proofs} \label{sec:proof of basic facts}
\subsection{Proof of \texorpdfstring{\cref{lem:basic inequalities}}{Basic Inequalities}}
    Fix a time step $t\geq 1$.
    For convenience, we omit the subscript $t$ and use $\alpha$ to denote $\alpha_{t-1}$ and
    $\alpha'$ to denote $\alpha_t$.
    Let $g(i)=\alpha(i)(1+\alpha(i)-\norm{\alpha}^2)$.
    Note that \cref{item:expectation of alpha} is shown in \cref{sec:preliminaries}.
    We prove other claims.
    
    \begin{proof}[Proof of \cref{item:difference of alpha}.]
        The claim is obvious since at most one vertex changes its opinion at every step.
    \end{proof}

    \begin{proof}[Proof of \cref{item:square difference of alpha}.]
        Note that $(\alpha'(i)-\alpha(i))^2=1/n^2$ if
            the population of opinion $i$ changes.
        Therefore, we have $\E_{t-1}[(\alpha'(i)-\alpha(i))^2]= \frac{1}{n^2}\Pr_{t-1}[\alpha'(i)\neq \alpha(i)] = \frac{\alpha(i)(1-g(i)) + (1-\alpha(i))g(i)}{n^2} \leq \frac{3\alpha(i)}{n^2}$.
    \end{proof}
    
    \begin{proof}[Proof of \cref{item:expectation of delta}.]
        Apply the Jensen's inequality for the function $x\mapsto\abs{x}$.
        Then we have
        \[
        \E_{t-1}[\delta_t(i,j)] \geq \abs*{\E_{t-1}[\alpha'(i) - \alpha'(j)]} = \delta_{t-1}(i,j)\rbra*{1+\frac{\alpha(i)+\alpha(j)-\norm{\alpha}^2}{n}}.
        \]
    \end{proof}
    
    \begin{proof}[Proof of \cref{item:difference of delta}.]
        From \cref{item:difference of alpha}, we have
            $\abs{\delta_t(i,j) - \delta_{t-1}(i,j)} \leq \abs{\alpha'(i)-\alpha(i)}+\abs{\alpha'(j)-\alpha(j)}\leq \frac{2}{n}$.
    \end{proof}
    
    \begin{proof}[Proof of \cref{item:square difference of delta}.]
        Since $\Pr_{t-1}[\delta_t(i,j)\neq \delta_{t-1}(i,j)]\leq 3(\alpha(i)+\alpha(j))$, we have
            \[
            \E_{t-1}[(\delta_t(i,j)-\delta_{t-1}(i,j))^2]\leq \frac{4}{n^2}\cdot 3(\alpha(i)+\alpha(j)).
            \]
    \end{proof}
    
    \begin{proof}[Proof of \cref{item:variance of delta}.]
        By calculation, we have
        \begin{align*}
            \Var_{t-1}\sbra*{n\delta_t} &=
                g(i) + g(j) + \alpha(i) + \alpha(j)
                    - (g(i) - g(j))^2 - \delta_{t-1}(i,j)^2 \\
                &\geq \alpha(i) + \alpha(j) - 5\delta_{t-1}(i,j)^2.
        \end{align*}
        In the last inequality, note that
            $\rbra*{g(i) - g(j)}^2 = \delta_{t-1}(i,j)^2\cdot \rbra*{1+\alpha(i) + \alpha(j) - \norm{\alpha}^2}^2 \leq 4\delta_{t-1}(i,j)^2$.
    \end{proof}
    \begin{proof}[Proof of \cref{item:expectation of 2norm}.]
        By calculation, we have
    \begin{align}
        &\E_{t-1}\sbra*{\norm{\alpha'}^2-\norm{\alpha}^2} \nonumber\\
        &=\sum_{i,j\in [k]:i\neq j}\E_{t-1}\sbra*{\norm{\alpha'}^2-\norm{\alpha}^2 \condition \begin{array}{l}
             \alpha'(i)=\alpha(i)+1/n,   \\
              \alpha'(j)=\alpha(j)-1/n
        \end{array}}\cdot
            \Pr\sbra*{\begin{array}{l}
             \alpha'(i)=\alpha(i)+1/n,   \\
              \alpha'(j)=\alpha(j)-1/n
        \end{array}} \nonumber \\
        &=\sum_{i,j\in [k]:i\neq j}\rbra*{\frac{2}{n^2}+\frac{2(\alpha(i)-\alpha(j))}{n}}\alpha(j)g(i). \label{eq:ex2norm}
    \end{align}
    Since $\sum_{i\in [k]}\alpha(i)g(i)=\norm{\alpha}^2+\norm{\alpha}_3^3-\norm{\alpha}^4$, we have
    \begin{align*}
        &\sum_{i,j\in [k]:i\neq j}\alpha(j)g(i) =\sum_{i\in [k]}\alpha(i)(1-g(i)) =1-\norm{\alpha}^2-\norm{\alpha}_3^3+\norm{\alpha}^4, \\
        &\sum_{i,j\in [k]:i\neq j}(\alpha(i)-\alpha(j))\alpha(j)g(i)
        =\sum_{i\in [k]}(1-\alpha(i))\alpha(i)g(i)-\sum_{j\in [k]}\alpha(j)^2(1-g(j)) 
        =\norm{\alpha}_3^3-\norm{\alpha}^4.
    \end{align*}
    Consequently, we obtain the claim
    \begin{align*}
        \E\sbra*{\norm*{\alpha'}^2}
        %&=\|\alpha\|_2^2+\frac{2}{n^2}(1-\|\alpha\|_2^2-\|\alpha\|_3^3+\|\alpha\|_2^4) +\frac{2}{n}(\|\alpha\|_3^3-\|\alpha\|_2^4)\\
        &=\norm{\alpha}^2+\frac{2}{n}\rbra*{1-\frac{1}{n}}\rbra*{\norm{\alpha}_3^3-\norm{\alpha}^4}+\frac{2}{n^2}(1-\norm{\alpha}^2)
        \geq \norm{\alpha}^2.
    \end{align*}
    In the last inequality, we use the Cauchy--Schwarz inequality to obtain
    $\norm{\alpha}^4=\rbra*{\sum_{i\in [k]}\alpha_i^2}^2\leq \sum_{i\in [k]}\alpha_i\sum_{i\in [k]}\alpha_i^3=\norm{\alpha}_3^3$.
    \end{proof}

    \begin{proof}[Proof of \cref{item:square difference of 2norm}.]
    Similar to \cref{eq:ex2norm}, we obtain
    \begin{align*}
        \E\sbra*{\rbra*{\norm{\alpha'}^2-\norm{\alpha}^2}^2}
        &=\frac{4}{n^2}\sum_{i,j\in [k]:i\neq j}\rbra*{\frac{1}{n}+\alpha(i)-\alpha(j)}^2\alpha(j)g(i)\\
        &\leq \frac{4}{n^2}\sum_{i,j\in [k]:i\neq j}\rbra*{\frac{1}{n^2}+\alpha(i)^2+\alpha(j)^2+\frac{2(\alpha(i)-\alpha(j))}{n}}\alpha(j)g(i)\\
        &\leq \frac{4}{n^2}\rbra*{\frac{1}{n^2}+2\norm{\alpha}_3^3+\norm{\alpha}_3^3+\frac{2}{n}\norm{\alpha}_3^3}\\
        &\leq \frac{24\norm{\alpha}_3^3}{n^2}.
    \end{align*}
    This proves \cref{item:square difference of 2norm}.
    \end{proof}

    \begin{proof}[Proof of \cref{item:difference of 2norm}.]
    In each update, at most two opinion changes.
    This implies $|\alpha'(i)^2-\alpha(i)^2|\leq \frac{3\max_{j\in[k]}\alpha(j)}{n}$ for each $i\in [k]$ and thus
    \begin{align*}
        \abs*{\norm{\alpha'}^2-\norm{\alpha}^2}
        &=\abs*{\sum_{i\in [k]}\rbra*{\alpha'(i)^2-\alpha(i)^2}}
        \leq \frac{6\max_{j\in[k]}\alpha(j)}{n}.
    \end{align*}
    This proves \cref{item:difference of 2norm}.
    \end{proof}

\subsection{Proof of \texorpdfstring{\cref{lem:ai_L2norm}}{Expectation of Z}}
We prove \cref{lem:ai_L2norm}.
The following lemma approximates the expectation $\E_{t-1}[\alpharatio_t]$ by $\frac{\E_{t-1}[\alpha_t(i)]}{\E_{t-1}[\norm{\alpha_t}^2]}$.
\begin{lemma}
    \label{lem:Taylor}
    Let $M>0$ be a positive constant.
    For two random variables $X$ and $Y$ such that $Y\geq M$, we have
    \begin{align*}
        \abs*{\E\sbra*{\frac{X}{Y}}-\frac{\E\sbra*{X}}{\E\sbra*{Y}}}
        \leq \frac{1}{M\E[Y]} \rbra*{\frac{\E[X]}{\E[Y]}\Var[Y]+\sqrt{\Var[X]\Var[Y]}}.
    \end{align*}
\end{lemma}
\begin{proof}
    First, we observe for any reals $x,a$ and non-zero reals $y,b$ that 
    \begin{align*}
        \frac{x}{y}-\frac{a}{b}-\frac{x-a}{b}+\frac{a(y-b)}{b^2}=\frac{a(y-b)^2}{yb^2}-\frac{(x-a)(y-b)}{yb}
    \end{align*}
    holds.
    Setting $a=\E[X],b=\E[Y]$ and
    taking expectation, we have
    \begin{align*}
        \left|\E\left[\frac{X}{Y}\right]-\frac{\E\left[X\right]}{\E\left[Y\right]}\right|
        &=\left|\E\left[\frac{X}{Y}-\frac{\E\left[X\right]}{\E\left[Y\right]}-\frac{X-\E[X]}{\E\left[Y\right]}+\frac{\E[X](Y-\E[Y])}{\E[Y]^2}\right]\right|\\
        &=\left|\E\left[\frac{\E[X](Y-\E[Y])^2}{Y\E[Y]^2}-\frac{(X-\E[X])(Y-\E[Y])}{Y\E[Y]}\right]\right|\\
        &\leq \frac{\E[X]\E[|Y-\E[Y]|^2]}{M\E[Y]^2}+\frac{\E[|X-\E[X]||Y-\E[Y]|]}{M\E[Y]}\\
        &\leq \frac{\E[X]\Var[Y]}{M\E[Y]^2}+\frac{\sqrt{\Var[X]\Var[Y]}}{M\E[Y]}.
    \end{align*}
\end{proof}
In the following, we bound $\E_{t-1}[(\alpharatio_t-\alpharatio_{t-1})^2]$ in terms variance.
\begin{lemma}
    \label{lem:First_dif}
    Let $X$ and $Y$ be non-negative random variables such that $Y\geq M$ for a positive constant $M$.
    Then, we have
    \begin{align*}
        \E\sbra*{\rbra*{\frac{X}{Y}-\frac{\E[X]}{\E[Y]}}^2}
        \leq \frac{1}{M^2}\rbra*{\sqrt{\Var[X]}+\frac{\E[X]}{\E[Y]}\sqrt{\Var[Y]}}^2.
    \end{align*}
\end{lemma}
\begin{proof}
    Let $x=\E[X]$ and $Y=\E[Y]$.
    From the Cauchy-Schwarz inequality, we immediately obtain
    \begin{align*}
        \E\left[\left(\frac{X}{Y}-\frac{x}{y}\right)^2\right]
        &= \E\left[\left(\frac{(X-x)y+x(y-Y)}{Yy}\right)^2\right]\\
        &\leq \frac{1}{M^2y^2}\left(y^2\E\left[(X-x)^2\right]+x^2\E\left[(y-Y)^2\right]+2xy\left|\E\left[(X-x)(y-Y)\right]\right|\right)\\
        &\leq \frac{1}{M^2y^2}\left(y^2\E\left[(X-x)^2\right]+x^2\E\left[(y-Y)^2\right]+2xy\sqrt{\E\left[(X-x)^2\right]\E\left[(y-Y)^2\right]}\right).
    \end{align*}
\end{proof}

\begin{proof}[Proof of \cref{lem:ai_L2norm}]
For simplicity, fix a time step $t$
    and we write $\alpha$ to denote $\alpha_{t-1}$ and
    $\alpha'$ to denote $\alpha_t$.
In this proof,
    by $\E[\cdot]$, we mean the conditional expectation $\E_{t-1}\sbra*{\cdot\condition \mathcal{E}_{t-1}}$ (we also use the similar notation for the variance $\Var[\cdot]$).
\paragraph*{Proof of \cref{lab:ail2norm1}.}
    We apply \cref{lem:Taylor} for $X=\alpha'(i)$ and $Y=\norm{\alpha'}^2\geq 1/k$.
    From \cref{item:expectation of alpha,item:expectation of 2norm,item:difference of 2norm,item:square difference of 2norm,item:square difference of alpha} of \cref{lem:basic inequalities}, we have
    \begin{align*}
        &\frac{\E\sbra*{\alpha'(i)}}{\E\sbra*{\norm{\alpha'}^2}}
        \leq \frac{\alpha(i)}{\norm{\alpha}^2}-\frac{(1-c_2)\alpha(i)}{n}
        \leq \alpharatio_{t-1} - (1-c_2)c_1\frac{\norm{\alpha}^2}{n}, & & (\text{$c_1\norm{\alpha}^2\leq \alpha(i)\leq c_2\norm{\alpha}^2$}) \\
%        &\abs*{\norm{\alpha'}^2-\norm{\alpha}^2}\leq \norm{\alpha}^2\cdot \frac{6}{n\norm{\alpha}}\leq \frac{6\norm{\alpha}^2}{\sqrt{n}},\\
        &\Var[\alpha'(i)] \leq \frac{3\alpha(i)}{n^2},\\
        &\Var[\norm{\alpha'}^2] \leq \frac{24\norm{\alpha}^3}{n^2}.
    \end{align*}
    From \cref{lem:Taylor}, we have
    \begin{align*}
        \E\sbra*{\frac{\alpha'(i)}{\norm{\alpha'}^2}} 
        - \frac{\E\sbra*{\alpha'(i)}}{\E\sbra*{\norm{\alpha'}^2 }}
        &\leq \frac{k}{\E\sbra*{\norm{\alpha'}^2 }}\rbra*{\frac{\E\sbra*{\alpha'(i)}\Var\sbra*{\norm{\alpha'}^2}}{\E\sbra*{\norm{\alpha'}^2 }}+\sqrt{\Var\sbra*{\alpha'(i)}\Var\sbra*{\norm{\alpha'}^2}}}\\
        &\leq \frac{k}{\norm{\alpha}^2}\rbra*{\frac{24\norm{\alpha}^3}{n^2}+\sqrt{\frac{72\norm{\alpha}^3}{n^2}\cdot\frac{\alpha(i)}{n^2}}}\\
        &\leq \frac{24k\norm{\alpha}}{n^2} + \frac{k}{n^2}\sqrt{\frac{72\alpha(i)}{\norm{\alpha}}}\\
        &= O\rbra*{\frac{k}{n^2}}.
    \end{align*}
    Therefore, we obtain
    \begin{align*}
        \E\sbra*{\alpharatio_t} \leq \frac{\E\sbra*{\alpha'(i)}}{\E\sbra*{\norm{\alpha'}^2 }} + O\rbra*{\frac{k}{n^2}} \leq \alpharatio_{t-1} - \frac{(1-c_2)c_1 - o(1)}{kn}.
    \end{align*}
    In the last inequality, we used $k=o(\sqrt{n})$.
    
\paragraph*{Proof of \cref{lab:ail2norm2}.}
    We apply \cref{lem:First_dif} for $X=\alpha'(i)$ and $Y=\norm{\alpha'}^2$.
    Note that we can take $M=\norm{\alpha}^2/2$ since $\norm{\alpha}^2\geq 1/k$ and thus $Y\geq \norm{\alpha}^2-2/n^2\geq \norm{\alpha}^2/2$.
    Therefore, we have
    \begin{align*}
       \E\sbra*{\rbra*{\frac{\alpha'(i)}{\norm{\alpha'}^2}-\frac{\alpha(i)}{\norm{\alpha}^2}}^2}
       &\leq \frac{4}{\norm{\alpha}^4}\rbra*{ \sqrt{\Var[\alpha'(i)]} + \frac{\E[\alpha'(i)]}{\E[\norm{\alpha'}^2]}\sqrt{\Var[\norm{\alpha'}^2]} }^2\\
       &\leq \frac{4}{\norm{\alpha}^4}\rbra*{\sqrt{\frac{3\alpha(i)}{n^2}} + \sqrt{\frac{24\norm{\alpha}^3}{n^2}}}^2 \\
       &\leq \frac{8}{\norm{\alpha}^4}\rbra*{ \frac{3\alpha(i)}{n^2} + \frac{24\norm{\alpha}^3}{n^2} } & & ((x+y)^2\leq 2x^2+2y^2)\\
%       &\leq \rbra*{\frac{\norm{\alpha}^2\sqrt{\frac{3\alpha(i)}{n^2}}+\alpha(i)\sqrt{\frac{24\norm{\alpha}^3}{n^2}}}{(1/2)\norm{\alpha}^4}}^2\\
       &\leq \frac{24}{\norm{\alpha}^2n^2} + \frac{192}{\norm{\alpha}n^2} & & (\alpha(i)\leq\norm{\alpha}^2)\\
       &= \frac{(24+o(1))k}{n^2}.
    \end{align*}

\paragraph*{Proof of \cref{lab:ail2norm3}.}
    Since $\norm{\alpha'}^2\geq(1/2)\norm{\alpha}^2$, we have
    \begin{align*}
       \abs*{\frac{\alpha'(i)}{\norm{\alpha'}^2}-\frac{\alpha(i)}{\norm{\alpha}^2}}
       &=\abs*{\frac{(\alpha'(i)-\alpha(i))\norm{\alpha}^2+\alpha(i)(\norm{\alpha}^2-\norm{\alpha'}^2)}{\norm{\alpha'}^2\norm{\alpha}^2}} \\
       &\leq \frac{\abs*{\alpha'(i)-\alpha(i)}\norm{\alpha}^2+\alpha(i)\abs*{\norm{\alpha}^2-\norm{\alpha'}^2}}{(1/2)\norm{\alpha}^4}\\
       &\leq \frac{2}{n\norm{\alpha}^2}+\frac{12\norm{\alpha}_\infty^2}{n\norm{\alpha}^4} & & \text{(\cref{item:difference of 2norm} of \cref{lem:basic inequalities})}\\
       &\leq \frac{14}{n\norm{\alpha}^2} \\
       &\leq \frac{14k}{n}.
    \end{align*}
\end{proof}

\subsection{Proof of \texorpdfstring{\cref{lem:add preserve}}{Addition Lemma}} \label{sec:addition proof}
In this subsection, we prove \cref{lem:add preserve}.
The proof is the same as the proof of \cref{lem:add preserve} except for considering addition instead of subtraction.
\begin{proof}[Proof of \cref{lem:add preserve}.]
  Let $\c\succeq \ct$ be the initial configurations.
  We first prove \cref{lem:add preserve} in the special case that $M(\c,\ct)$ consists of a single minimal block.
Then, we reduce the general multiple block case to the single block case.

  \paragraph*{Case I. Single block case.}
  Suppose the matrix $M(\c,\ct)$ consists of a single minimal block.
  Consider the following coupling:
  \begin{enumerate}
  \item Choose a uniformly random vertex $u\sim V$ and let $j,\jt$ be the opinion of $u$ in the configuration $\c,\ct$, respectively.
  \item Output $\c'\defeq \c+e_j$ and $\ct'\defeq \ct+e_{\jt}$.
  \end{enumerate}
  Note that $\c'$ is the configuration obtained by first computing $\c+e_i$ as a vector and then rearrange the elements in descending order.
  Therefore, $\c'$ is equal (as a vector) to $\c + e_\ell$,
    where $\ell\in[k]$ is the leftmost opinion whose size is equal to the size of $u$'s opinion in the configuration $\c$ (this preserves the descending order).
  Similarly, let $\lt$ be the left-most opinion in $\ct$ so that $\ct'=\c + e_{\lt}$.

  Compare the partial sums of $\c'$ and $\ct'$.
  For any $m$, we have
  \begin{align*}
    \c'^{\le m} &\ge \c^{\le m} \\
    &\ge \ct^{\le m}+1  & & \text{$\because$$M(\c,\ct)$ consists of a single block} \\
    &\ge \ct'^{\le m} & & \text{$\because$Addition increases the partial sum by at most one}
  \end{align*}
  and thus $\c'\succeq \ct'$.

  \paragraph*{Case II. Multiple block case.}
Suppose that the matrix $M(\c,\ct)$ consists of two or more blocks $S_1,\dots,S_\ell$ (for $\ell \ge 2$).
  Then, $\c$ and $\ct$ can be seen as the concatenation of subvectors $\c[S_1],\dots,\c[S_\ell]$ and $\ct[S_1],\dots,\ct[S_\ell]$, where  $\c[S_b],\ct[S_b]$ are configurations.
  Note that $\c[S_b]\succeq \ct[S_b]$ since $\c \succeq \ct$.

  Consider the following coupling:
  \begin{enumerate}
  \item Choose a block index $b\in[\ell]$ randomly, where the probability of choosing a specific block $S_a$ is proportional to the sum of $c_i$s belonging to $S_a$.
  \item Simulate the coupling of Case I on configurations $\c[S_b]$ and $\ct[S_b]$. This yields configurations $\c'[S_b]$ and $\ct'[S_b]$ such that $\c'[S_b] \succeq \ct'[S_b]$.
  \item Output the concatenations 
  \begin{align*}
  &\c'\defeq \sbra*{\c[S_1],\dots,\c[S_{b-1}], \c'[S_b], \c[S_{b+1}],\dots,\c[S_\ell] }, \\
  &\ct' \defeq \sbra*{\ct[S_1],\dots,\ct[S_{b-1}], \ct'[S_b], \ct[S_{b+1}],\dots,\ct[S_\ell] }.
  \end{align*}
  \end{enumerate}
  Note that the probability that this coupling changes the $b$-th block $S_b$ is equal to the probability that a uniformly random vertex $u\sim V$ has an opinion in the block $S_b$.
  Therefore, $(\c',\ct')$ is a coupling of $\Add(\c)$ and $\Add(\ct)$.
  Since each corresponding subvectors satisfy the majorization (i.e., $\c'[S_a] \succeq \ct'[S_a]$ for all $a\in[\ell]$),
  from \cref{lem:concatenation}, we have $\c'\succeq \ct'$.
\end{proof}

\end{document}